\documentclass[12pt,a4paper]{article}

\pdfoutput=1

\usepackage[margin=1.25in,bmargin=1in,tmargin=1in]{geometry}
\usepackage{setspace}
\usepackage[
  bookmarks=true,
  bookmarksnumbered=true,
  bookmarksopen=true,
  pdfborder={0 0 0},
  breaklinks=true,
  colorlinks=true,
  linkcolor=black,
  citecolor=black,
  filecolor=black,
  urlcolor=black,
]{hyperref}
\usepackage[round]{natbib}
\usepackage{bibunits}
\defaultbibliography{MarketDesignReferences}
\defaultbibliographystyle{abbrvnat}
\usepackage{amssymb,amsmath,amsthm}
\usepackage[capitalise]{cleveref}
\usepackage{nicefrac}
\usepackage{mathtools}
\usepackage{color}
\usepackage{xcolor}
\usepackage{bbm}
\usepackage{enumitem}
\usepackage{tabularx}
\usepackage{inlinelabel}
\usepackage[all=normal,paragraphs=tight]{savetrees}
\setlist{itemsep=1pt,topsep=1pt,parsep=0pt}
\usepackage[compact]{titlesec}

\usepackage{fnbreak}

\allowdisplaybreaks

\theoremstyle{plain}
\newtheorem{theorem}{Theorem}[section]
\newtheorem{lemma}[theorem]{Lemma}
\newtheorem*{example*}{Example}
\Crefname{equation}{Equation}{Equations}
\Crefname{figure}{Figure}{Figures}

\theoremstyle{definition}
\newtheorem{definition}[theorem]{Definition}
\newtheorem{example}[theorem]{Example}

\newcommand{\eqdef}{\triangleq}
\newcommand{\NN}{\mathbb{N}}
\newcommand{\ZZ}{\mathbb{Z}}
\newcommand{\RR}{\mathbb{R}}
\newcommand{\RRp}{\RR_+}
\newcommand{\RRm}{\RR^m}
\newcommand{\RRmp}{\RRp^m}
\newcommand{\RRmpne}{\RRmp\setminus\{\bar{0}\}}
\newcommand{\datapoints}{\bigl(\RRmpne\big)\times\RRmp}
\newcommand{\dataset}{D}
\newcommand{\inattentiondataset}{(D,\{P_A\}_{A\in D})}
\newcommand{\inattentiondatasetprime}{(D',\{P_A\}_{A\in D'})}

\DeclareMathOperator{\supp}{supp}

\newcommand{\restr}{|}
\newcommand{\prefs}[1]{\succ_{#1}}
\newcommand{\rprefs}[2]{{\prefs{#1}}\restr_{#2}}
\newcommand{\gt}[2]{{}_{#1}\mathtt{gt}_{#2}}
\newcommand{\attention}[2]{\mathtt{attn}_{(#1,#2)}}

\newcommand{\matched}[2]{\signed{(#1,#2)}}
\newcommand{\tmatched}[3]{\signed{(#1,#2,#3)}}
\newcommand{\plays}[2]{\mathtt{plays}_{(#1,#2)}}
\newcommand{\price}[3]{\mathtt{price}^{#1}_{(#2,#3)}}
\newcommand{\consumes}[2]{\mathtt{consumes}_{(#1,#2)}}
\newcommand{\utility}[3]{\mathtt{utility}^{#1}_{#2,#3}}
\newcommand{\cost}[3]{\mathtt{cost}^{#1}_{#2,#3}}
\newcommand{\prob}[3]{\mathtt{prob}^{#1}_{#2,#3}}
\newcommand{\quota}[2]{\mathtt{capacity}_{(#1,#2)}}
\newcommand{\signed}[1]{\mathtt{matched}_{#1}}
\newcommand{\simplemarket}{\ensuremath{(M,W,\prefs{M},\prefs{W})}}
\newcommand{\couplesmarket}{\ensuremath{(D,H,\prefs{D},\prefs{H},k)}}
\newcommand{\dynamicmarket}{\ensuremath{(M,W,\prefs{M},\prefs{W},(a_m)_{m\in M},(d_m)_{m\in M})}}

\newcommand{\formulae}[1]{\Phi_{#1}}
\newcommand{\epsfloor}[1]{\lfloor#1\rfloor_{\varepsilon_n}}
\newcommand{\epsppfloor}[1]{\lfloor#1\rfloor_{\varepsilon_{n+1}}}
\newcommand{\cnfloor}[1]{\lfloor#1\rfloor_{C_n}}
\newcommand{\cnnfloor}[1]{\lfloor#1\rfloor_{C_{n+1}}}
\newcommand{\proofstep}[1]{\medskip\noindent\underline{#1:}}

\DeclareMathOperator{\argmax}{argmax}

\let\oldFootnote\footnote
\newcommand{\nextToken}{\relax}
\renewcommand{\footnote}[1]{\oldFootnote{#1}\futurelet\nextToken\isFootnote}
\newcommand{\isFootnote}{\ifx\footnote\nextToken\textsuperscript{,}\fi}

\renewcommand{\And}{\text{ and }}
\let\citeN\citet
\let\cite\citep

\hypersetup{
  pdfauthor      = {Yannai A. Gonczarowski <yannai@gonch.name>, Scott Duke Kominers <kominers@fas.harvard.edu>, and Ran I. Shorrer <shorrer@psu.edu>},
  pdftitle       = {To Infinity and Beyond: A General Framework for Scaling Economic Theories},
}

\title{To Infinity and Beyond:\texorpdfstring{\\}{ }\Large{A General Framework for Scaling Economic Theories}\thanks{This paper subsumes an earlier paper entitled ``To Infinity and Beyond: Scaling Economic Theories via Logical Compactness,'' a one-page abstract of which appeared in the \textit{Proceedings of the 21st ACM Conference on Economics and Computation}.
We thank
David Ahn, Bob Anderson, Morgane Austern, Archishman Chakrabortyz, Chris Chambers, Yunseo Choi, Henry Cohn, Piotr Dworczak, Andrew Ellis, Tam\'as Fleiner, Drew Fudenberg, Wayne Gao, Jerry Green, Joseph Halpern, Ron Holzman, Ravi Jagadeesan, M.\ Ali Khan, David Laibson, Rida Laraki, Bar Light, Elliot Lipnowski, Ce Liu, George Mailath, Michael Mandler, Paul Milgrom, Ankur Moitra, Yoram Moses, Juan Pereyra, Marek Pycia, Debraj Ray, John Rehbeck, Phil Reny, Joseph Root, Ariel Rubinstein, Dov Samet, Chris Shannon, Tomasz Strzalecki, Sergiy Verstyuk, Rakesh Vohra, Shing-Tung Yau, Bill Zame, and numerous seminar audiences for helpful comments.  Gonczarowski was supported in part by the Adams Fellowship Program of the Israel Academy of Sciences and Humanities; his work was supported in part by ISF grants 1435/14, 317/17, and 1841/14 administered by the Israeli Academy of Sciences; by the United States--Israel Binational Science Foundation (BSF grant 2014389); and by the European Research Council (ERC) under the European Union's Horizon 2020 research and innovation programme (grant No.\ 740282), and under the European Union's Seventh Framework Programme (FP7/2007-2013) / ERC grant number~337122.
Kominers gratefully acknowledges the support of the National Science Foundation (grant SES-1459912), as well as the Ng Fund and the Mathematics in Economics Research Fund of the Harvard Center of Mathematical Sciences and Applications. Shorrer was supported by a grant from the United States--Israel Binational Science Foundation (BSF grant 2016015).}}
\author{Yannai A. Gonczarowski\thanks{Department of Economics and Department of Computer Science, Harvard University | \emph{E-mail}: \href{mailto:yannai@gonch.name}{yannai@gonch.name}. Parts of the work of Gonczarowski were carried out while at the Hebrew University of Jerusalem, at Tel Aviv University, and at Microsoft Research.}
\and
Scott Duke Kominers\thanks{Entrepreneurial Management Unit, Harvard Business School; Department of Economics and CMSA, Harvard University; and a16z crypto | \emph{E-mail}: \href{mailto:kominers@fas.harvard.edu}{kominers@fas.harvard.edu}.}
\and
Ran I. Shorrer\thanks{Department of Economics, Penn State University | \emph{E-mail}: \href{mailto:shorrer@psu.edu}{shorrer@psu.edu}.}
}

\date{April 9, 2023}

\begin{document}
\maketitle

\vspace{-1.5em}

\begin{abstract}
Many economic theory models incorporate finiteness assumptions that, while introduced for simplicity, play a real role in the analysis. We provide a principled framework for scaling results from such models by removing these finiteness assumptions. Our sufficient conditions are on the theorem statement only, and not on its proof. This results in short proofs, and even allows to use the same argument to scale similar theorems that were proven using distinctly different tools. We demonstrate the versatility of our approach via examples from both revealed-preference theory and matching theory.
\end{abstract}

\clearpage

\thispagestyle{empty}
\begin{spacing}{1.06}
\setcounter{tocdepth}{1}
\tableofcontents
\end{spacing}
\thispagestyle{empty}
\clearpage

\setcounter{page}{2}

\begin{bibunit}

{\raggedleft\small
\emph{``More is good\ldots All is better.''}

\footnotesize ---Ferengi Rule of Acquisition \#242, Star Trek

\vspace{-1em}
}

\section{Introduction}

In economic theory,
we frequently make finiteness assumptions for simplicity and/\linebreak or tractability---and those assumptions can play a real role in the analysis. Of course, the real world is itself finite, so there is in some sense no ``loss'' from assuming finiteness in our models. But finiteness assumptions nevertheless sometimes lead to  conceptual problems---if our understanding of economic theory hinges on finiteness, then our models may not quite tell the whole story. For example:
\begin{itemize}

\item In decision theory, revealed preference analysis seeks to understand what we can infer about agents from their choice behavior. While a list of observed choices is always finite, if we make parametric assumptions such as homotheticity, then each data point becomes infinitely many data points. Even without such assumptions, we would like to reason about possible demand functions---defined everywhere---that are consistent with the data, and this requires conjecturing about behavior over an infinite dataset. Furthermore, theorizing about observing infinite datasets lets us separate the limitations of inference about agents' preferences that are just imposed by data finiteness from those that are inherent even with access to every possible observation.

\item If a game-theoretic finding is true only when the set of agents is finite, then there is an implicit discontinuity, possibly relying on an edge effect or a specific starting condition that may not be robust to small frictions or perturbations (e.g., the existence of a highest-surplus match).\footnote{For an example of a different kind of discontinuity---between a finite and a continuum setting---see the work of \citet{MirallesPycia2015}, showing that a continuum model may rule out important phenomena that are observed in the finite models that converge to it.} Thus finite-market results that also hold in infinite markets are in some sense more robust.

\item Dynamic games serve to model long-run behavior and steady-states---\linebreak representing interactions that will be repeated over and over with no fixed start or end time. We thus often think of models with infinite time horizons as being better approximations of long-run play. In particular, results that are true in models with an infinite past horizon and an infinite future horizon avoid relying on assumption that agents have coordinated on or structured their behavior around ``starting'' or ``ending'' at a fixed time.\footnote{\citet*{cfw} refer to time in such a model as ``doubly infinite,'' and note such a model is conceptually useful in excluding strategies that condition on calendar time.}

\end{itemize}

\medskip

In this paper, we present a general framework for strengthening results that assume finiteness by scaling them to infinite settings.\footnote{As a side note: economic theory sometimes also turns to infinite models when their finite analogues are hard to analyze---for example, to smooth out integer effects. That is not our focus here.}  
Our approach, by relying on results in Propositional Logic, implicitly leverages topological properties of the space of \emph{theorem statements} rather than any features or techniques from their proofs. As such, it allows us to prove results along the lines of ``if a certain statement holds when assuming finiteness (regardless of how one would prove it), then---due merely to the structure of this statement---it must hold even if the finiteness assumption is dropped.''

Our methods apply flexibly across all of the settings just described, allowing us to relax finiteness assumptions on dataset size, market size, and time horizons. To illustrate the versatility of our approach and its main features, in this paper we demonstrate applications in two quite disparate fields within economic theory: decision theory (where the infinity that we tackle is the infinity of data), and matching theory (where the infinity that we tackle is that of the market size). In each of these fields, we showcase our approach using a similar arc:
\begin{itemize}
\item
\textbf{Warm-Up.} We start with applying our approach to scale a fundamental result for which the finite case is considerably simpler than the infinite case. 
For decision theory, this is that any dataset satisfying the strong axiom of revealed preferences (SARP) is rationalizable, for which
the infinite case is usually proven by appealing to Zorn's lemma.
For matching, we scale the existence of a stable matching in finite two-sided markets \citep{GaleShapley1962} to infinite markets, a result previously proven by \citet{fleiner2003fixed} using Tarski's fixed point theorem.
\item
\textbf{Reusing the warm-up for a novel application.} We then proceed with using a proof very similar to the simple proof in the warm-up to relax a finiteness assumption in a prominent result in the field. For decision theory, this is a novel infinite-data version of \citeauthor{masatlioglu2012revealed}'s (\citeyear{masatlioglu2012revealed}) characterization of limited-attention rationalizability. For matching, this is a novel infinite-market version of \citeauthor{nguyen2018couples}'s (\citeyear{nguyen2018couples}) existence result for near-feasible stable matchings with couples.
For both decision theory and matching, the \emph{proofs} of the finite cases of the results that we scale in this part and in the warm-up are starkly different. Nonetheless, their \emph{statements} are similar, and this enables us to scale them using essentially the same proof. We furthermore show that since essentially the same proof can scale similar statements, our approach can also conditionally scale not-yet-proven results.
\item
\textbf{Overcoming inherent limitations of finite logical statements.} At first glance, our approach might seem limited to proving results that are discrete in nature (see discussion below). Nonetheless, we proceed to show how to use it to prove results regarding objects that are nondiscrete (coming from a continuum space). For decision theory, here we reprove \citeauthor{reny2015}'s (\citeyear{reny2015}) infinite-data version of \citeauthor{afriat67}'s (\citeyear{afriat67}) theorem, where utilities come from a continuum space, as well as \citeauthor*{caplin2017rationally}'s (\citeyear{caplin2017rationally}) infinite-data version of \citeauthor{caplin2015revealed}'s (\citeyear{caplin2015revealed}) characterization of having a costly information acquisition representation, where priors come from a continuum space; again, our scaling proofs for both theorems are nearly identical. For matching, here we (newly) prove an infinite-market version of the \cite{hatfield2013stability} result on the existence of Walrasian equilibria (and hence stable outcomes) in trading networks.
\item
\textbf{A deep dive: combining with domain-specific knowledge.} Finally, we conclude with an example geared toward experts in the field, showing how to further leverage our framework by combining it with domain-specific techniques. For decision theory, we
(newly) prove an infinite-data version of \citeauthor{MR1971}'s (\citeyear{MR1971}) characterization of rationalizable stochastic datasets, strengthening previous infinite variants of their result.
For matching, we
re-prove the existence of a man-optimal stable matching and (newly) prove the strategy-proofness of the man-optimal stable mechanism in infinite markets, resolving an open problem posed by \citet{jagadeesan2017lone2}.
\end{itemize}
Additional results can be found in the appendix. These include results for other models (such as existence of Nash equilibria in infinite markets) as well as for an additional notion of infinity (that of infinite time\footnote{Specifically, we use our approach to dispense with the assumption of a ``period $0$'' to obtain existence results in a doubly infinite model, which is perhaps more appropriate as a representation of a steady-state \cite*[see, e.g.,][]{ory,cfw}.}).

As mentioned above, some of the results that we prove in this paper
are novel to this work. Other results that we (re-)prove have already been obtained using other, very different methods, which allows us to compare and contrast the previous proof techniques with ours. As our illustrative applications demonstrate, proofs that use our framework have several notable features. First, they use \textbf{one tool} rather than having to choose from various setting-specific tools. Second, the proof structure is  \textbf{modular}: our conditions for scaling the finite-case result to the infinite case depend only on the statement of the finite-case result and are completely \textbf{agnostic} to the argument/methods used to prove that result. Furthermore, the proofs are \textbf{robust} in that even their dependence on the details of the model is quite weak, and essentially the same proof can sometimes be used in quite different models.

\subsection{Technique}\label{intro-technique}
Our general approach requires that problems have what we call a \textit{well description}, and that this well description satisfies what we call the \textit{finite-subset property}. We define these concepts precisely in \cref{sec:core-lemma}. Here, we provide an informal description and two illustrative applications (which we later formalize): showing that infinite datasets that satisfy SARP are rationalizable and showing the existence of a stable matching in infinite markets.

A \textit{well description} specifies for each problem a (potentially infinite) set of individually finite logical statements over Boolean variables, such that the problem has a solution if and only if there is an assignment of truth values to these Boolean variables under which all these statements hold simultaneously.\footnote{The reader may think of an assignment of truth values to the Boolean variables as a ``state of the world.'' In the language of Mathematical Logic, such an assignment under which all statements hold is called a \emph{model for the statements} (perhaps confusingly within an economics setting).}
For example, in a revealed-preferences setting we can encode a rationalizing preference order using a set of Boolean variables $\{\gt{a}{b}\}$, each being $\mathsf{True}$ if $a\succ b$ for the corresponding $a,b$. We can then express all the required properties of a rationalizing order (transitivity, antisymmetry, and consistency with whichever outcomes are revealed preferred to others) using logical statements phrased in terms of these variables (infinitely many such statements, but each statement individually finite). We further elaborate on this example in \cref{revealed-prefs-warmup}. Similarly, in a matching setting we can encode a (one-to-one) stable matching in terms of a set of Boolean variables $\{\matched{m}{w}\}$, each being $\mathsf{True}$ if the corresponding man $m$ and woman $w$ are matched. We can then express all the required properties of a stable matching (one-to-oneness, individual rationality, and no blocking pairs) using logical statements phrased in terms of these variables. We further elaborate on this example in \cref{matching-warmup}. In either setting, an assignment of truth values to the variables under which all statements hold simultaneously corresponds exactly to a solution (a rationalizing order or a stable matching), and vice versa. In particular, for each problem, our set of logical statements has such an assignment if and only if the problem has a solution, and therefore this is a well description.

Each of the preceding well descriptions captures finite and infinite problems equally: the only difference that arises is in the cardinalities of the sets of Boolean variables and logical statements.
When the problem is infinitary (infinite dataset or infinitely many agents), the associated set of logical statements is infinite as well. Yet,
 each of the logical statements we construct is nonetheless individually finite, that is, it contains only finitely many of the Boolean variables.

 Fix a well description, and call the set of logical statements associated with each problem ``the description of the problem.''
 The well description satisfies the \textit{finite-subset property} if every finite subset of the description of any problem belongs to the description of a problem that has a solution. In either our revealed-preferences or matching example, we identify such a problem that has a solution by restricting the original problem to the elements (data points or agents) ``mentioned'' in the given finite subset. The given finite subset
 is indeed part of the well description of the restricted problem, and since this problem is finite, it can be solved by known existence results for finite problems, so long as we verify that it ``inherits'' from the infinite problem any properties required by these results (namely, for the revealed-preferences example, satisfying SARP). 
\cref{scaling-lemma} therefore guarantees  the existence of an appropriate solution of the infinite problem.

We prove \cref{scaling-lemma} using \textit{Logical Compactness} (see \cref{propositional}), a central result in the theory of Propositional Logic. While the above examples demonstrate the applicability of our approach to existence results of inherently discrete objects, we also show how to use this approach to scale economic results that go beyond discrete solutions and even beyond existence results into infinite settings.

\subsection{Paper Outline: Choose Your Own \emph{Scaling} Adventure}

The remainder of this paper is organized as follows. \Cref{preliminaries} introduces preliminaries from Propositional Logic and states and proves our main technical lemma (\cref{scaling-lemma}). \Cref{revealed-prefs,matching} then apply our framework to revealed preferences and matching 
settings, respectively, showcasing its wide applicability. \Cref{related-literature} reviews related literature and \Cref{discussion} concludes. Omitted proofs and further applications and examples (including for scaling results to settings with infinite time horizons) are presented in the appendix.

\paragraph{A note regarding \Cref{revealed-prefs,matching}.} Each of these \lcnamecrefs{revealed-prefs} provides a complete introduction to our framework by following the above-described ``arc'' in the \lcnamecref{revealed-prefs}'s respective setting: starting with a warm-up to reprove a fundamental result (\Cref{revealed-prefs-warmup} / \Cref{matching-warmup}); then reusing the simple proof from the warm-up, with small changes, to newly prove an infinite version of a prominent result (\Cref{filter} / \Cref{couples}) and to conditionally scale not-yet-proven results; then showing how to use our framework to prove  results that are nondiscrete in nature (\Cref{garp} / \Cref{walrasian}); and finally concluding with an example geared at experts in the field, showing how to combine our framework with domain-specific knowledge to overcome challenges in newly scaling an additional result (\Cref{stoch}, where the main challenge is to prove that the set of formulae that we use really is a well description / \Cref{sp}, where the main challenge is to prove that the finite-subset property is satisfied.).
Accordingly, in the spirit of ``Choose Your Own Adventure'' books,\footnote{See \url{https://en.wikipedia.org/wiki/Choose_Your_Own_Adventure}.} readers may choose to read either or both of these two \lcnamecrefs{revealed-prefs} depending on whether they wish to familiarize themselves with our framework through the lens of decision theory or through the lens of matching. Readers interested in a deeper view of the similarities and contrasts when applying our framework to two disparate fields are encouraged to read both \lcnamecrefs{revealed-prefs}, and may do so in any order.

\section{Framework}\label{preliminaries}

In this \lcnamecref{preliminaries},  
 we provide a brief introduction to Propositional Logic (in \cref{propositional}),\footnote{For a more in-depth look at Propositional Logic primitives and at the Compactness Theorem, see a textbook on Mathematical Logic
\citep{enderton,gonczarowskinisan}.}  and use it to state and prove our main technical lemma (in \cref{sec:core-lemma}).
\subsection{Propositional Logic Preliminaries}\label{propositional}

In Propositional Logic, we work with a set of Boolean variables, and study the truth values of statements---called formulae---made up of those variables. We construct formulae by conjoining variables with simple logical operators such as \emph{or}, \emph{not}, and \emph{implies}. Variables are abstract, and do not have meaning on their own---but we can imbue them with ``semantic" meaning by introducing formulae that reflect the structure of economic (or other) problems. Once given semantic meaning, the truth or falsity of statements in our Propositional Logic model imply the corresponding results in the associated economic model.

We start by formalizing the idea of \emph{(well-formed propositional) formulae}. To define the set of formulae at our disposal, we first define a basic (finite or infinite) set of (Boolean) \emph{variables}. In each \lcnamecref{matching-warmup} of this paper we define a different set of variables built around the economic setting that we model in that \lcnamecref{matching-warmup}.

Once we have defined a (finite or infinite) set $V$ of variables, we can define the set of all well-formed formulae inductively:
\begin{itemize}
\item
`$\phi$' is a well-formed formula for every variable $\phi \in V$.
\item
`$\lnot\phi$' is a well-formed formula for every well-formed formula $\phi$.
\item
`$(\phi\vee\psi)$', `$(\phi\wedge\psi)$', `$(\phi\rightarrow\psi)$', and `$(\phi\leftrightarrow\psi)$' are well-formed formulae for every two well-formed formulae $\phi$ and $\psi$.
\end{itemize}

\renewcommand{\aa}{\mathtt{P}}
\newcommand{\bb}{\mathtt{Q}}
\newcommand{\cc}{\mathtt{R}}
\newcommand{\dd}{\mathtt{S}}
\newcommand{\aap}{\aa}
\newcommand{\bbp}{\bb}
\newcommand{\ccp}{\cc}
\newcommand{\ddp}{\dd}

\begin{example*}
We could start with a set of four variables $V=\{\aap,\bbp,\ccp,\ddp\}$.  Then, each of the following is a well-formed formula:

\noindent\begin{tabularx}{\textwidth}{@{}l@{}X@{}l@{}X@{}l@{}X@{}l@{}}
$`\aap\text{\emph{'}}\inlinelabel{ab-set}$
& &
$`(\aap \vee \bbp)\text{\emph{'}}\inlinelabel{ab-either}$
& &
$`\lnot(\aap\wedge\bbp)\text{\emph{'}}\inlinelabel{ab-notboth}$
& &
$`((\aap\wedge\ccp)\rightarrow\ddp)\text{\emph{'}}\inlinelabel{abc-trans}$
\end{tabularx}
\end{example*}

We sometimes abuse notation by omitting parentheses and writing, e.g., `$\phi\vee\psi\vee\xi$' when any arbitrary placement of parentheses in the formula (e.g., `$((\phi\vee\psi)\vee\xi)$' or `$(\phi\vee(\psi\vee\xi))$') will make do for our analysis. We sometimes abuse notation even further by writing, e.g., `$\bigvee_{i=1}^{10}\phi_i$' to mean `$\phi_1\vee\phi_2\vee\cdots\vee\phi_{10}$' (once again, only when the precise placement of omitted parentheses is of no consequence to our analysis).

\newcommand\ctr{n}
We note that while well-formed formulae can be arbitrarily long, each well-formed formula is always finite in length. Thus, for example, a disjunction `$\phi_1\vee\phi_2\vee\cdots$' of infinitely many formulae is \emph{not} a well-formed formula. We therefore take special care when we claim that formulae of the form `$\bigvee_{\phi \in \Psi} \phi$' are well-formed, as this is true only if $\Psi$ is finite.

A \emph{model} is a mapping from the set $V$ of all variables to Boolean values, i.e., each variable is mapped either to being $\mathsf{True}$ or to being $\mathsf{False}$. This induces a truth value for every formula `$\phi$' where $\phi\in V$.
A model also induces a \emph{truth value} for every other formula, defined inductively as follows: 
\begin{itemize}
	\item `$\lnot\phi$' is $\mathsf{True}$ iff $\phi$ is $\mathsf{False}$;
	\item `$(\phi\vee\psi)$' is $\mathsf{True}$ iff either or both of $\phi$ and $\psi$ is $\mathsf{True}$;
	\item `$(\phi\wedge\psi)$' is $\mathsf{True}$ iff both $\phi$ and $\psi$ are $\mathsf{True}$;
	\item `$(\phi\rightarrow\psi)$' is $\mathsf{True}$ iff either $\phi$ is $\mathsf{False}$ or $\psi$ is $\mathsf{True}$ or both (that is, `$(\phi\rightarrow\psi)$' is $\mathsf{False}$ iff both $\phi$ is $\mathsf{True}$ and $\psi$ is $\mathsf{False}$); and
	\item`$(\phi\leftrightarrow\psi)$' is $\mathsf{True}$ iff $\phi$ and $\psi$ are either both $\mathsf{True}$ or both $\mathsf{False}$.
\end{itemize}
\begin{example*}
\setcounter{equation}{0}
Given the concept of truth values, we can reinterpret the formulae \eqref{ab-set}--\eqref{abc-trans} as follows:
\begin{align}
`\aap\text{\emph{'}}&&&\text{``\,$\aap$ \emph{[}is $\mathsf{True}$\emph{]}\!''},\\
`(\aap \vee \bbp){\emph{'}}&&&\text{``\,$\aap$ or $\bbp$ \emph{[}is $\mathsf{True}$\emph{]}\!''},\\
`\lnot(\aap\wedge\bbp){\emph{'}}&&&\text{``not \emph{(}$\aap$ and $\bbp$ \emph{[}are both $\mathsf{True}$\emph{])}\!''},\\
`((\aap\wedge\ccp)\rightarrow\ddp){\emph{'}}&&&\text{``\,$\aap$ and $\ccp$ \emph{[}both being $\mathsf{True}$\emph{]}, implies $\ddp$ \emph{[}being $\mathsf{True}$\emph{]}\!''}.
\end{align}
The formula in \eqref{ab-either} is $\mathsf{True}$ in a model if and only if either \emph{`}$\aap$\emph{'} or \emph{`}$\bbp$\emph{'} (or both) are $\mathsf{True}$ in that model; the formula in \eqref{ab-notboth} is $\mathsf{True}$ in a model unless both \emph{`}$\aap$\emph{'} and \emph{`}$\bbp$\emph{'}  are $\mathsf{True}$ in that model; and the formula in \eqref{abc-trans} is $\mathsf{True}$ in a model unless both \emph{`}$\aap$\emph{'} and \emph{`}$\ccp$\emph{'} are $\mathsf{True}$ in that model while \emph{`}$\ddp$\emph{'} is $\mathsf{False}$ in that model. 
\end{example*}

We say that a formula is \emph{satisfied} by a model if it is $\mathsf{True}$ under that model. For example, each of the formulae \eqref{ab-set}, \eqref{ab-either}, and \eqref{abc-trans} is satisfied by the model that assigns value $\mathsf{True}$ to all variables, however the formula \eqref{ab-notboth} is not satisfied by it. We say that a (possibly infinite) set of formulae is \emph{satisfied} by a model if every formula in the set is satisfied by the model.  For example, the set of the formulae \eqref{ab-set}--\eqref{abc-trans} is satisfied by the model that assigns value $\mathsf{True}$ to all variables except $\bbp$.
We say that a (possibly infinite) set of formulae is \emph{satisfiable}, or that it \emph{has a model}, if it is satisfied by some model. For example, the set containing $`\aap\text{'}$ and $`\lnot\aap\text{'}$ is not satisfiable.

Clearly, if a (finite or infinite) set of formulae $\Phi$ is satisfiable, then every subset of $\Phi$ is also satisfiable (by the same model), and in particular every \emph{finite} subset of $\Phi$ is satisfiable; the \emph{Compactness Theorem for Propositional Logic} gives a surprising and nontrivial converse to this statement.

\begin{theorem}[The Compactness Theorem for Propositional Logic \citep{godel1930vollstandigkeit,malcev1936untersuchungen}]\label{thm:compact}
A set of formulae $\Phi$ is satisfiable if (and only if) every \textbf{finite} subset $\Phi'\subseteq\Phi$ is satisfiable.
\end{theorem}

\subsection{A Scaling Lemma for Economic Theories}\label{sec:core-lemma}

In this section, we use Propositional Logic to derive a sufficient condition for scalability of an economic theorem to infinite cases. This condition, formalized in \cref{scaling-lemma}, is at the heart of all of our proofs.

Let $\mathcal{S}$ be a set to which we refer as a set of (potential) \emph{solutions}.
Let $\mathcal{P}$ be a set to which we refer as a set of \emph{(economic) problems} whose solutions (if such exist) are in $\mathcal{S}$.
For example, for the consumer choice rationalization example from \cref{intro-technique}, $\mathcal{S}$ is the set of all consumer preferences over some set of objects, and a problem $P\in\mathcal{P}$ is to rationalize a specific dataset~$\dataset$. For the stable matching example from \cref{intro-technique}, $\mathcal{S}$ is the set of all matchings between two given sets, and a problem is to find a stable matching in a specific two-sided market $M$.

Define $I:\mathcal{P}\times\mathcal{S}\rightarrow\{\mathsf{True},\mathsf{False}\}$ such that $I(P,S)$ is $\mathsf{True}$ if and only if $S$ is a solution for $P$. Note that given a problem $P$, even if we can easily describe for any given $S$ whether $S$ is a solution of $P$ (i.e., whether $I(P,S)$ is $\mathsf{True}$), it may not be clear just from examining $P$ (and $I$) whether or not it has \emph{any} solution (i.e., whether there exists $S\in\mathcal{S}$ such that $I(P,S)$ is $\mathsf{True}$). For example, while it is easy to describe when given preferences rationalize a given dataset, it is not immediate from examining a dataset whether there exist preferences that rationalize it. Similarly, while it is easy to describe when a given matching is stable in a given market, it is not immediate from examining a market whether it admits a stable matching.
Regardless of the economic setting, given a problem, our goal will be to ascertain whether a solution for it indeed exists.

We say that the set $\mathcal{P}$ of economic problems is a set of \emph{well-describable economic problems} if for every $P\in\mathcal{P}$ there exists a set $\formulae{P}$ of well-formed formulae such that $P$ has a solution if and only if $\formulae{P}$ has a model. We call a collection $(\formulae{P})_{P\in\mathcal{P}}$ of such sets a \emph{well description} of $\mathcal{P}$. (Constructing such a set was discussed in \cref{intro-technique}.\footnote{In both cases demosntrated in \cref{intro-technique}, an even stronger property holds: the set of solutions of $P$ is in \emph{one-to-one} correspondence with the set of models of $\formulae{P}$. While such a one-to-one correspondence holds in many of our applications, this is not required for our arguments.})
Given a well description of $\mathcal{P}$,
we say that a problem $P\in \mathcal{P}$ satisfies the \emph{finite-subset property} (with respect to the given well description of~$\mathcal{P}$) if for every finite subset $\Phi'\subset\formulae{P}$ there exists an economic problem $P'\in\mathcal{P}$ that has a solution and for which $\Phi'\subseteq \formulae{P'}$.
By \cref{thm:compact}, we then have:

\begin{lemma}[Scaling Lemma for Economic Theories]\label{scaling-lemma}
Let $\mathcal{P}$ be a set of well-describable economic problems and let $(\formulae{P})_{P\in\mathcal{P}}$ be a well description of $\mathcal{P}$.
Let $P\in\mathcal{P}$.
If $P$ satisfies the finite-subset property, then $P$ has a solution.
\end{lemma}

\begin{proof}
Let $P\in\mathcal{P}$ be a problem satisfying the finite-subset property.
By well describability, it is enough to show that $\formulae{P}$ is satisfiable. By \cref{thm:compact}, it is therefore enough to show that every finite $\Phi'\subset\formulae{P}$ is satisfiable. Let $\Phi'$ be such a finite subset. Since $P$ satisfies the finite-subset property, there exists $P'\in\mathcal{P}$ that has a solution such that $\Phi'\subseteq\formulae{P'}$. Since $P'$ has a solution, by our well-describability assumption we have that $\formulae{P'}$ is satisfiable by some model. Since $\Phi'\subseteq\formulae{P'}$, the same model also satisfies~$\Phi'$, and so $\Phi'$ is satisfiable as required.
\end{proof}

As demonstrated in \cref{intro-technique}, in many cases of interest the existence of a solution for an appropriate $P'$ for any $\Phi'$ can be established by finite-case theorems (for example, on rationalizability of finite datasets or stable matching in finite markets). Thus by \cref{scaling-lemma} we obtain existence of solutions for the infinite case of such problems as well. In this paper, we demonstrate the wide applicability of \cref{scaling-lemma} to a wide range of economic problems.

\addtocontents{toc}{\protect\setcounter{tocdepth}{2}}

\section{Revealed Preferences}\label{revealed-prefs}

In this section, we apply our framework to revealed-preferences analysis. First, in \cref{revealed-prefs-warmup} we provide a concise proof of a well-known result: that the strong axiom of revealed preferences (SARP) is necessary and sufficient for strict rationalization of infinite choice data. In \cref{filter} we
utilize the same argument from \cref{revealed-prefs-warmup}, with minimal changes, to  prove a novel existence result for a more complex setting in which the finite case has been analyzed using completely different tools: we scale the necessity and sufficiency of WARP-limited attention (WARP-LA) for limited-attention rationalizability  \citep{masatlioglu2012revealed} to encompass infinite datasets, and conditionally scale a class of related not-yet-proven results. In \cref{garp} we show how to handle non-discrete objects: we scale \citeauthor{afriat67}'s theorem on the existence of rationalizing utility function to infintie datasets; in \cref{inattention}, we use essentially the same proof to scale \citeauthor{caplin2015revealed}'s (\citeyear{caplin2015revealed}) theorem on the existence of rationalizing attention cost to infinite datasets.
 Finally, in \cref{stoch} we scale \citeauthor{MR1971}'s (\citeyear{MR1971}) characterization of rationalizable stochastic datasets to infinite datases. In this proof, we use specialized tools to prove that the set of formulae that we use really is a well description. 

\subsection{Warm Up: Rational Choice Functions}\label{revealed-prefs-warmup}
We begin with a classic  revealed preferences setup. 
Let $X$ be a (possibly infinite) set of goods. The set of \textit{menus} includes all finite subsets of $X$. A \textit{dataset} $\dataset\subseteq\bigl\{(S,a)\in 2^X\times X ~\big|~ a\in S\bigr\}$ consists of the (unique) respective choices made by an agent in a (possibly infinite) subset of menus. We say that a dataset $\dataset$ is \textit{rationalized} by a strict preference relation  $\succ$ (complete, antisymmetric, and transitive) over $X$, if for every $(S,a)\in\dataset$, the agent's choice $a$ is the maximal element from $S$ according to $\succ$. A dataset is \textit{rationalizable} if it is rationalized by some strict preference relation over $X$.

Given a dataset and a pair of goods, $x,y\in X$, we say that $x$ is \textit{revealed preferred} to $y$ ($x\succ^Ry$) if $x$ is chosen from a menu that includes $y$. A dataset satisfies the \textit{strong axiom of revealed preferences (SARP)} if $\succ^R$ is acyclic (i.e., there does not exist $k>1$ and $x_1,\dots,x_k\in X$ such that $x_i\succ^R x_{i+1}$ for all $i$, and in addition $x_k\succ^R x_1$).  
Clearly, satisfying SARP is a necessary condition for a dataset to be rationalizable.  A classic result due to \citet{richter1966} and \citet{hansson1968choice} is that satisfying SARP is also sufficient. \cref{FiniteRichter} states this result for the special case of a finite dataset.

\begin{theorem}\label{FiniteRichter}
A finite dataset is rationalizable if and only if it satisfies SARP.
\end{theorem}

The general version of \cref{FiniteRichter}, due to \citet{richter1966} and \citet{hansson1968choice}, builds on \citeauthor{szpilrajn1930extension}' Extension Theorem, a fundamental result in revealed preferences analysis whose variants are used to prove many key results in the theory of revealed preferences \citep{ChambersEchenique2009}. \citet*[p.\ 17]{ok2007real} explains that ``[a]lthough it is possible to prove this [fundamental result of order theory] by mathematical induction when $X$ is finite,\footnote{I.e., the setting corresponding to \cref{FiniteRichter}; for an explicit proof see, e.g., \citet{lahiri}.} the proof in the general case is built on a relatively advanced method[\ldots].'' 
Indeed, the standard way\footnote{\citet{mandler2020quick} provides an alternative simple proof of \citeauthor{richter1966} and \citeauthor{hansson1968choice}'s result.} to prove an infinite-dataset version of \cref{FiniteRichter} (equivalently, to prove \citeauthor{szpilrajn1930extension}'s Extension Theorem) as well as to prove variant results, is to use Zorn's Lemma (e.g., \citealp{richter1966}; \citealp{DUGGAN}; \citealp*{mwg}, Proposition 3.J.1; \citealp{ce}, Theorems~1.4 and~1.5).

As a warm-up, we use \cref{scaling-lemma} to scale \cref{FiniteRichter} to also apply to infinite datasets (implicitly reproving \citeauthor{szpilrajn1930extension}'s Extension Theorem). 
While \cref{scaling-lemma} relies on Logical Compactness, which, like Zorn's Lemma, relies on some variant of the Axiom of Choice,\footnote{Strictly speaking, under ZF set theory, Logical Compactness is weaker than Choice.} we suspect that the proof we present here may complement the standard approach. In particular, our argument
 may in some ways be more accessible to students than the traditional proof because it avoids the ``overhead'' of understanding the full statement of Zorn's Lemma.\footnote{\citet*{mwg}, similarly to \citeauthor{ok2007real}, label their proof (which uses Zorn's Lemma) as ``advanced.''}

\begin{theorem}[\citealp{richter1966,hansson1968choice}]\label{szpilrajn}
A (possibly infinite) dataset is rationalizable if and only if it satisfies SARP.
\end{theorem}

\begin{proof}[Proof of \cref{szpilrajn}]
As noted, the ``only if'' direction is trivial, so we prove the ``if'' direction. We do so using \cref{scaling-lemma}.

\proofstep{Definition of $\mathcal{P}$}
Fixing $X$, let $\mathcal{P}$ be the set of all pairs $(X',\dataset)$ such that $X'\subseteq X$ and $\dataset$ is a dataset with menus over goods in $X'$ that satisfies SARP. A \emph{solution} for a pair $(X',\dataset)\in\mathcal{P}$ is a strict preference order over $X'$ that rationalizes $\dataset$.

\proofstep{Well describability}
We define a variable $\gt{a}{b}$ for every pair of distinct $a,b\in X$.
In what follows, for each $(X',\dataset)\in\mathcal{P}$ we define a set $\formulae{(X',\dataset)}$ of formulae over these variables so that models (over the variables that appear in $\formulae{(X',\dataset)}$) of $\formulae{(X',\dataset)}$
are in one-to-one correspondence with 
the 
(not-yet-proven-to-be-nonempty) set of solutions for $(X',\dataset)$. The correspondence is obtained by endowing the variable $\gt{a}{b}$ with the semantic interpretation ``$a$ is preferred to~$b$.'' That is, it maps a model for $\formulae{(X',\dataset)}$ to the preference~$\succ$ such that for every distinct $a,b\in X'$, we have that $a\succ b$ if and only if the variable $\gt{a}{b}$ is $\mathsf{True}$ in that model.
We define the set $\formulae{(X',\dataset)}$ to consist of the following formulae:
\begin{enumerate}
\item
for all distinct $(S,a)\in D$ and all $b\in S\setminus\{a\}$, the formula `$\gt{a}{b}$',
requiring that the preferences (that correspond to any model of the formulae) rationalize $\dataset$;
\item
for all distinct $a,b\in X'$,  the formula `$\gt{a}{b}\vee\gt{b}{a}$',
requiring that the preferences be complete;
\item
for all distinct $a,b\in X'$,  the formula `$\lnot(\gt{a}{b}\wedge\gt{b}{a})$',
requiring that the preferences be antisymmetric;
\item
for all distinct $a,b,c\in X'$, the formula `$(\gt{a}{b}\wedge\gt{b}{c})\rightarrow\gt{a}{c}$',
requiring that the preferences be transitive.
\end{enumerate}
By construction, $(\formulae{(X',\dataset)})_{(X',\dataset)\in\mathcal{P}}$ is a well description of~$\mathcal{P}$.
(Astute readers will notice that the preceding formulae correspond exactly with the formulae \eqref{ab-set}--\eqref{abc-trans} from \cref{propositional} upon taking $\aap=\gt{a}{b}$, $\bbp=\gt{b}{a}$, $\ccp=\gt{b}{c}$, and $\ddp=\gt{a}{c}$.)

\proofstep{Finite-subset property}
Let $(\bar{X},\dataset)\in\mathcal{P}$.
Let
$\Phi'\subset\formulae{(\bar{X},\dataset)}$ be a finite subset.
Since $\Phi'$ is finite, it ``mentions'' (through variables used) only finitely many elements of~$\bar{X}$; denote the set of these elements by $X'\subset\bar{X}$.
Let $\dataset'\eqdef\bigl\{(S\cap X',a)~\big|~ (S,a)\in\dataset \And a\in X'\bigr\}$. By definition, $\Phi'\subseteq\formulae{(X',\dataset')}$. Furthermore, $\dataset'$ satisfies SARP since had a cycle been induced by $\dataset'$, it would have also been induced by $\dataset$, however $\dataset$ satisfies SARP, so it induces no cycles. Hence, by \cref{FiniteRichter}, $\dataset'$~is rationalizable by some strict preference order over $X'$. Therefore, $(\bar{X},\dataset)$~satisfies the finite-subset property. Thus, by \cref{scaling-lemma}, $\dataset$~is rationalizable by a strict preference order over $\bar{X}$.
\end{proof}

\subsection{Applicability of the Same Proof to Other Settings: Limited Attention}\label{filter}

A notable strength of our approach is that it is agnostic to the methods used to prove the finite result being scaled. Therefore, the same proof can be used to scale similar statements even when the finite-case proofs of these statements hinge on very different tools. As an illustration, 
we use essentially the same proof as in \cref{revealed-prefs-warmup} to scale to infinite datasets the WARP-limited attention (WARP-LA) result of \citet{masatlioglu2012revealed}, despite their proof using approaches that are quite different from those used to prove the result of \cref{FiniteRichter}.

The setup resembles that of \cref{revealed-prefs-warmup}: let $X$ be a (possibly infinite) set of goods. A \textit{menu} is a finite nonempty subset of $X$. A dataset is \textit{full} if it consists of the (unique) respective choices made by an agent in \emph{all possible} (infinitely many, if $X$ is infinite) menus. A \emph{filter} is a function $\Gamma$ that maps each menu $S$ to a menu $\Gamma(S)\subseteq S$. A filter is an \emph{attention filter} if $\Gamma\bigl(S\setminus\{x\}\bigr)=\Gamma(S)$ for every menu $S$ and $x\in S\setminus\Gamma(S)$.
A dataset is \textit{limited-attention rationalizable} if there exist a strict preference relation $\succ$ and an attention filter $\Gamma$ such that for every menu $S$ in the dataset, the agent's choice is the most-preferred element in $\Gamma(S)$ according to $\succ$. 

Analogously to \cref{FiniteRichter}, \citet{masatlioglu2012revealed} uncover a condition, WARP-LA,\footnote{\label{fn: WARP-LA} A full dataset satisfies WARP-LA if, for any menu $S$, there exists $x^*\in S$ such that for any menu $T$ that includes $x^*$, if $(T,x)\in D$ for some $x\in S$ and $\left(T\setminus\{x^*\},x\right)\notin D$ then $(T,x^*)\in D$.} that is necessary and sufficient for limited-attention rationalizability:

\begin{theorem}[\citealp{masatlioglu2012revealed}]\label{FiniteLA}
A full finite dataset is limited-attention rationalizable if and only if it satisfies WARP-LA.
\end{theorem}

As is the case for SARP, 
if a certain dataset satisfies WARP-LA, then so does any sub-dataset. This suffices for us to scale \cref{FiniteLA} to infinite datasets using the same approach we used in \cref{revealed-prefs-warmup}. 

\begin{theorem}\label{InfiniteLA}
A full (possibly infinite) dataset is limited-attention rationalizable if and only if it satisfies WARP-LA.
\end{theorem}

The well description that we build to prove \cref{InfiniteLA} is conceptually similar to the one from our proof of \cref{szpilrajn}, but requires slight modification due to the addition of the attention filter. The idea is to have, as before, for every pair $a,b\in X$ a variable $\gt{a}{b}$ that will be $\mathsf{True}$ in a model if and only if $a$ precedes $b$ in the preference relation corresponding to the model. But furthermore, for every menu $S$ and  $\emptyset\ne T\subseteq S$, we introduce a variable $\attention{S}{T}$ that will be $\mathsf{True}$ in a model if and only if $\Gamma(S)=T$ for the attention filter corresponding to the model, and upon whose value the formulae that represent the dataset observations will be conditioned. So, if $a$ is chosen from $S$ in the dataset, for every $b\in S\setminus\{a\}$ instead of having a single formula $\gt{a}{b}$ mandating that $a$ precede $b$ in the preference relation, we have for each $b\in T\subseteq S$ a formula that says ``if $\Gamma(S)=T$ then $a$ precedes $b$ in the preference relation,'' i.e., `$\attention{S}{T}\rightarrow\gt{a}{b}$'. Additionally, we introduce formulae that say that for each menu $S$, there exists precisely one $T$ such that $\attention{S}{T}$ holds. Finally, for each menu $S$, each $T\subset S$ and $x\in S\setminus T$ we introduce the formula `$\attention{S}{T}\rightarrow\attention{S\setminus\{x\}}{T}$', requiring that the filter be an attention filter. Except for these additions, the proof runs along the same lines as that of \cref{szpilrajn}; we relegate the details to \cref{app:warpla}.

\subsubsection{Conditional Scaling}

Limited-attention rationalizability, as defined above, flexibly accommodates a wide array of attention filters.  But, in some cases, one may wish to impose additional structure (e.g., requiring that the agent always pays attention to at least two options), or to consider filters that fall outside the domain of attention filters \citep[][discuss a wide array of examples from the literature]{masatlioglu2012revealed}. 
In many cases of interest, the constraint on the filters to be considered takes the form  $\forall S_1 \forall S_2 \cdots\forall S_n: H\bigl(S_1,\Gamma(S_1),S_2,\Gamma(S_2),\ldots,S_n,\Gamma(S_n)\bigr)=\mathsf{True}$ for some predicate $H$ that takes $2n$ menus, where all $S_i$ are menus.\footnote{For example, requiring that the agent always pays attention to at least two options can be expressed with $n=1$, setting $H(S_1,T_1)$ to be $\mathsf{True}$ unless $|S_1|>1$ yet $|T_1|=1$. As another example, the filter being an attention filter could have been expressed in this form with $n=2$, setting $H(S_1,T_1,S_2,T_2)$ to be $\mathsf{True}$ unless $S_1=S_2\cup\{x\}$, $x\notin T_1$, and yet $T_1\ne T_2$.} Such restrictions can be encoded into our well description (in the same fashion that we encoded the restriction that $\Gamma$ is an attention filter in our proof of \cref{InfiniteLA}).\footnote{The only change to the well description would be to add one more (finite) formula type, for every $S_1,\ldots,S_n$: `$\bigvee_{T_1\subseteq S_1,\ldots,T_n\subseteq S_n:H(S_1,T_1,\ldots,S_n,T_n)}\bigwedge_{i=1}^n\attention{S_i}{T_i}$'.} 

This provides an opportunity to point out that our framework is agnostic not only to \emph{how} the finite-case theorem being scaled was proved (as already discussed), but furthermore, to \emph{whether} it has even been proved. Indeed, even absent a proof for the finite-case theorem, our framework can yield conditional statements. Consider a result  (potentially, one that could be uncovered in the future) that, for some predicate $H$, determines that a finite dataset is $H$-rationalizable (i.e., rationalizable using a filter that meets the above requirement w.r.t.\ $H$) if and only if it satisfied some given condition WARP-H defined on finite datasets.
Extend the definition of WARP-H to infinite datasets by defining that an infinite dataset satisfies WARP-H if and only if every finite sub-dataset of it satisfies WARP-H.\footnote{If for finite datasets, WARP-H is a ``no cycles of some form'' condition (like SARP, WARP, or \mbox{WARP-LA}), then extending it to infinite datasets results in the same condition for infinite datasets as well: no (finite) cycles of that form.} A proof completely analogous to our proof of \cref{InfiniteLA} then also proves it imediately applies to infinite datasets:

\begin{theorem}\label{WARP-H}
If it holds that a finite dataset is $H$-ratiolalizable iff it satisfies \mbox{WARP-H}, then it also holds that an infinite dataset is $H$-ratiolalizable iff it satisfies WARP-H.
\end{theorem}

\subsection{Handling Nondiscrete Solution Concepts:\texorpdfstring{\\}{ }Rationalizing Consumer Demand}\label{garp}

We now move to rationalizing consumption behavior in the presence of prices. For the most part, in this \lcnamecref{garp} we  follow the notation of \citet{reny2015}. Fix a number of goods $m\in\NN$. A dataset $\dataset\subset\datapoints$ with generic element $(\bar{p},\bar{x})\in\dataset$ represents a set of observations, where in each, a consumer with a budget faces a price vector $\bar{p}\ne \bar{0}$ and chooses to consume the bundle $\bar{x}$. A utility function $u:\RRmp\rightarrow\RR$ \emph{rationalizes} the dataset $\dataset$ if for every $(\bar{p},\bar{x})\in\dataset$ and every $\bar{y}\in\RRmp$, it holds that if 
$\bar{p}\cdot\bar{y}\le\bar{p}\cdot\bar{x}$ (i.e., $\bar{y}$ can also be bought with the budget) then $u(\bar{y})\le u(\bar{x})$, and if $\bar{p}\cdot\bar{y}<\bar{p}\cdot\bar{x}$ (i.e., $\bar{y}$ can be bought without spending the entire budget) then $u(\bar{y})<u(\bar{x})$.\footnote{This assumption rules out trivial rationalizations such as constant utility functions. See \citet{ce} for a more detailed discussion.} If only the former implication holds for every such $(\bar{p},\bar{x})$ and $\bar{y}$, then we say that $u$ \emph{weakly rationalizes} $\dataset$. 

A dataset $\dataset$ satisfies the \emph{Generalized Axiom of Revealed Preference (GARP)} if for every (finite) sequence $(\bar{p}_1,\bar{x}_1),\ldots,(\bar{p}_k,\bar{x}_k)\in\dataset$, if for every $i\in\{1,2,\ldots,k\!-\!1\}$ it holds that $\bar{p_i}\cdot\bar{x}_{i+1}\le\bar{p_i}\cdot\bar{x}_i$, then $\bar{p}_k\cdot\bar{x}_1\ge\bar{p}_k\cdot\bar{x}_k$. It is straightforward from the definitions that satisfying GARP is a precondition for rationalizability (indeed, otherwise we would have for any rationalizing utility function $u$ that $u(\bar{x}_1)\ge u(\bar{x}_2)\ge\cdots\ge u(\bar{x}_k)>u(\bar{x}_1)$). In a celebrated result, \citet{afriat67} showed  that GARP is also a sufficient condition for rationalizability of a finite dataset---and furthermore  GARP is a sufficient condition for rationalizability of such a dataset by a utility function with many properties that are often assumed in simple economic models. This finding implies that the standard economic model of consumer choice has no testable implications beyond GARP. 

\begin{theorem}[\citealp{afriat67}]\label{garp-finite}
A finite dataset $\dataset\subseteq\datapoints$ satisfies GARP if and only if it is rationalizable. Moreover, when GARP holds there exists a utility function that rationalizes $\dataset$ that is continuous, concave, nondecreasing, and strictly increasing when all coordinates strictly increase.
\end{theorem}

There are well-known examples of infinite datasets that are generated by quasiconcave utility functions but may not be rationalized by a concave utility function \cite[see][]{aumann1975values,reny2013simple}. \citet{kannai2004individual} and \citet{apartsin2006demand} provide necessary conditions, stronger than GARP,  for rationalizability by a concave function. Recently, \citet{reny2015} unified the literature and clarified the boundaries of \citeauthor{afriat67}'s theorem by showing that GARP is indeed necessary and sufficient for rationalization of even infinite datasets---and in fact, GARP also guarantees rationalizability by a utility function with many desired properties (yet not all the properties that are attainable in the finite case).

\begin{theorem}[\citealp{reny2015}]\label{garp-infinite}
A (possibly infinite) dataset $\dataset\subseteq\datapoints$ satisfies GARP if and only if it is rationalizable. Moreover, when GARP holds there exists a utility function that rationalizes $\dataset$ that is quasiconcave, nondecreasing, and strictly increasing when all coordinates strictly increase.
\end{theorem}

\citet{reny2015} provides examples showing that continuity and concavity (the properties of the rationalizing utility function from \cref{garp-finite} that are absent from \cref{garp-infinite}) cannot be guaranteed to be attainable for any rationalizable dataset. \citet{reny2015} then proves \cref{garp-infinite} using a novel construction that---unlike \citeauthor{afriat67}'s construction---applies also to infinite data sets. We instead give a concise alternative proof\footnote{While we prove the result of \citet{reny2015} in its full generality, it is worth noting that \citeauthor{reny2015}'s proof does not use the Axiom of Choice, while ours does to some extent. More specifically, the Compactness Theorem, which we use for proving \cref{scaling-lemma}, is equivalent (under~ZF) to the Boolean Prime Ideal (BPI) Theorem (equivalently, to the Ultrafilter Lemma), which is known to be a ``weaker form of the Axiom of Choice'' in the sense that ZF+BPI is strictly weaker than ZFC but strictly stronger than ZF \cite[see, e.g.,][Theorems~6.7 and~8.16]{Halbeisen}.} of \cref{garp-infinite} by scaling \cref{garp-finite} as a black box using \cref{scaling-lemma}. One of the  challenges in our argument is that utility functions have an (uncountably) infinite range, so it is not \emph{a priori} obvious how to encode such a function by a model defined via individually finite formulae (e.g., how to require that each bundle is associated with some real number that represents the utility from it); to overcome this challenge, our well description encodes the utility from each bundle as the limit of a sequence of discrete utilities.
This approach, 
 in turn, introduces additional challenges, such as how to make sure, using only constraints on these discrete functions, that the limit utilities satisfy all desired properties. This is particularly challenging with properties that are not preserved by limits, such as being strictly increasing.

\begin{proof}[Proof of \cref{garp-infinite}]
As noted, the ``only if'' direction is trivial, so we prove the ``if'' direction. We do so using \cref{scaling-lemma}.

\proofstep{Definition of $\mathcal{P}$}
Fixing $m$, let $\mathcal{P}$ be the set of all datasets $\dataset\subseteq\datapoints$ satisfying GARP. A \emph{solution} for a dataset~$\dataset\in\mathcal{P}$ is a utility function that rationalizes $\dataset$ that is quasiconcave, nondecreasing, and strictly increasing when all coordinates strictly increase.

\proofstep{Well describability}
We set $\varepsilon_n\eqdef2^{-n}$ for every $n\in\NN$.
We define a variable $\utility{n}{\bar{x}}{v}$ for every $n\in\NN$, every $\bar{x}\in\RRmp$, and every $v\in V_n\eqdef\{0,\varepsilon_n,2\cdot\varepsilon_n,\ldots,1\}$.
In what follows, for each $\dataset\in\mathcal{P}$ we define a set $\formulae{\dataset}$ of formulae over these variables so that models of $\formulae{\dataset}$ are in one-to-one correspondence with the (not-yet-proven-to-be-nonempty) set of solutions for $\dataset$ that satisfy certain properties (we then have to show that the existence of any solution implies the existence of a solution with such properties).
The correspondence is obtained by endowing the variable $\utility{n}{\bar{x}}{v}$ with the semantic interpretation ``$\epsfloor{u(\bar{x})}=v$ for the corresponding utility function $u$,'' where for every $n\in\NN$ and every $x$ we denote by $\epsfloor{x}\eqdef2^{-n}\cdot\lfloor2^n\cdot x\rfloor$ the rounding-down of $x$ to the nearest multiple of $\varepsilon_n$.
Fixing an enumeration $(\bar{q}^k_1,\bar{q}^k_2)_{k=1}^\infty$ of the countable set $\bigl\{(\bar{q}_1,\bar{q}_2)\in\mathbb{Q}^m\times\mathbb{Q}^m\mid \bar{q}_1\ll \bar{q}_2\bigr\}$,\footnote{For $\bar{x},\bar{y}\in\RRm$, we write $\bar{x}\ll\bar{y}$ to denote that $x_i<y_i$ for every $i=1,\ldots,m$.}
we define the set $\formulae{\dataset}$ to consist of the following formulae:
\begin{enumerate}
\item
for all $n\in\NN$ and all $\bar{x}\in\RRmp$, the (finite!) formula
`$\bigvee_{v\in V_n}\utility{n}{\bar{x}}{v}$',
requiring that $\bar{x}$ have a rounded-down-to-$\varepsilon_n$ utility in $[0,1]$;
\item
for all $n\in\NN$, all $\bar{x}\in\RRmp$, and all distinct $v,w\in V_n$, the formula
`$\utility{n}{\bar{x}}{v}\rightarrow\lnot\utility{n}{\bar{x}}{w}$',
requiring that the rounded-down-to-$\varepsilon_n$ utility from $\bar{x}$ be unique;
\item
for all $n\in\NN$, all $\bar{x}\in\RRmp$, and all $v\in V_n$, the formula
`$\utility{n}{\bar{x}}{v}\rightarrow\bigl(\utility{n+1}{\bar{x}}{v}\vee\utility{n+1}{\bar{x}}{v+\varepsilon_{n+1}}\bigr)$',
requiring that $\epsfloor{u(\bar{x})}=\epsfloor{\epsppfloor{u(\bar{x})}}$;
\item
\begin{sloppypar}
for all $n\in\NN$, all $\bar{x},\bar{y}\in\RRmp$, all convex combinations $\bar{z}\in\RRmp$ of $\bar{x},\bar{y}$, and all $v,w\in V_n$, the (finite) formula
`$\bigl(\utility{n}{\bar{x}}{v}\wedge\utility{n}{\bar{y}}{w}\bigr)\rightarrow
\bigvee_{
v'\in V_n:
v'\ge\min\{v,w\}
}\utility{n}{\bar{z}}{v'}$',
requiring that the rounded-down-to-$\varepsilon_n$ utility function be quasiconcave;
\end{sloppypar}
\item
for all $n\in\NN$, all $\bar{x},\bar{y}\in\RRmp$ s.t.\ $\bar{x}\le\bar{y}$, and all $v\in V_n$, the (finite) formula
`$\utility{n}{\bar{x}}{v}\rightarrow
\bigvee_{
w\in V_n:
w\ge v
}\utility{n}{\bar{y}}{w}$',
requiring that the rounded-down-to-$\varepsilon_n$ utility function be nondecreasing;
\item
\begin{sloppypar}
for all $k\in\NN$ and all $n>k$, the (finite) formula
`$\utility{n}{\bar{q}^k_1}{v}\rightarrow
\bigvee_{
w\in V_n:
w\ge v+2^{-k-1
}}\utility{n}{\bar{q}^k_2}{w}$',
requiring that starting at some $n$, the rounded-down-to-$\varepsilon_n$ utility from $\bar{q}^k_2$ be greater by at least $2^{-k-1}$ than the rounded-down-to-$\varepsilon_n$ utility from $\bar{q}^k_1$;
\end{sloppypar}
\item
for all $n\in\NN$, all datapoints $(\bar{p},\bar{x})\in\dataset$, all $\bar{y}\in\RRmp$ s.t.\ $\bar{p}\cdot\bar{y}\le\bar{p}\cdot\bar{x}$, and all $v\in V_n$, the (finite) formula
`$\utility{n}{\bar{x}}{v}\rightarrow
\bigvee_{
w\in V_n:
w\le v
}\utility{n}{\bar{y}}{w}$',
requiring that the rounded-down-to-$\varepsilon_n$ utility weakly rationalize $\dataset$.
\end{enumerate}
We now argue that $(\formulae{\dataset})_{\dataset\in\mathcal{P}}$ is a well description of~$\mathcal{P}$. Let $\dataset\in\mathcal{P}$.

We first claim that every model that satisfies $\formulae{\dataset}$ corresponds to a solution for $\dataset$. Fix a model for $\formulae{\dataset}$. For every $\bar{x}\in\RRmp$ and every $n\in\NN$, let $v_n\in V_n$ be the value such that $\utility{n}{\bar{x}}{v_n}$ is $\mathsf{True}$ in the model (well defined by the first and second formula-types above), and define $u(\bar{x})=\lim_{n\rightarrow\infty} v_n$ (well defined, e.g., by the third formula-type above since $v_n$ is a Cauchy sequence). The resulting utility function $u$ is a limit of nondecreasing quasiconcave functions (by the fourth and fifth formula-types above) that weakly rationalize the data (by the seventh formula-type above). Hence, $u$ itself is a nondecreasing quasiconcave function that weakly rationalizes the data. Furthermore, for every $\bar{x},\bar{y}\in\RRmp$ s.t.\ $\bar{x}\ll\bar{y}$, there exist two rational number vectors ``in between'' them, i.e., there exists $k\in\NN$ s.t.\ $\bar{x}\ll\bar{q}^k_1\ll\bar{q}^k_2\ll\bar{y}$. Therefore, we have that $u(\bar{x})\le u(\bar{q}^k_1)\le u(\bar{q}^k_2)-2^{-k-1}<u(\bar{q}^k_2)\le u(\bar{y})$ (the second inequality stems from this inequality holding for almost all functions of which $u$ is the limit, by the sixth formula-type above), so $u$ is strictly increasing when all coordinates strictly increase. Finally, since $u$ weakly rationalizes~$\dataset$ and is also strictly increasing when all coordinates strictly increase, then $u$ also rationalizes~$\dataset$.

Second, we claim that if $\dataset$ has a solution, then $\formulae{\dataset}$ has a model. Fix a solution $u$ for $\dataset$, and let $\bar{u}(\bar{x})\eqdef1/4+(1/2\pi)\cdot\arctan(u(\bar{x}))+\sum_{k:\bar{q}^k_2\le \bar{x}}2^{-k-1}$ for every $\bar{x}\in\RRmp$. As this transformation of utilities is strictly monotone, the resulting function~$\bar{u}$ still rationalizes the data, and is quasiconcave, nondecreasing, and strictly increasing when all coordinates strictly increase. Furthermore, the sum of the first two summands is in $[0,\nicefrac{1}{2}]$, and so is the third summand, so the overall sum is in $[0,1]$. Finally, due to the third summand, $\bar{u}(\bar{q}^k_2)>\bar{u}(\bar{q}^k_1)+2^{-k-1}$ for every $k\in\NN$. Using $\bar{u}$ we can therefore construct a model for $\formulae{\dataset}$ (by setting each $\utility{n}{\bar{x}}{v}$ to be $\mathsf{True}$ iff $v=\epsfloor{\bar{u}(x)}$), and so $\formulae{\dataset}$ has a model. To sum up, $(\formulae{\dataset})_{\dataset\in\mathcal{P}}$ is a well description of~$\mathcal{P}$.

\proofstep{Finite-subset property}
Let $\dataset\in\mathcal{P}$. Let $\Phi'\subset\formulae{\dataset}$ be a finite subset.
Since $\Phi'$ is finite, there are only finitely many formulae of the above
seventh type (the only formula type that depends on the dataset) in $\Phi'$.
Let $\dataset'\subset\dataset$ be the set of datapoints that induce these formulae.
By definition, $\Phi'\subseteq\formulae{\dataset'}$. Furthermore, $\dataset'$ satisfies GARP since any sub-dataset of $\dataset$ satisfies GARP, and hence, by \cref{garp-finite}, $\dataset'$~is rationalizable. Therefore, $\dataset$~satisfies the finite-subset property. Thus, by \cref{scaling-lemma}, $\dataset$~is rationalizable.
\end{proof}

A natural question is why the same argument cannot be used to scale \cref{garp-finite} while maintaining concavity rather than quasiconcavity. The short answer is that---due to the requirement that $u$ be strictly monotone, and the inherent need to make each formula finite---our proof of \cref{garp-infinite} relies heavily on the fact that quasiconcavity, unlike concavity, is maintained under weakly monotone transformations (such as the mapping of $u$ to $\bar{u}$); we discuss this further in \cref{sec:afriat-limitat}.

\subsubsection{Additional Application of the Same Proof: Rational Inattention}

Our proof of \cref{garp-infinite} is quite a bit more flexible than one might imagine.
In \cref{inattention}, we use essentially the same well description to scale, from finite to infinite datasets, the seminal result of \citet{caplin2015revealed} in quite a different rationalization domain: that a state-dependent stochastic choice dataset has a costly information acquisition representation if and only if it satisfies 
\emph{No Improving Action Switches (NIAS)} and \emph{No Improving Attention Cycles (NIAC)}.
\citet{caplin2017rationally} recently proved the infinite version of this result via a novel proof that diverges from \citeauthor{caplin2015revealed}'s proof of the finite case.\footnote{\citet{denti2017rationally} provide a similar result for infinite datasets of a different kind.}
We reprove this result using essentially the same well description as in our proof of \cref{garp-infinite}, despite the differences between the two settings considered, and despite the fact that neither the original proofs of the finite versions nor the original proofs of the infinite versions of any of these quite different theorems share any common core technique. The main difference between the two well descriptions is that this application does not require strict monotonicity. Therefore, the sixth formula-type of the above well description is not required, and the proof that a solution implies a model is simpler as it does not require carefully ``massaging'' the function $u$ into $\bar{u}$ as above.

\subsection{Combining with Domain-Specific Knowledge:\texorpdfstring{\\}{ }Rationalizing Stochastic Demand}\label{stoch}

Fix a set $X$ of alternatives. A \emph{(stochastic choice) dataset} is a function \[P:\bigl\{(A,x)\in (2^X\setminus\{\emptyset\})\times X ~\big|~ A\in\mathcal{A} \And x\in A\bigr\}\rightarrow[0,1]\] such that $\mathcal{A}\subseteq\bigl\{A\in 2^X\setminus\{\emptyset\} ~\big|~ |A|<\infty\bigr\}$ and $\sum_{x\in A}P(A,x)=1$ for every $A\in\mathcal{A}$. The dataset~$P$ is interpreted as probabilities with which different alternatives are chosen given various menus $\mathcal{A}$ from $X$. Probabilistic choice may emerge from random shocks to preferences over time, or represent fractions of deterministic choices in a population.\footnote{For a detailed discussion and a textbook treatment of this setting, see \citet{ce}.}

A dataset is \emph{rationalizable} if there exists a probability measure $\nu$ over the space\footnote{In this \lcnamecref{stoch}, when we consider probability measures over the space of total orders (i.e., complete, transitive, and antisymmetric orders) over a set $X$, then if $X$ is countable then we take this space as a measurable space w.r.t.\ the discrete $\sigma$-algebra, and more generally for arbitrary $X$ we take this space as a measurable space w.r.t.\ the $\sigma$-algebra generated by all of its subset of the form $\{\pi \mid a_1\succ_{\pi}\cdots\succ_{\pi} a_n\}$, where the $a_i$ are distinct elements in $X$.} of total orders over $X$ such that $\Pr_{\nu}[x\succ A\setminus\{x\}]=P(A,x)$ for every $(A,x)$ in the domain of $P$.
A dataset~$P$ satisfies the \emph{Axiom of Revealed Stochastic Preference} if for every $n$ and every finite sequence (possibly with repetitions) $(A_i,x_i)_{i=1}^n$, where each $(A_i,x_i)$ is in the domain of $P$, 
\[
\sum_{i=1}^n P(A_i,x_i)\le \max_{\pi\in (\cup_{i=1}^n A_i)!}\left\{\sum_{i=1}^n\mathbf{1}_{[x_i\succ_\pi A_i\setminus\{x_i\}]}\right\},
\]
where the factorial symbolizes the set of all possible permutations.

\begin{theorem}[\citealt{MR1971,MR1990}]\label{stoch-finite}
Let $X$ be a finite set of items. A dataset $P$ is rationalizable if and only if it satisfies the Axiom of Revealed Stochastic Preference.
\end{theorem}

We use \cref{scaling-lemma} to scale \cref{stoch-finite} to infinite datasets.

\begin{theorem}\label{stoch-infinite}
Let $X$ be a (possibly infinite) set of items. A dataset $P$ is rationalizable if and only if it satisfies the Axiom of Revealed Stochastic Preference.
\end{theorem}

This infinite setting and its economic importance have been discussed by \citet{cohen1980} and \citet{mcfadden2005review}. \citet{cohen1980} showed that the celebrated representation result of \citet{falmagne1978representation} using Block--Marschak polynomials \citep{block1959random} scales to infinite sets $X$ if the definition of rationalizability is weakened; \citet{cohen1980} also gave several stronger structural conditions on~$X$ that are sufficient for (``un-weakened,'' i.e., as defined above) rationalizability. \citet{mcfadden2005review} showed how to scale \cref{stoch-finite} to a different infinite setting, once again by either weakening the definition of rationalizability or by demanding the existence of a certain topological structure on $X$ (to obtain ``un-weakened'' rationalizability). \Cref{stoch-infinite}, which we now prove, does \emph{not} weaken the definition of rationalizability and does \emph{not} impose any assumptions on~$X$.

\begin{proof}[Proof of \cref{stoch-infinite}]
The ``only if'' direction is trivial (as in the finite case), so we prove the ``if'' direction. We do so using \cref{scaling-lemma}.

\proofstep{Definition of $\mathcal{P}$}
Fixing $X$, let $\mathcal{P}$ be the set of all datasets satisfying the Axiom of Revealed Stochastic Preference. A \emph{solution} for a dataset~$P\in\mathcal{P}$ is a probability measure $\mu$ on the space of total orders over elements that appear in menus in $P$ that rationalizes $P$.

\proofstep{Well describability}
In our proof, we use the following \lcnamecref{marginals} to help us encode via individually finite formulae a probability measure $\mu$ over the total orders of $X$. The \lcnamecref{marginals}, whose proof we spell out in \cref{app:marginals}, follows directly from the Kolmogorov Extension Theorem.

\begin{lemma}\label{marginals}
Let $X$ be a (possibly infinite) set, and for every $m\in\NN$ and sequence $\bar{a}=(a_i)_{i=1}^m $ of $m$ distinct elements of $X$, let $p_{\bar{a}}\in[0,1]$. Then the following are equivalent:
\begin{itemize}
\item
There exists a probability measure $\mu$ over the space of total orders over $X$ such that for every $m\in\NN$ and sequence $\bar{a}=(a_i)_{i=1}^m$  of $m$ distinct elements of~$X$, it is the case that $p_{\bar{a}}=\Pr_{\mu}[a_1\succ a_2\succ\cdots\succ a_m].$
\item
$p_{(a)}=1$ for every $a\in X$ (sequence of length $1$); and for every $m\in\NN$ and sequence $(a_1,\ldots,a_m,a)$ of $m+1$ distinct elements of $X$, it is the case that $p_{(a_1,\ldots,a_m)}=\sum_{i=0}^m p_{(a_1,\ldots,a_i,a,a_{i+1},\ldots,a_m)}$.
\end{itemize}
\end{lemma}

We set $\varepsilon_n=2^{-n}$ for every $n\in\NN$.
We define a variable $\prob{n}{\bar{a}}{p}$ for every $n\in\NN$, every $m\in\NN$, every $m$-tuple of distinct items $\bar{a}=(a_1,\ldots,a_m)$ from $X$, and every $p\in V_n\eqdef\{0,\varepsilon_n,2\cdot\varepsilon_n,\ldots,1\}$.
In what follows, for each $P\in\mathcal{P}$ we define a set $\formulae{P}$ of formulae over these variables so that models of $\formulae{P}$ are in one-to-one correspondence with the (not-yet-proven-to-be-nonempty) set of solutions for $P$.
The correspondence is obtained by endowing the variable $\prob{n}{\bar{a}}{p}$ with the semantic interpretation ``$p=\epsfloor{\Pr_{\mu}[a_1\succ a_2\succ\cdots\succ a_m]}$ for the corresponding probability measure $\mu$,'' where for every $n\in\NN$ and every $x$ we denote by $\epsfloor{x}=2^{-n}\cdot\lfloor2^n\cdot x\rfloor$ the rounding-down of $x$ to the nearest multiple of $\varepsilon_n$.
Letting $X'\subseteq X$ be the set of elements that appear in $P$, we define the set $\formulae{P}$ to consist of the following formulae:
\begin{enumerate}
\item
for all $n\in\NN$ and all (finite) tuples $\bar{a}$ of distinct items from $X'$, the (finite!) formula
`$\bigvee_{v\in V_n}\prob{n}{\bar{a}}{p}$',
requiring that $\bar{a}$ have a rounded-down-to-$\varepsilon_n$ probability;
\item
for all $n\in\NN$, all tuples $\bar{a}$ of distinct items from $X'$, and all distinct $p,q\in V_n$, the formula
`$\prob{n}{\bar{a}}{p}\rightarrow\lnot\prob{n}{\bar{a}}{q}$',
requiring that the rounded-down-to-$\varepsilon_n$ probability of $\bar{a}$ be unique;
\item
for all $n\in\NN$, all tuples $\bar{a}$ of distinct items from $X'$, and all $p\in V_n$, the formula
`$\prob{n}{\bar{a}}{p}\rightarrow\bigl(\prob{n+1}{\bar{a}}{p}\vee\prob{n+1}{\bar{a}}{p+\varepsilon_{n+1}}\bigr)$',
requiring that $\epsfloor{\Pr_{\mu}[a_1\succ a_2\succ\cdots\succ a_m]}=\epsfloor{\epsppfloor{\Pr_{\mu}[a_1\succ a_2\succ\cdots\succ a_m]}}$;
\item
for all $n\in\NN$ and all $a\in X'$, the (finite) formula
`$\prob{n}{(a)}{1}$',
requiring that the rounded-down-to-$\varepsilon_n$ probability of the ordering $a$ is $1$;
\item
for all $n\in\NN$, all $m\in\NN$, and all $(m\!+\!1)$-tuple $(a_1,\ldots,a_m,a)$ of distinct items from $X'$, the (finite) formula
\[
\smashoperator[r]{\bigvee_{\substack{p,p_0,\ldots,p_m\in V_n \\ \text{s.t. } \sum_{i=1}^{m+1}p_i\in[p-(m+1)\cdot\varepsilon_n,p]}}}~~~~~~\left(\prob{n}{(a_1,\ldots,a_m)}{p}\wedge\bigwedge_{i=0}^m\prob{n}{(a_1,\ldots,a_i,a,a_{i+1},\ldots,a_m)}{p_i}\right),
\]
requiring that---up to rounding errors---the second condition of \cref{marginals} hold for every ``rounding-down of $\mu$'';
\item
for all $n\in\NN$ and all $(A,x)$ in the domain of $P$, the (finite) formula
\[
\smashoperator[r]{\bigvee_{\substack{p_1,\ldots,p_{(|A|-1)!}\in V_n \\ \text{s.t. } \sum_{i=1}^{(|A|-1)!}p_i\in\\ [P(A,x)-(|A|-1)!\cdot\varepsilon_n\,,\,P(A,x)]}}}~~~~~~~~\bigwedge_{i=1}^{(|A|-1)!}\prob{n}{(x,\underbrace{{\scriptstyle a_1,\ldots,a_{|A|-1}}}_{\clap{$\substack{\text{$i$th permutation}\\\text{of $A\setminus\{x\}$}}$}})}{p_i},
\]
requiring that---up to rounding errors---every ``rounding-down of $\mu$'' rationalize $P$.
\end{enumerate}
We now argue that $(\formulae{P})_{P\in\mathcal{P}}$ is a well description of~$\mathcal{P}$. Let $P\in\mathcal{P}$ and let $X'\subseteq X$ be the set of elements that appear in $P$.

We first claim that every model that satisfies $\formulae{P}$ corresponds to a solution for~$P$. Fix a model for $\formulae{P}$. For every $m\in\NN$, every $m$-tuple $\bar{a}$ of distinct items from $X'$, and every $n\in\NN$, let $p_n\in V_n$ be the probability such that $\prob{n}{\bar{a}}{p_n}$ is $\mathsf{True}$ in the model (well defined by the first and second formula-types above), and define $p_{\bar{a}}=\lim_{n\rightarrow\infty} p_n$ (well defined, e.g., by the third formula-type above since $\prob{n}{\bar{a}}{p_n}$ is a Cauchy sequence). The resulting probabilities $p_{\bar{a}}$ satisfy the second condition of \cref{marginals} (by the fourth and fifth formula-types above), and hence there exists a probability measure $\mu$ over the space of total orders over $X'$ that induces these probabilities. Since $\mu$ induces these probabilities, then by the sixth formula-type above, $\mu$ rationalizes $P$.

Second, if $P$ has a solution $\mu$, then using it we can construct a model for $\formulae{P}$ (by setting each $\prob{n}{\bar{a}}{p}$ to be $\mathsf{True}$ iff $p=\epsfloor{\Pr_{\mu}[a_1\succ a_2\succ\cdots\succ a_m]}$), and so $\formulae{P}$ has a model. To sum up, $(\formulae{P})_{P\in\mathcal{P}}$ is a well description of~$\mathcal{P}$.

\proofstep{Finite-subset property}
Let $P\in\mathcal{P}$. Let $\Phi'\subset\formulae{P}$ be a finite subset.
Since $\Phi'$ is finite, it ``mentions'' (through variables used) only finitely many elements of~$X$; denote the set of these elements by $X'\subset X$.
Let $P'\eqdef\bigl\{(A,x)\in P~\big|~ A\subseteq X'\bigr\}$. By definition, $\Phi'\subseteq\formulae{P'}$. Furthermore, $P'$ satisfies the Axiom of Revealed Stochastic Preference since any sub-dataset of $P$ satisfies this axiom, and hence, by \cref{stoch-finite}, $P'$~is rationalizable.
Therefore, $P$~satisfies the finite-subset property. Thus, by \cref{scaling-lemma}, $P$~is rationalizable.
\end{proof}

\subsection{Further Remarks and Limitations}\label{sec:afriat-limitat}

In this section, we demonstrated the versatility of our approach across several revealed-preference settings. Our approach can be used to scale many additional finite-data results to encompass infinite datasets. For example, it may allow to scale results such as those of \citet{yusufcan2020random} on the existence of random attention representation, or \citet{filiz2020progressive} on progressive random choice.\footnote{We thank Yusufcan Masatlioglu for proposing these applications.}
In addition to scaling a wide array of finite-data results to encompass infinite datasets, our approach can also be used to adapt finite-data rationalization results to support parametric restrictions, as in \citet{hu2021theory}, since such restrictions often translate into infinitely many constraints. Our approach, however, is not without limitations. We conclude this section with some remarks on limitations of proofs presented throughout the section. A more high-level discussion of settings in which our approach is not applicable is provided in \cref{non-applications}.

Afriat's theorem (\cref{garp-finite}) guarantees that finite demand datasets satisfying GARP can be rationalized using a concave utility function. But, there are well-known examples of quasiconcave utility functions whose full (infinite) demand dataset (which satisfies GARP since it is derived from the choices of a utility function) cannot be rationalized using a concave utility function. In \cref{garp}, we reproved the main result of \citet{reny2015} that unified these settings: any demand dataset, finite or infinite, that satisfies GARP can be rationalized using a quasiconcave utility function (\cref{garp-infinite}).

By \cref{scaling-lemma}, the existence of a counterexample, together with the correctness of Afriat's theorem for finite datasets, implies that the existence of a concave rationalizing utility function (as guaranteed by Afriat's theorem) has no well description that satisfies the finite-subset property. This might seem puzzling since a simple modification to the fourth formula type in our proof (which imposes quasiconcavity) can be used to impose concavity (as in our similar scaling proof in \cref{inattention}), and so should seemingly result in a well description as required. The answer to this puzzle is that this well description does not, in fact, satisfy the finite-subset property. Specifically, the sixth formula type in our proof makes a stronger monotonicity requirement than the monotonicity that is guaranteed by Afriat's theorem, and therefore Afriat's theorem cannot be used to show that the finite-subset property required by \cref{scaling-lemma} holds.

To address this issue, a natural approach would be to change the monotonicity requirement that we use to require only strict monotonicity, as guaranteed by Afriat's theorem. But, it is not possible to well-describe strict monotonicity with our variables (since strict inequalities are not preserved in the limit). Our way around this limitation was to make a stronger requirement that is well-describable. But in order to use Afriat's theorem to show that the finite-subset property holds, we had to relax the concavity requirement (recall that our proof applied monotonic transformations to the utility function; while these transformations do not preserve concavity, they do preserve quasiconcavity). We note that while this may appear to be an artefact of using \cref{scaling-lemma}, the existence of the abovementioned counterexample guarantees that no other approach could circumvent this issue. The tradeoff between strengthening monotonicity and weakening concavity, so that well describability and the finite-subset property are satisfied, sheds some new light on what breaks in the infinite case, which at first glance might look like an issue with concavity, but at a deeper look reveals itself as an issue with strict monotonicity. This affords some degree of intuition for ``why'' the concavity assumption in \cref{garp-finite} must be relaxed to quasiconcavity when scaling it to infinite datasets.

We note that when we use a very similar proof in \cref{inattention} to scale a result by \citet{caplin2015revealed}, we do require concavity rather than merely quasiconcavity. This is possible because in that result only weak (rather than strict) monotonicity (in information) is required, which can be well described without being strengthened. Contrasting these two proofs provides yet another example of the power of our approach to very tangibly pinpoint why certain conditions can be maintained when some theorems are scaled but not when others are.

\cref{garp-infinite} illustrates some of the limitations of our framework. A Propositional Logic formulation precludes the use of quantifiers (e.g., ``\emph{there exists} a positive gap by which the utility from $\bar{q}_2$ is greater than the utility from $\bar{q}_1$'') as well as precludes infinitely long formulae such as infinite disjunctions (e.g., ``the utility from $\bar{q}_2$ is greater than the utility from $\bar{q}_1$ by at least \emph{one of the following infinitely many} positive gaps''). This prohibits the well description of certain properties of interest (e.g., strict monotonicity, unless strengthened) without the use of variables that refer to infinitely many objects. But such variables oftentimes hinder the ability to invoke finite theorems to show that the finite-subset property holds.  

The requirement of strict monotonicity underlies another well-known counterexample. While a strict preference order over a finite set of objects can always be represented by a utility function, the same need not be true when the set of objects is uncountable.\footnote{For example, lexicographic preferences over $\mathbb{R}^2$ or any strict preference order over $2^\mathbb{R}$.} Accordingly, any attempt to use our strengthened monotonicity requirement to scale the finite case is of course bound to fail when the set of objects is uncountable. It is instructive to consider how it would fail. Recall that our strengthened monotonicity requires fixed positive gaps between various utility values. In the case of \citeauthor{afriat67}'s theorem, requiring countably many such gaps sufficed, and hence the required gap lengths could be chosen so that their sum is finite. By contrast, scaling the existence of a utility representation for strict preferences to uncountable sets would involve requiring uncountably many positive gaps. This means that the sum of lengths of required gaps would be infinite, and so some objects would not be associated with a finite utility. 

\section{Matching Theory}\label{matching}
In this section, we apply our framework to analyze matching markets. First, in \cref{matching-warmup} we scale \citeauthor{GaleShapley1962}'s classic result on the existence of a stable matching in finite marriage markets to cover the infinite case. 
While one can prove this result using other methods \cite[see][]{fleiner2003fixed,jagadeesan2017lone2},
 in  \cref{couples}  we demonstrate how to utilize the same argument from \cref{matching-warmup}, with minimal changes, to prove a novel existence result for a more complex setting in which the finite case has been analyzed using completely different tools, which render the proof techniques of \citet{fleiner2003fixed} and \citet{jagadeesan2017lone2} inapplicable: we scale \citeauthor{nguyen2018couples}'s (\citeyear{nguyen2018couples}) near-feasibility result for stable matching with couples to encompass infinite markets, and conditionally scale a conjectured strengthened version of this result.  
In \cref{walrasian} we show how to handle non-discrete objects: we scale the \cite{hatfield2013stability} result on the existence of Walrasian equilibria in trading networks to cover infinite networks. Finally, in \cref{sp} we scale the strategy-proofness of the man-optimal stable matching mechanism in two-sided matching to infinite markets,\footnote{In our infinite markets agents maintain their mass, which is particularly economically appealing as strategic issues remain in full force. In continuum-limit market models, by contrast, strategic issues often disappear as agents become measure-zero \cite[see, e.g.,][]{AzevedoBudish2018}.} resolving a standing open question;
in \cref{menopt} we use essentially the same proof to reprove the existence of the man-optimal stable matching in infinite two-sided matching markets.\footnote{Subsequent to our work, \citet{choi20} has shown how to use our proof technique to scale other structural results with similar statements.} In these proofs, we use specialized tools to prove that the finite-subset property holds.

\subsection{Warm Up: Stable Marriage}\label{matching-warmup}

We begin with the simplest possible matching market setting: a one-to-one two-sided ``marriage'' matching market. Such a market is represented by a quadruplet \simplemarket, where $M$ is a (possibly infinite\footnote{While (as we describe soon) we must require that each agent finds at most countably many agents acceptable, we make no assumptions on the cardinality of the set of agents.}) set of men, $W$ is a (possibly infinite) set of women, and $\prefs{M}$ is a profile of preference lists for the men over the women consisting, for each man $m\in M$, of a linearly ordered preference list of women that either is finite, or specifies man $m$'s $n$th-choice woman for every $n\in\NN$. Any woman on $m$'s list is considered to be \emph{acceptable} to $m$, i.e., preferred by $m$ over being unmatched, while any woman not on $m$'s list is considered unacceptable to $m$. Similarly, $\prefs{W}$ is a profile of preference lists for the women over the men. A (one-to-one, not necessarily perfect) \emph{matching} between $M$ and $W$ is a pairwise-disjoint set of man-woman pairs. A \emph{blocking pair} with respect to a matching $\mu$ is a man-woman pair $(m,w)$ such that $m$ prefers $w$ to his partner in $\mu$ (or, if he is unmatched in $\mu$, prefers $w$ to being unmatched) and $w$ prefers $m$ to her partner in $\mu$ (or, if she is unmatched in $\mu$, prefers $m$ to being unmatched). 

A matching $\mu$ is called \emph{stable} if (1) under $\mu$, no participant is matched to a partner he or she finds unacceptable (individual rationality), and, (2) there are no blocking pairs with respect to $\mu$.
A classic result of \citeN{GaleShapley1962} shows that   stable matchings exist for any \textsl{finite} matching market in the setting just described.

\begin{theorem}[\citealp{GaleShapley1962}]\label{stable-finite}
In any finite, one-to-one two-sided matching market, a stable matching exists.
\end{theorem}

As a warm-up, we use \cref{scaling-lemma} to scale \cref{stable-finite} to infinite markets. Previous studies have established the existence of stable matchings in infinite, one-to-one, two-sided matching markets via a fixed-point argument, or---for a special case---via an infinite variant of \citeauthor{GaleShapley1962}'s algorithm (see \citeN{fleiner2003fixed} and \citeN{jagadeesan2017lone2}, respectively).

\begin{theorem}[\citealp{fleiner2003fixed}]\label{stable-infinite}
In any (possibly infinite) one-to-one two-sided matching market, a stable matching exists.
\end{theorem}

\begin{proof}
We prove the theorem using \cref{scaling-lemma}.

\begin{sloppypar}
\proofstep{Definition of $\mathcal{P}$}
Let $\simplemarket$ be a matching market. Let $\mathcal{P}\eqdef\bigl\{(M',W')~\big|~M'\subseteq M \And W'\subseteq W\bigr\}$ be the set of all pairs of subsets of men and subsets of women. A \emph{solution} for $(M',W')\in\mathcal{P}$ is a stable matching between $M'$ and $W'$ (in the induced submarket of $\simplemarket$, i.e., in the market $(M',W',\rprefs{M'}{W'},\rprefs{W'}{M'})$, where $\rprefs{M'}{W'}$ denotes the preference lists of $M'$ induced by $\prefs{M}$, restricted to $W'$, and $\rprefs{W'}{M'}$ is defined analogously). 
\end{sloppypar}

\proofstep{Well describability}
We define a variable $\matched{m}{w}$ for every  $(m,w)\in M\times W$.
In what follows, for each $(M',W')\in\mathcal{P}$ we define a set $\formulae{(M',W')}$ of formulae over these variables so that models (over the variables that appear in $\formulae{(M',W')}$) of $\formulae{(M',W')}$ are in one-to-one correspondence with the (not-yet-proven-to-be-nonempty) set of stable matchings between $M'$ and $W'$.
The correspondence is obtained by endowing the variable $\matched{m}{w}$ with the semantic interpretation ``$m$ and $w$ are matched.'' That is, it maps a model for $\formulae{(M',W')}$ to the matching such that for every $(m,w)\in M'\times W'$, we have that $m$ and $w$ are matched if and only if the variable $\matched{m}{w}$ is $\mathsf{True}$ in that model.
We define the set $\formulae{(M',W')}$ to consist of the following formulae:
\begin{enumerate}
\item
for all $m\in M'$ and all distinct $w,w'\in W'$, the formula
`$\matched{m}{w}\rightarrow\lnot\matched{m}{w'}$',
requiring that $m$ be matched to at most one woman;
\item
for all $w\in W'$ and all distinct $m,m'\in M'$, the formula `$\matched{m}{w}\rightarrow\lnot\matched{m'}{w}$',
requiring that $w$ be matched to at most one man;
\item
for all $m\in M'$ and $w\in W'$ that are not both acceptable to each other, the formula `$\lnot\matched{m}{w}$',
requiring that no one be matched to someone who is unacceptable to them;
\item
for all $m\in M'$ and $w\in W'$ that are acceptable to each other, the formula
$`\lnot\matched{m}{w}\rightarrow \left(\bigvee_{w'\succ_{m}w}\matched{m}{w'}\right)\vee \left(\bigvee_{m'\succ_{w} m}\matched{m'}{w}\right)$',
requiring that $(m,w)$ not be a blocking pair. (Note that for this formula to hold either the left-hand side must be $\mathsf{False}$, i.e., $m$ and $w$ must be matched, or the right-hand side must be $\mathsf{True}$, i.e., one of $m$ and $w$ must be matched to someone they prefer.) This formula is finite, even if the preference lists of $w$ or $m$ are infinite, since $m$ (resp.\ $w$)  prefers only $k$ partners over his (resp.\ her) $k$-th ranked partner.\footnote{Our assumption of a preference list being of the order type of the natural numbers is what allows us to express stability via individually finite formulae---as required for a well description, so that \cref{scaling-lemma} is applicable. \citeN{fleiner2003fixed} studied a model with infinite preference lists of more general order types (beyond which stable matchings are known not to exist), under which such a construction would not be possible.}
\end{enumerate}

By construction, the models (over the variables that appear in $\formulae{(M',W')}$) that satisfy all the formulae of the first and second type above are in one-to-one correspondence with one-to-one matchings between $M'$ and $W'$. Furthermore, the models that satisfy $\formulae{(M',W')}$ are in one-to-one correspondence with stable matchings between $M'$ and $W'$. As noted above, the crux of our argument is that we were able to characterize stability using \emph{individually finite} formulae over our variables. Therefore, $(\formulae{(M',W')})_{(M',W')\in\mathcal{P}}$ is a well description of~$\mathcal{P}$.

\proofstep{Finite-subset property}
Let $\Phi'\subset\formulae{(M,W)}$ be a finite subset.
Since $\Phi'$ is finite, it ``mentions'' (through variables used) only finitely many men and women; denote the sets of these men and these women by $M'\subset M$ and $W'\subset W$, respectively.
By definition, $\Phi'\subseteq\formulae{(M',W')}$. By \cref{stable-finite}, there exists a stable matching between $M'$ and $W'$. Therefore, $(M,W)$~satisfies the finite-subset property. Thus, by \cref{scaling-lemma}, there exists a stable matching between $M$ and~$W$, as required.
\end{proof}

\subsubsection{A Cautionary Tale of a Non-Proof}\label{cautionary}

In constructing a well description,
one has to be careful beyond making sure that each logical formula is finite. Consider, for example, the following ``alternative'' to the fourth formula-type in the preceding proof:
\begin{enumerate}
\item[4'.]
for all distinct $m,m'\in M'$ and distinct $w,w'\in W'$ such that $m$ prefers $w$ to $w'$ and $w$ prefers $m$ to $m'$, the formula
`$\lnot\bigl(\matched{m}{w'}\wedge\matched{m'}{w}\bigr)$'.
\end{enumerate}
This seems to similarly preclude the possibility of $m$ and $w$ being a blocking pair, and is certainly simpler than the corresponding formula-type used in the proof. Nonetheless, upon closer inspection, we see that there is a trivial model that satisfies all formulae if this alternative formula-type is used: the model in which $\matched{m}{w}$ is $\mathsf{False}$ for \emph{every} $m$ and~$w$. So, it is not clear that an existence of a model for this set guarantees the existence of a stable matching (that is, without relying on the correctness of \cref{stable-infinite}). Indeed, we have crafted the original fourth formula-type to also preclude the possibility of this model. One may be tempted to attempt to ``fix'' this, for example by introducing a formula-type that says that each participant must be matched, or by adding a formula-type analogous to this alternative formula-type for the case where $m$ is unmatched (rather than matched to $w'$). However, any such ``fix'' requires expressing the concept ``$m$ is unmatched'' in our formulae, which requires an infinite disjunction that cannot be expressed via individually finite formulae.\footnote{Any attempt to circumvent this by adding a variable that is $\mathsf{True}$ if and only if $m$ is unmatched similarly fails, as forcing this variable to be $\mathsf{True}$ if all $\matched{m}{w}$ are $\mathsf{False}$ again requires an infinite disjunction.}

\subsection{Applicability of the Same Proof to Other Settings:\texorpdfstring{\\}{ }Near-Feasible Stable Matching with Couples}\label{couples}

A notable strength of our approach is that it is agnostic to the methods used to prove the finite result being scaled. Therefore, the same proof can be used to scale similar statements even when the finite-case proofs of these statements hinge on very different tools. As an illustration, 
we use essentially the same proof as in \cref{matching-warmup} to scale to infinite markets the result of \citet{nguyen2018couples} on existence of near-feasible stable matching with couples, despite their proof using an approach that is quite different from those used to prove \cref{stable-finite}.

Following \citeN{nguyen2018couples}, we study the standard matching with couples model \citep[e.g.,][]{Roth1984rural,KlausKlijn2005,KPR2010,ashlagi2011stability}. In this model, a market is a tuple $\couplesmarket$, where $D$ is the disjoint union of the set of single doctors, $D^1,$ and all the doctors in the set $D^2$ of couples of doctors; $H$ is a set of hospitals, $k=\{k_h\}_{h\in H}$ is a vector of hospital capacities; and $\prefs{H}$ is a profile of rankings for the hospitals over the doctors consisting, for each hospital $h\in H$, of a linearly ordered ranking over doctors that either is finite, or specifies hospital $h$'s $n$th-ranked doctor for every $n\in\NN$. Hospitals preferences are \emph{responsive}---from any set of available doctors they choose the highest-ranked ones up to the hospital's capacity (always rejecting unranked doctors). $\prefs{D}=\left(\prefs{D^1},\prefs{D^2}\right)$, is a profile of doctor preference lists. Single doctors' preference lists, $\prefs{D^1}$, are defined as in \cref{matching-warmup}. For couples, $\prefs{D^2}$ is a linearly ordered preference list over ordered pairs in $(H\cup\{\emptyset\})\times(H\cup\{\emptyset\})\setminus\{(\emptyset,\emptyset)\}$, representing the assignment of the first, and second, member of the couple, that the couple prefer to both being unmatched. Each such preference list either is finite, or specifies the couple's $n$th-ranked hospital pair for every $n\in\NN$.

Given a matching-with-couples market and a vector of capacities, $k^*$, a matching is \emph{individually rational with respect to the capacities $k^*$} if no single doctor is assigned to an unacceptable hospital, couples are assigned to $(h,h')$ which they weakly prefer to $(h,\emptyset)$, $(\emptyset,h')$, and $(\emptyset,\emptyset)$, and each hospital $h$ is assigned no more than $k^*_h$ doctors, all of whom are ranked by $\prefs{h}$. A matching $\mu$ is \emph{blocked with respect to the capacities $k^*$} if one of the following holds: (1) there exists a single doctor, $d\in D^1$, and a hospital $h$, such that $d$ prefers $h$ to $\mu(d)$  and $h$'s most preferred subset  of $\mu(h)\cup \{d\}$, subject to $k^*_h$, includes $d$; (2) there exists a couple $c\in D^2$ and a hospital $h$, such that the couple prefers $(h,h)$ to $\mu(c)$ and $h$ would select both members of the couple as above; or, (3) there exists a triple $(c,h,h')\in D^2 \times \left( H \cup \{\emptyset\} \right) \times \left( H \cup \{\emptyset\} \right)$ with $h\ne h'$ such that the couple prefers $(h,h')$ to $\mu(c)$ and each of the hospitals chooses the respective member of the couple from the set as above. Finally, a matching is \emph{stable with respect to the capacities $k^*$} if  it is individually rational and not blocked. 

\begin{theorem}[\citealp{nguyen2018couples}]\label{couples-finite}
In any finite, many-to-one matching market with couples with capacity vector $k$, there exists a capacity vector $k^*$ with $|k_h-k^*_h|\le2$ for every $h\in H$ and $\sum_{h\in H}k_h\le \sum_{h\in H}k^*_h\le \sum_{h\in H}k_h+4$ such that a stable matching w.r.t.\ $k^*$ exists.
\end{theorem}

Despite the fact that \cref{stable-finite,couples-finite} have very different proofs (the former uses deferred acceptance; the latter uses \citeauthor{scarf1967core}'s Lemma), our proofs of their infinite extensions are remarkably similar. As in the one-to-one stable case, we want to describe the solution concept of near-feasible stability with couples via a set of individually finite formulae and then use \cref{scaling-lemma} to take the existence result for finite markets (\cref{couples-finite}) and scale it to infinite markets. To make sure that each logical formula that we come up with is indeed finite, we need a technical assumption on the preferences of couples in the market.

\begin{definition}
We say that a preference order for a couple over $(H\cup\{\emptyset\})^2\setminus\{(\emptyset,\emptyset)\}$ is \emph{downward closed} if for every pair of (actual, i.e., not $\emptyset$) hospitals $(h,h')\in H^2$ ranked by the order, both $(h,\emptyset)$ and $(\emptyset,h')$ are also ranked by the order.
\end{definition}

\begin{theorem}\label{couples-infinite}
In any (possibly infinite) matching market with couples with capacity vector $k$, if the preference of each couple is downward closed,\footnote{Alternatively, we can replace the requirement that couples' preference lists are downward closed with the requirement that they are finite, and essentially the same proof would go through.} then there exists a capacity vector $k^*$ with $|k_h-k^*_h|\le2$ for every $h\in H$ such that a stable matching w.r.t.\ $k^*$ exists.
\end{theorem}

The well description that we build to prove \cref{couples-infinite} is conceptually similar to the one from our proof of \cref{stable-infinite}, but requires slight modification due to the many-to-one nature of the market and the presence of couples, as well as to the variability in capacity.
The idea is to have, as before, for every doctor $d\in D$ and hospital $h\in H$, a variable $\matched{d}{h}$ that will be $\mathsf{True}$ in a model if and only if $d$ and $h$ are matched in the matching corresponding to the model. But furthermore, for every hospital $h\in H$ with capacity $k_h$, we introduce five variables $\quota{h}{k_h-2},\quota{h}{k_h-1},\ldots,\quota{h}{k_h+2}$ such that $\quota{h}{q}$ is $\mathsf{True}$ in a model if and only if $k^*_h=q$ in the corresponding matching, and upon whose value each formula will be conditioned. So, for instance, for each $q\in\{k_h\!-\!2,\ldots,k_h\!+\!2\}$ and for every $q\!+\!1$ doctors we introduce a formula that says ``if the capacity of $h$ is $q$, then $h$ is not matched to these $q\!+\!1$ doctors.''
Except for the potential need to perturb the capacities, and the need to express couples' preferences and the absence of blocks involving couples, the proof runs along the same lines as that of \cref{stable-infinite}; we relegate the details to \cref{app:couples}.

\subsubsection{Conditional Scaling}

 \citet{nguyen2018couples} note that they do not know whether or not the guarantee from \cref{couples-finite} that $|k_h-k^*_h|\le2$ for every $h\in H$ can be improved to a guarantee that $|k_h-k^*_h|\le1$ for every $h\in H$. This provides an opportunity to point out that our framework is agnostic not only  to \emph{how} the finite-case theorem being scaled was proved (as already discussed), but furthermore, to \emph{whether} it has even been proved. Indeed, even absent a proof for the finite-case theorem,
our framework can yield conditional statements.
For example, a proof completely analogous to our proof of \cref{couples-infinite} also proves the following.

\begin{theorem}
If the conclusion of \cref{couples-finite} holds even if the guarantee that $|k_h-k^*_h|\le2$ for every $h\in H$ is replaced with a guarantee that $|k_h-k^*_h|\le1$ for every $h\in H$, then \cref{couples-infinite} also holds with the same change.
\end{theorem}

Hence, if the stronger version of the finite-case result of \cref{couples-finite} is proven in the future, this strengthening will immediately imply the analogous strengthening to the infinite-case result of \cref{couples-infinite} as well.

\subsection{Handling Nondiscrete Solution Concepts:\texorpdfstring{\\}{ }Walrasian Equilibria in Infinite Trading Networks}\label{walrasian}
We turn to matching in trading networks with transferable utility. We study (an infinite variant of) the trading-network framework of \citeN{hatfield2013stability} and show existence of a Walrasian equilibrium---the standard solution concept in trading network matching. 
A trading network is comprised of a (potentially infinite) set $I$ of agents.
A \emph{trade} $\omega$ transfers an underlying object, $\mathsf{o}(\omega)$, from a seller $\mathsf{s}(\omega)$ to a buyer $\mathsf{b}(\omega)$. 
We denote the set of potential trades by $\Omega$.
 For $i\in I$ we denote by~$\Omega_i$ the set of trades in which $i$ participates, namely, $\Omega_i\eqdef\bigl\{ \omega \in \Omega ~\big|~ i\in \{\mathsf{s}(\omega),\mathsf{b}(\omega)\}\bigr\}.$
 We assume that $\Omega_i$ is finite for every $i$---i.e., each agent is a party to finitely many (potential) trades (note that this implies that each agent is endowed with at most finitely many objects to trade, and has at most finitely many trading partners). 

Each agent's utility depends only on the trades that she executes (or does not execute), and the prices at which these trades are executed. 
Specifically, each agent $i$ is associated with a utility function that is quasilinear in prices and otherwise depends only on the set of trades $\Omega'_i \subseteq \Omega_i$ that are executed.

The ``trades'' terminology that we use highlights that in this model, objects are linked to specific trading partners---so ``car sold to Alice"  is a different object than ``car sold to Bob," even when the physical good that is traded in reality is the same car. (To rule out the possibility that the same car is traded to multiple people, the agent's utility function can assign value $-\infty$  to  executing ``car sold to Alice" and ``car sold to Bob" simultaneously.) Having clarified this issue, it is easier to think about the model in terms of objects from this point on.

Let $O$ denote the set of all objects.  Note that objects and trades are in one-to-one correspondence. For an object $o$, we let $\mathsf{t}(o)$ denote the trade associated with that object, so for each object $o$ we have $\mathsf{o}(\mathsf{t}(o))=o$, and for each trade $\omega$ we have $\mathsf{t}(\mathsf{o}(\omega))=\omega$.
For each agent~$i$ we denote by  $O_i\eqdef\bigl\{o\in O ~\big|~ i\in\{\mathsf{s}(\mathsf{t}(o)),\mathsf{b}(\mathsf{t}(o))\}\bigr\}$ the set of all objects that can be held by~$i$. (Note that $O_i$ is finite since $|O_i|=|\Omega_i|$.) Each $o\in O$ belongs to exactly two sets in $\{O_i\}_{i\in I}$.

We assume that for each $i$, the valuation $u_i(\cdot):2^{O_i}\rightarrow\RR\cup \{-\infty\}$ takes values in $\RR\cup\{-\infty\}$ (again, where we use $-\infty$ to model technological impossibilities such as selling the same car to multiple buyers).
Agent $i$'s valuation expresses the value of the objects that~$i$ consumes, or ``holds'' after the execution of trades (that is, the set of objects $\bigl\{\mathsf{o}(\omega)\mid\omega\in \Omega'_i \text{ and } \mathsf{b}(\omega)=i\bigr\} \cup \bigl\{\mathsf{o}(\omega)\mid\omega\notin \Omega'_i \text{ and } \mathsf{s}(\omega)=i\bigr\}$; see \citealp{hatfield2015full}).
We furthermore assume that in the absence of trade, $i$'s valuation is equal to $0$ (formally, $u_i\bigl( \{\mathsf{o}(\omega)\mid \mathsf{s}(\omega)=i\}\bigr)=0$).

For each agent $i$, let the \textit{demand correspondence} $D_i: p\in \RR^{O_i} \rightrightarrows 2^{O_i}$ be defined by
$D_i(p)=\argmax_{O'\subset O_i}\bigl(u_i(O')-\sum_{o\in O'}p_o\bigr)$.
The preferences of agent $i$ are \emph{(gross) substitutable} if
for all price vectors $p,p'\in\RR^{O_i}$ such that $\bigl|D_i(p)\bigr|=\bigl|D_i(p')\bigr|=1$ and $p\leq p'$,
if $o \in  D_i(p)$ then $o \in D_i(p')$ for each $o \in O_i$ such that $p_{o}=p'_{o}$.

A \emph{Walrasian equilibrium} consists of
	a profile of prices $p\in\RR^{O}$, and
	a partition $\{O'_i\}_{i \in I}$ of $O$,
	such that $O'_i \in D_i(p\restr_{O_i})$ for each $i\in I$.
\citeN{hatfield2013stability}   have shown that substitutable preferences suffice to guarantee the existence of a Walrasian equilibrium in finite trading networks.

\begin{theorem}[\citealp{hatfield2013stability}]\label{walrasian-finite}
Every finite trading network in which all agents have substitutable preferences has a Walrasian equilibrium. 
\end{theorem}

We use \cref{scaling-lemma} to scale \cref{walrasian-finite} to infinite trading networks.

\begin{theorem}\label{exact-walrasian}
Every (possibly infinite) trading network in which all agents have  substitutable preferences has a Walrasian equilibrium. 
\end{theorem}

We prove \cref{exact-walrasian} using \cref{scaling-lemma}. One of the challenges in our argument is that prices have an (uncountably) infinite range, so it is not \emph{a priori} obvious how to encode prices by a model defined via individually finite formulae (e.g., how to require that each item is associated with some real number that represents its price); to overcome this challenge, our well description encodes each price  as the limit of a  sequence of discrete prices.
This approach, in turn, introduces additional challenges, such as how to make sure, using only constraints on these discrete prices, that the limit prices satisfy all desired properties.

\begin{proof}
We prove the theorem using \cref{scaling-lemma}. We first note that for every object $o\in O$, there exists a positive integer $H_o$ such that (1)~if the price of~$o$ is $H_o$, then $\mathsf{s}\bigl(\mathsf{t}(o)\bigr)$ always wishes to sell~$o$ and $\mathsf{b}\bigl(\mathsf{t}(o)\bigr)$ always wishes to not buy $o$ (unless this results in a technological impossibility for one of them), and, (2)~if the price of $o$ is $-H_o$ (and there is no technological impossibility), then $\mathsf{s}\bigl(\mathsf{t}(o)\bigr)$ always wishes to hold $o$ and $\mathsf{b}\bigl(\mathsf{t}(o)\bigr)$ always wishes to buy $o$ (unless this results in a technological impossibility for one of them). Formally, there exists $H_o$ such that for every $i\in\bigl\{\mathsf{s}(\mathsf{t}(o)),\mathsf{b}(\mathsf{t}(o)\bigr\}$ and $O'_i\subseteq O_i$, if $u_i\bigl(O'_i\cup  \{o\}\bigr) > -\infty$ and $u_i(O'_i) > -\infty$ (i.e., if there is no technological impossibility) then $\bigl|u_i(O'_i\cup  \{o\})-u_i(O'_i)\bigr|<H_o$.\footnote{Such a finite $H_o$ exists since there are finitely many such differences for each object $o\in O$.}

\proofstep{Definition of $\mathcal{P}$}
Let $\mathcal{P}$ be all the subsets of $I$. A \emph{solution} for $I'\in\mathcal{P}$ is a Walrasian equilibrium in the induced trading network containing only the agents in $I'$, and only the objects $O_{I'}\eqdef\bigcup_{i\in I'}O_i\subseteq O$ (i.e., objects that can be owned by agents in $I'$), such that each $o\in O_{I'}$ is priced in $[-H_o,H_o]$. (Objects for which only the buyer or only the seller are in $I'$ can be priced arbitrarily, as for them there is no notion of market-clearing prices.)

\proofstep{Well describability}
We set $\varepsilon_n\eqdef2^{-n}$ for every $n\in\NN$.
We define a variable $\consumes{i}{o}$ for every $i\in I$ and every $o\in O_i$, and a variable $\price{n}{o}{p}$ for every $n\in\NN$, every object $o\in O$, and every possible price\linebreak $p\in P_o^n\eqdef\{-H_o,\ldots,-\varepsilon_n,0,\varepsilon_n,2\varepsilon_n,\ldots,H_o\}$. In what follows, for each $I'\in\mathcal{P}$ we define a set $\formulae{I'}$ of formulae over these variables so that 
models of $\formulae{I'}$ are in one-to-one correspondence with the (not-yet-proven-to-be-nonempty) set of solutions for $I'$.
The correspondence is obtained by endowing the variable $\consumes{i}{o}$ with the semantic interpretation ``$i$ consumes object $o$,'' and the variable $\price{n}{o}{p}$ with the semantic interpretation ``$\epsfloor{p_o}=p$,'' where for every $n\in\NN$ and every $x$ we denote by  
$\epsfloor{x}\eqdef2^{-n}\cdot\lfloor2^n\cdot x\rfloor$ 
the rounding-down of~$x$ to the nearest multiple of $\varepsilon_n$.
For every $\varepsilon>0$, the \textit{approximate demand correspondence} $D_i^{\varepsilon}: p\in\RR^{O} \rightrightarrows 2^{O_i}$ is defined similarly to the (exact) demand correspondence, except that  its range includes bundles from which agent $i$'s utility is at least that of the utility-maximizing bundle minus~$\varepsilon$. Thus, $D_i(p)\subseteq D_i^{|O_i|\varepsilon_n}(\epsfloor{p})$ for all $i$ and $p=(p_o)_{o\in O_i}\in\bigtimes_{o\in O_i}[-H_o,H_o]$ (where $\epsfloor{p}$ rounds each coordinate separately). Furthermore, $D_i(p)=\bigcap_{n\in\NN}D_i^{|O_i|\varepsilon_n}(\epsfloor{p})$, since weak inequalities are preserved in the limit.  
We define the set $\formulae{I'}$ to consist of the following formulae:
\begin{enumerate}
\item
for all $n\in\NN$ and all $o \in O_{I'}$, the (finite!) formula `$\bigvee_{p\in P_o^n}\price{n}{o}{p}$',
requiring that $o$ have a rounded-down-to-$\varepsilon_n$ price in $P_o^n$;
\item
for all $n\in\NN$, all $o\in O_{I'}$, and all distinct $p,p'\in P_o^n$, the formula `$\price{n}{o}{p}\rightarrow\lnot\price{n}{o}{p'}$',
requiring that the rounded-down-to-$\varepsilon_n$ price of $o$ be unique;
\item
for all $n\in\NN$, all $o\in O_{I'}$, and all $p\in P_o^n$, the formula `$\price{n}{o}{p}\rightarrow\bigl(\price{n+1}{o}{p}\vee\price{n+1}{o}{p+\varepsilon_{n+1}}\bigr)$',
requiring that $\epsfloor{p_o}=\epsfloor{\epsppfloor{p_o}}$;
\item
for all $o\in O_{I'}$ such that both $\mathsf{b}(\mathsf{t}(o))$ and $\mathsf{s}(\mathsf{t}(o))$ are in $I'$, the formula
`$\consumes{\mathsf{b}(\mathsf{t}(o))}{o}\leftrightarrow\lnot\consumes{\mathsf{s}(\mathsf{t}(o))}{o}$',
requiring that $o$ either be sold or not at its rounded-down-to-$\varepsilon_n$ price. Equivalently, it requires that the associated trade $t(o)$ either be executed or not at these prices;
\item For all $n\in\NN$, all $i\in I'$ and all profiles of possible prices $p=(p_o)_{o\in O_i}\in\bigtimes_{o\in O_i}P_o^n$, the (finite) formula \pagebreak
\[
\biggl(\bigwedge_{o \in O_i}\price{n}{o}{p_o}\biggr) \rightarrow\biggl( \bigvee_{X\in D_i^{|O_i|\varepsilon_n}(p)}\Bigl(\bigwedge_{x\in X}\consumes{i}{x}\wedge\bigwedge_{x\in O_i \setminus X}\lnot\consumes{i}{x}\Bigr)\biggr),\label{f4}
\]
requiring that $i$ consumes one (and only one) of her $(|O_i|\cdot\varepsilon_n)$-utility-maxi\-mizing bundles (and not any object outside this bundle) with respect to the rounded-down-to-$\varepsilon_n$ prices. Taken together for all $n\in\NN$, together with the third formula-type above, this guarantees that $i$ consumes a bundle in $D_i(p)$.
\end{enumerate}
By construction, $(\formulae{I'})_{I'\in\mathcal{P}}$ is a well description of~$\mathcal{P}$. (To obtain an equilibrium from a model, for every $o\in O_{I'}$ and $n\in\NN$, let $p_o^n\in P_o^n$ be the value such that $\price{n}{o}{p}$ is $\mathsf{True}$ in the model (well defined by the first and second formulat-types above), and define $p_o=\lim_{n\rightarrow\infty}p^n_o$ (well defined, e.g., by the third formula-type above since $p^n_o$ is a Cauchy sequence.)

\proofstep{Finite-subset property}
Let $\Phi'\subset\formulae{I}$ be a finite subset.
Since $\Phi'$ is finite, it ``mentions'' (through variables used) only finitely many agents; denote the set of these agents by $I'\subset I$. By definition, $\Phi'\subseteq\Phi_{I'}$. Fix the prices of all objects $o\in O_{I'}$ for which only the seller or the buyer are in $I'$ to be zero. By \cref{walrasian-finite}, there exists a Walrasian equilibrium in the trading network induced by $I'$. For each $o\in O_{I'}$, by definition of $H_0$, if the price of $o$ in this equilibrium is not in $[-H_o,H_o]$, then it can be replaced by $H_o$ (if it is positive) or by $-H_o$ (if it is negative) without changing the allocation, and we would still have a Walrasian equilibrium, this time with each object $o\in O_{I'}$ priced in $[-H_o,H_o]$.
Therefore, $I$~satisfies the finite-subset property. Thus, by \cref{scaling-lemma}, there exists a Walrasian equilibrium in the grand trading network (containing all agents in $I$), as required.
\end{proof}

\cref{exact-walrasian} can also be proved using an alternative approach that does not use formulae of the third formula-type above. The same argument then establishes the existence of a weak kind of approximate Walrasian equilibruim (where the approximation level is agent-specific and equal to $|O_i|\cdot\varepsilon$, which is not necessarily uniformly bounded). Then, a standard diagonalization argument suffices to complete the proof. While we demonstrate this approach in our proof of the existence of Nash equilibrium in game on infinite graphs (\cref{nash}), we choose here the shorter and more generalizable proof approach.\footnote{While it is simple to encode diagonalization arguments using additional formulae, the converse is not always true. See, for example, the applications discussed in \cref{garp}. }

\subsection{Combining with Domain-Specific Knowledge:\texorpdfstring{\\}{ }Structure and Incentives}\label{sp}

In addition to proving \cref{stable-finite}, \cite{GaleShapley1962} also proved that there exists a \textit{man-optimal stable matching}, that is, a stable matching that is most preferred (among all stable matchings) by all men simultaneously. \citet{fleiner2003fixed} generalized this result to infinite markets as part of his proof of \cref{stable-infinite}.\footnote{\citeauthor{fleiner2003fixed}'s result is in fact more general: that the set of stable matchings has a lattice structure.}

\begin{theorem}[\citealp{GaleShapley1962}]\label{menopt-finite}In any finite, one-to-one matching market, there exists a man-optimal stable matching.
\end{theorem}

\begin{theorem}[\citealp{fleiner2003fixed}]\label{menopt-infinite}
In any (possibly infinite) one-to-one matching market, there exists a man-optimal stable matching. 
\end{theorem}

\noindent For completeness, we reprove \cref{menopt-infinite} using \cref{scaling-lemma} in \cref{menopt}.

With \cref{menopt-finite}/\cref{menopt-infinite} in hand, one can define the \emph{man-optimal stable matching mechanism}, which, given any preference profile for the participants, outputs the man-optimal stable matching with respect to those preferences. In finite markets,  \citeN{DubinsFreeman1981} and \citeN{roth1982economics} show that this mechanism is \textit{strategy-proof} (for the men), i.e., that truthfully reporting one's preferences is a (weakly) dominant strategy for each man in the market.

\begin{theorem}[\citealp{DubinsFreeman1981,roth1982economics}]\label{sp-finite}
In any finite, one-to-one matching market, the man-optimal stable matching mechanism is strategy-proof for men.  
\end{theorem}

The challenge in scaling  \cref{sp-finite} using \cref{scaling-lemma} is twofold.
First, standard arguments for strategy-proofness rely on versions of the Lone Wolf/Rural Hospitals Theorem \citep{Roth1984rural},
which \citeN{jagadeesan2017lone2} has shown to not hold in our setting, so an innovation on proof strategy is needed here---and moreover, the mere applicability of \cref{sp-finite} to infinite markets is not \emph{a priori} clear.
Indeed, while in the same paper \citet{jagadeesan2017lone2} proves a special case of \cref{sp-infinite} in which all preference profiles are finite, he leaves the general case (that we prove) as an open problem.
Second, \cref{sp-finite} is of a very different flavor than the results we have scaled to infinite settings using \cref{scaling-lemma} so far. What \cref{scaling-lemma} most naturally helps us prove is existence, yet \cref{sp-finite} is not a result on the existence of a stable matching. Moreover, it does not describe any structural property of a stable matching, but rather an elusive game-theoretic/economic property of a function from preference profiles to stable matchings,
which might seem quite far from a standard existence result. Nevertheless, for the main result of this \lcnamecref{sp}, we use \cref{scaling-lemma} to scale \cref{sp-finite} to infinite markets by first recasting \cref{sp-finite} as an existence result; in doing so, we thus illustrate that \cref{scaling-lemma} has applications beyond results we might naturally (at least at first glance) expect to be able to prove with limiting or continuity arguments. 

\begin{theorem}\label{sp-infinite}
In any (possibly infinite) one-to-one matching market, the man-optimal stable matching mechanism is strategy-proof for men.  
\end{theorem}
The proof of \cref{sp-infinite} is conceptually more involved than those above:
\begin{itemize}
\item
The well description constructed for the infinite problem may not be the most straightforward way to model the result we wish to prove. Indeed, as already mentioned, we recast \cref{sp-infinite}---which might seem quite far from an existence result---as an existence result to which we can apply \cref{scaling-lemma}.
\item
The solution defined for the smaller (finite) problems is not what one would expect as ``the analogous stand-alone problem but in the smaller market,'' i.e., it is \emph{not} equivalent to ``lack of a profitable manipulation in the smaller market'' but rather a problem defined relying on the specifics of the infinite problem, i.e., ``lack of a manipulation in the smaller market that would improve upon the man-optimal stable match \emph{in the original, infinite market}.''
\item
The argument more carefully chooses the finite market to be used for proving the finite-subset property, i.e., it is not ``the market of all participants mentioned in any formula in the given finite subset of formulae,'' but rather a carefully chosen larger (yet still finite) market, chosen so that certain specific facts about the infinite market can be projected down to the finite market.
\end{itemize}

\begin{proof}[Proof of \cref{sp-infinite}]
Let $\simplemarket$ be a matching market. Let $\tilde{m}\in M$ be a man, and let $\prefs{\tilde{m}}'$ be alternative preferences for $\tilde{m}$ to report in lieu of~$\prefs{\tilde{m}}$. We must show that $\tilde{m}$ weakly (truly) prefers his match when reporting his true preferences $\prefs{\tilde{m}}$ (i.e., his match in the man-optimal stable matching in the market $\simplemarket$) to his match when reporting $\prefs{\tilde{m}}'$ (i.e., his match in the man-optimal stable matching in the market $\bigl(M,W,(\prefs{M\setminus\{\tilde{m}\}},\prefs{\tilde{m}}'),\prefs{W}\bigr)$).

If $\tilde{m}$ is either unmatched or matched to a woman not on his preference list $\prefs{\tilde{m}}$ when reporting $\prefs{\tilde{m}}'$, then the claim follows immediately from the individual rationality of stable matchings. We therefore focus on the case in which when reporting $\prefs{\tilde{m}}'$, man $\tilde{m}$ is matched to a woman $\tilde{w}$ who is on his preference list $\prefs{\tilde{m}}$ (i.e., who he prefers to being unmatched).
Since the man-optimal stable matching is preferred by $\tilde{m}$ to any other stable matching, it suffices to show that there exists a stable matching in the market $\simplemarket$ in which $\tilde{m}$ is matched to a woman he weakly prefers to $\tilde{w}$. We prove this existence result using \cref{scaling-lemma}.

\proofstep{Definition of $\mathcal{P}$}
We use the same set of problems as in the proof of \cref{stable-infinite}, i.e., $\mathcal{P}\eqdef\bigl\{(M',W')~\big|~M'\subseteq M \And W'\subseteq W\bigr\}$. The definition of a solution  will be different, though: A \emph{solution} for $(M',W')\in\mathcal{P}$ is a stable matching between $M'$ and $W'$ (in the induced submarket of $\simplemarket$) in which $\tilde{m}$ is matched to a woman he weakly prefers to $\tilde{w}$.

\proofstep{Well describability}
We use the same set of variables as in the proof of \cref{stable-infinite}, with the same semantic interpretation, however the well description is different.
For each $(M',W')\in\mathcal{P}$, we define the set $\formulae{(M',W')}$ to consist of the same formulae as in the proof of \cref{stable-infinite}, and in addition:
\begin{enumerate}
    \item[5.] the (finite!) formula
`$\bigvee_{w\succcurlyeq_{\tilde{m}}\tilde{w}}\matched{\tilde{m}}{w}$',
requiring that $\tilde{m}$ be matched to a woman he weakly prefers to $\tilde{w}$.
\end{enumerate}
 By construction, models of $\formulae{(M',W')}$ (over the variables that appear in $\formulae{(M',W')}$)  are in one-to-one correspondence with the (not-yet-proven-to-be-nonempty) set of solutions of $(M',W')$, 
and so $(\formulae{(M',W')})_{(M',W')\in\mathcal{P}}$ is a well description of~$\mathcal{P}$.

\proofstep{Finite-subset property}
Let $\Phi'\subset\formulae{(M,W)}$ be a finite subset.
Since $\Phi'$ is finite, it ``mentions'' (through variables used) only finitely many men and women; denote the sets of these men and these women by $M'\subset M$ and $W'\subset W$, respectively.
We now reach the main point of divergence from the proof of \cref{stable-infinite}, and, in fact, from all proofs presented so far: the problem that we use to show that the finite-subset property is satisfied is not the problem $(M',W')$.
Rather, it is $(M'',W'')$ for $M''\eqdef M'\cup\{\tilde{m}\}\cup\mu(W')$ and $W''\eqdef W'\cup\{\tilde{w}\}\cup\mu(M')$, where we denote by $\mu$ the matching of the entire market $(M,W)$ that the mechanism chooses when $\tilde{m}$ reports $\prefs{\tilde{m}}'$. Recall that $\mu(\tilde{m})=\tilde{w}$. By definition, $\Phi'\subseteq\formulae{(M'',W'')}$. We now show that the man-optimal stable matching in $(M'',W'')$ constitutes a solution for $(M'',W'')$.\footnote{The argument that we give for this would not have shown that the man-optimal stable matching in the market $(M',W')$ is a solution for $(M',W')$.}

One stable matching in the market $\bigl(M'',W'',(\rprefs{M''\setminus\{\tilde{m}\}}{W''},\{\tilde{w}\}),\rprefs{W''}{M''}\bigr)$ is the restriction of $\mu$ to this market, which matches $\tilde{m}$ with $\tilde{w}$ (it is stable since any blocking pair would also block the matching $\mu$ in $\bigl(M,W,(\prefs{M\setminus\{\tilde{m}\}},\prefs{\tilde{m}}'),\prefs{W}\bigr)$). Since $\tilde{w}$ is the only woman that $\tilde{m}$ ranks in this market, the man-optimal stable matching in this market also matches $\tilde{m}$ with $\tilde{w}$. So, by \cref{sp-finite}, it must be that the man-optimal stable matching in $(M'',W'')$ matches $\tilde{m}$ to a woman he weakly $\prefs{\tilde{m}}$-prefers to $\tilde{w}$ (otherwise, $\tilde{m}$ would have had a profitable manipulation in the latter market: declaring his list to consist solely of $\tilde{w}$), and hence this matching is a solution for $(M'',W'')$.
Therefore, $(M,W)$~satisfies the finite-subset property. Thus, by \cref{scaling-lemma}, there exists a stable matching between $M$ and~$W$ in which $\tilde{m}$ is matched to a woman he weakly prefers to $\tilde{w}$, as required.
\end{proof}

 \subsubsection{Additional Application of the Same Proof: Man-Optimality}

Our proof of \cref{sp-infinite} is quite a bit more flexible than one might imagine.
In \cref{menopt}, we use a similar well description and a similar argument to scale \cref{menopt-finite} from finite to infinite markets, reproving \cref{menopt-infinite}.
This is despite the differences between the flavors of these two results, and despite the fact that it is not clear that techniques from \citeauthor{fleiner2003fixed}'s original proof of \cref{menopt-infinite} could be used to prove \cref{sp-infinite}. The main differences between our proofs of \cref{menopt-infinite} and \cref{sp-infinite}  is that in this application on the one hand the problem is already an existence problem so we do not need to recast it as one, but on the other hand we choose the finite markets so that several (not just one) stable matchings of the infinite market could be restricted into them. 
Subsequent to our work, \citet{choi20} has shown how to further use our proof technique to scale several additional structural results for two-sided matching, such as the classic comparative static on the impact of adding an agent to one side of the market.

\subsection{Further Remarks and Limitations}\label{sec:limitat}

In this section, we demonstrated the versatility of our approach across several settings within matching theory. Our approach can be used to scale many additional finite-market results to encompass infinite markets. For example, \cite{choi20} has used our framework to scale the classic group strategy-proofness result for the man-optimal stable matching mechanism. Our approach can also be applied to dynamic matching markets to adapt finite-time results to settings with no start time or no end time. For example, in \cref{dynamic} we apply it to the dynamic matching result of \citet{pereyra2013dynamic}.
Our approach, however, is not without limitations. We conclude this section with some remarks on limitations of proofs presented throughout the section. A more high-level discussion of settings in which our approach is not applicable is provided in \cref{non-applications}.

In \cref{sp} we proved the strategy-proofness of the man-optimal stable matching mechanism. We have previously mentioned that our proof of this result does not rely on the Lone Wolf/Rural Hospitals Theorem,\footnote{The Lone Wolf Theorem states that in a finite one-to-one matching market, each participant is either matched in all stable matchings, or unmatched in all stable matchings. The Rural Hospitals Theorem is a strengthening of this theorem for many-to-one matching markets.} while the latter was shown by \citeN{jagadeesan2017lone2} via a counterexample to not hold in infinite settings.
By \cref{scaling-lemma},
the existence of this counterexample, together with the correctness of the Lone Wolf Theorem for finite markets, imply that for infinite markets to satisfy the Lone Wolf Theorem is not well-describable with the variables that we used in this section (with the same semantic interpretation).
A natural question, though, is what would fail in one's first attempt at such a well description.
The answer is that a Propositional Logic formulation precludes the use of quantifiers (e.g., ``\emph{there exists} a man $m$ to whom woman $w$ is matched'') as well as precludes infinitely long formulae such as infinite disjunctions (e.g., ``woman $w$ is matched to \emph{one of the following infinitely many} men''). Furthermore, it does not allow statements that condition truths in one model upon truths in other models, such as requiring that a certain statement (e.g., that a certain agent is matched) hold either in all models of our formulae or in none of them. Any attempt at a well description for satisfying the Lone Wolf/Rural Hospitals Theorem that satisfies the finite-subset property would in effect run into one of these three obstacles.

As another example, consider the following question:\footnote{We thank an anonymous referee for the 21st ACM Conference on Economics and Computation (EC 2020) for raising this question.} Can our framework be used to show that there exists a mechanism that is strategy-proof, efficient, and individually rational in a housing market \`a la \citeN{shapley1974cores}, when the sets of objects and agents are infinite?
The answer is that such a construction is at least not straightforward because it is unclear whether and how (Pareto) efficiency can be well described
(via individually finite formulae).
A well description can certainly naturally rule out any \emph{finite} trading cycles. However, unlike in the finite setting, in an infinite setting the absence of finite trading cycles does not imply efficiency, as there may exist Pareto improvements that require infinitely many trades. To use our framework, one would need to come up with
a well description that rules out any inefficiency (for example, by appealing to some variant of the first welfare theorem).\footnote{Similarly, \citeauthor{choi20}'s (\citeyear{choi20}) subsequent result scaling group strategy-proofness guarantees the absence of profitable deviations for any \emph{finite} coalition. Generalizing to infinite coalitions using the same proof technique would once again require a way to well describe the absence of profitable deviations by infinite coalitions (using individually finite formulae), which may be impossible.}

\addtocontents{toc}{\protect\setcounter{tocdepth}{1}}

\section{Non-Applications}\label{non-applications}

When using our framework, one faces an inherent tension. On the one hand, each formula in a well description must use finitely many variables. On the other hand, to be able to use finite results to establish the finite-subset property, each variable must be semantically related only to a finite set of elements in the economic problem. At first glance, this seems to preclude applications in which the desired solution has a parameter with an infinite domain, since requiring that the parameter take \emph{some} value would require an infinite disjunction. Indeed, it is simpler to handle parameters with finite domains (which, as we have seen, naturally occur in many of the applications). Nevertheless, we have successfully applied \cref{scaling-lemma} also to settings with infinite-domain parameters, such as utilities/costs (\cref{garp}), probabilities (\cref{stoch}), or prices (\cref{walrasian}). Still, as the following examples demonstrate, in seemingly similar problems this approach could not possibly work, since the infinite case has no solution. 

\begin{example}[Splitting the dollar]
A dollar must be split between a set $I$ of agents. A solution is an efficient and envy-free division. When $|I|<\infty $ splitting the dollar equally is a solution. But the case $I=\NN$ has no solution.      
\end{example}

\begin{example}[Higher number wins]
Two players can state a number in $S\subseteq \RR$. The player whose stated number is higher wins a prize (which is shared in case of a tie). A solution is a pure-strategy Nash equilibrium. 
When $|S|<\infty$, a solution exists (each player states $\max S$), but the case $S=\NN$ has no solution. 
\end{example}

What are the limitations of our approach that prevent it from covering these two examples?
Our approach for well-describing problems with infinite-domain parameters is to ``encode'' these parameters via a sequence of values, each from a finite domain. Specifically, we have encoded each of the abovementioned parameters using a sequence of increasingly fine discretizations.

Two features are critical for the success of this approach. First, that the desiderata on the encoded parameter can be imposed by individually finite formulae on the discretizations. For example, in \cref{garp}, since a utility function weakly rationalizes the data if and only if each of its discretizations weakly rationalize the data, we represented utility functions that satisfy the former by discretizations that satisfy the latter. Similarly, in \cref{walrasian}, since a consumption bundle maximizes utility with respect to prices if and only if it approximately maximizes utility with respect to each discretization of the prices, we represented prices that satisfy the former by discretizations that satisfy the latter. The second critical feature is that not only each parameter value can be encoded by such a sequence of values (discretizations), but also each such sequence from any valid model encode a valid parameter value (i.e., in the parameter domain). In \cref{stoch}, since probabilities are increasing and in $[0,1]$, they have a limit, which is in $[0,1]$. In \cref{inattention}, since discretized costs are increasing, they have a limit, which is in $\mathbb{R}_{\ge0}\cup\{\infty\}$---the domain of costs in that problem. For each of the above examples, there is no encoding that has both of these features.

In splitting the dollar, if, for example, we encode each player's allocation using discretizations, there is no set of individually finite formulae on the discretizations that holds if and only if the limit division is efficient. Hence, the first feature is missing. Any encoding that has this feature would lose the second feature.

In higher number wins, if, for example, we encode each player's number using discretizations, there is no set of individually finite formulae on the discretizations that holds if and only if the limit is finite. Hence, the second feature is missing. Any encoding that has this feature would lose the first feature.

\section{Related Literature} \label{related-literature}

\paragraph{Methodology.} To our knowledge, we are the first to use Propositional Logic as a general tool for scaling results in economics.\footnote{It is worth mentioning within this context, though, the work of \citet{holzman1984extension}, who used Logical Compactness to relax topological conditions in \citet{fishburn1984comment}.} We elaborate on related  approaches in \cref{related-methodology}.

\paragraph{Decision theory and revealed preferences applications.}
In decision theory, both finite- and infinite-data models are important and common.\footnote{For a recent example of a treatment of finite and infinite datasets, see \citet{aguiar2020rationalization}.}
\citet{reny2015} showed how to unify these two approaches in the setting of \citet{afriat67}.
In our view, our main contribution to this literature is in generalizing beyond any specific setting by providing a way to systematically unify these approaches. 
In \cref{stoch}, we scaled a finite-data theorem for stochastic choice originally due to \citet{MR1971,MR1990} to cover infinite datasets.  The importance of infinite data in this particular setting was highlighted by \citet{cohen1980} and \citet{mcfadden2005review}. \citet{mcfadden2005review}, in particular, showed how to scale the \citet{MR1971,MR1990}  result to an infinite setting (different from ours) by either weakening the concept of rationalizability or by imposing topological structure---neither of which we require in our analysis.
In \cref{filter}, we proved that the result of  \citet{masatlioglu2012revealed} scales to infinite datasets using essentially the same proof we used to reprove the classic result of \citet{richter1966} and \citet{hansson1968choice} that SARP suffices for rationalization by strict preferences. Like  \citeauthor{masatlioglu2012revealed}, our infinite version of their theorem applies to \emph{full} datasets. \citet{de2021bounded} provide an analogous theorem for finite datasets that need not necessarily be full; our proof can be used to similarly scale their theorem to infinite datasets.

\paragraph{Matching theory and  game theory applications.} Infinite models are used frequently in matching theory, and in game theory more broadly, as a way of representing limit---or ``large''---markets. Many of these models work with either discrete infinite markets \citep{fleiner2003fixed,kircher2009,jagadeesan2017complementary} or a  limit of finite markets \citep{immorlica2005marriage,kojima2009incentives,ashlagi2011stability,MirallesPycia2015}; we work with similar limit settings, but use Compactness to sidestep many of the challenges in analyzing those models.\footnote{A second class of large-market models features continua of agents, with each agent having a negligible contribution to the overall market \cite[see, e.g.,][]{AS1974,gretsky1992nonatomic,gretsky1999perfect,kaneko1986core,azevedo2013walrasian,noldeke2018implementation,AzevedoBudish2018,greinecker2018pairwise}. A series of recent papers has introduced models that mix between the two styles by featuring countably many ``large'' agents that can each match with a continuum of ``small'' agents \cite[see., e.g.,][]{azevedo2013supply,azevedo2012complementarity,jagadeesan2017complementary,che2019stable,fuentes2019stable}.} Our existence and structural results for large matching markets  (but not our strategy-proofness result) are implicitly covered by the main result of \citeN{fleiner2003fixed}, who introduced a fixed-point characterization of stable matchings that also holds in infinite markets. Our strategy-proofness result, meanwhile, generalizes the more restricted result of \citeN{jagadeesan2017lone2}, who proved strategy-proofness of
the man-optimal stable matching mechanism in markets with countably many agents under a ``local finiteness" condition (which we avoid).

In \cref{dynamic}, we consider the dynamic matching result of \citet{pereyra2013dynamic}, which originally applies to dynamic markets with a finite start time, and scale it to encapsulate ``ongoing'' dynamic markets with neither start nor end times.  This application connects our work to the broad literature on infinite-horizon games \cite[e.g.,][]{fudenberg2009long}. 
In \cref{nash} we reprove the main result of \citet{Peleg1969}, who directly scaled the seminal existence result of \citet{Nash1951} to infinite settings.  
\citet{hart1989existence} proved existence of correlated equilibrium in infinite games, including  some cases in which a Nash equilibrium does not exist; 
future studies may look to apply our approach to this and other equilibrium notions.  

\section{Discussion}\label{discussion}

This paper provides a novel, principled approach for scaling economic theory results from finite models to infinite ones. We identify a sufficient condition for scaling a result: well-describability with a description satisfying the finite-subset property. 
The bulk of this paper is dedicated to applications demonstrating that many results in revealed-preference theory and matching theory meet this condition.

We have focused on two ``types'' of scaling to infinity: allowing infinite datasets, and allowing infinite markets. In \cref{dynamic,nash}, we provide examples of a third type: allowing infinite time. Specifically, we generalize repeated economic interactions with a finite start-time (and in \cref{nash} also finite end-time) to accommodate ``ongoing'' dynamic interactions with neither start nor end, and show the existence of equilibria free from artefacts driven by time starting or ending at a fixed date. In \cref{dynamic}, our approach allows us to derive the infinite case by relying on an inductive proof for the finite case, even though in the infinite case there is no ``time~0'' induction base.

Our approach is not without limitation, and may fail where other approaches can succeed. That said, we have curated an array of applications showing that it has merit in both decision theory and matching theory (and potentially in other contexts) in proving novel results, as well as consolidating and shortening proofs of previously known results, in a way that oftentimes sheds new light on them.

Despite the novel decision theory and matching theory results we derive, we view the main contribution of this paper to be a methodological one: a new, easy-to-use versatile tool for the economics toolbox. As we illustrated across both fields, combining our approach with field-specific knowledge has the potential to overcome challenging technical barriers. We hope that readers of this paper will be able to
further leverage our approach.

\small
\singlespacing
\addcontentsline{toc}{section}{References}
\putbib
\end{bibunit}

\clearpage
\addtocontents{toc}{\vspace{1em}}
\setcounter{page}{1}
\renewcommand{\thepage}{A.\arabic{page}}
\setcounter{table}{0} 
\setcounter{footnote}{0} 
\renewcommand{\thefootnote}{\arabic{footnote}}
\appendix

\begin{bibunit}

\section{Rational Inattention}\label{inattention}

In this \lcnamecref{inattention}, we demonstrate how our proof for \cref{garp-infinite} is quite a bit more flexible than one might imagine, by using essentially the same well description to scale quite a different rationalization result from finite to infinite datasets. For the most part, we  follow the notation of \citet{caplin2015revealed}. 
Fix a finite set $\Omega$ of possible \emph{states of the world} and a \emph{prize space} $X$, as well as a \emph{utility function} $u:X\rightarrow\RR$.
An \emph{action} is a mapping from $\Omega$ to~$X$.
A \emph{decision problem} is a finite set of actions.
A \emph{(state-dependent stochastic choice) dataset} is a collection $D$ of decision problems along with functions $P_A:\Omega\rightarrow\Delta(A)$ for every $A\in D$, denoting the observed action distribution of a decision maker in each realized state of the world.

Fix a \emph{prior belief} $\mu\in\Delta(\Omega)$ for the decision maker.
An \emph{information structure} is a distribution over a finite number of \emph{posteriors} (distributions over $\Omega$) for the decision maker whose average is~$\mu$.\footnote{\citet{caplin2015revealed} define an information structure as a mapping from $\Omega$ to distributions over posteriors that satisfiy Bayes' law w.r.t.\ the prior $\mu$, which leads also to a different definition of rationalization; these definitions are known to be equivalent.}
The set of all information structures is $\Pi$.
The utility from an action $a$ when the posterior is $\gamma$ is $g(a,\gamma)=\sum_{\omega\in \Omega}\gamma(\omega)u(a(\omega))$.
The \emph{gross payoff} from a decision problem $A$ using an information structure $\pi$ is $G(A,\pi)=\sum_{\gamma\in\supp\pi}\pi(\gamma)\max_{a\in A}g(a,\gamma)$.
An \emph{information cost function} is a mapping $K:\Pi\rightarrow\RR\cup\{\infty\}$ where not all costs are infinite.

Given a dataset $\inattentiondataset$, for every $a\in A$ by slight abuse of notation we write $P_A(a)=\sum_{\omega\in\Omega}\mu(\omega)P_A(a\mid\omega)$. For each $a\in A$ s.t.\ $P_A(a)>0$, we write $P'_A(\cdot\mid a)$ for the \emph{revealed posterior} associated with the action $a$ in the decision problem $A$, defined by $P'_A(\omega\mid a)=\frac{P_A(a\mid\omega)\mu(\omega)}{P_A(a)}$ for every $\omega\in\Omega$.

A dataset $\inattentiondataset$ is said to have \emph{costly information acquisition representation} if there exists a cost function $K$ such that for every problem $A\in D$:
\begin{enumerate}
\item
For each $a\in A$ with $P_A(a)>0$, it holds that $a\in\argmax_{b\in A}g(b,P'_A(\cdot\mid a))$.
\item
Letting $\pi_A$ be the \emph{revealed information structure} for $A$, i.e., the information structure such that for every $a\in A$ with $P_A(a)>0$, the probability that $\pi_A$ assigns to the posterior $P'_A(\cdot\mid a)$ is $P_A(a)$ (note that by definition, these probabilties sum up to $1$, so this information structure is well defined), it holds that $\pi_A \in \argmax_{\pi\in\Pi}\bigl(G(A,\pi)-K(\pi)\bigr)$.
\end{enumerate}

An information cost function $K$ is \emph{weakly monotone in information} if for every two information structures $\pi,\phi$ s.t.\ $\pi$ is a garbling of $\phi$, we have that $K(\phi)\ge K(\pi)$.
An information cost function $K$ is \emph{mixture feasible} if for any two information structures $\pi,\phi$ and for every $\alpha\in(0,1)$, we have that $K\bigl(\alpha\circ\pi+(1-\alpha)\circ\phi\bigr)\le\alpha K(\pi)+(1-\alpha)K(\phi)$, where $\alpha\circ\pi+(1-\alpha)\circ\phi$ is the \emph{mixture distribution} assigning to each posterior a probability equal to the $\alpha$-weighted average of the probabilities assigned to it by $\pi$ and $\phi$.
Finally, an information cost function $K$ is \emph{normalized} if $K(\mathbbm{1}_\mu)=0$, where $\mathbbm{1}_\mu$ is the information structure assigning probability $1$ to the the prior $\mu$ being the posterior.

In a seminal result, \citet{caplin2015revealed} showed  that satisfying 
\emph{No Improving Action Switches (NIAS)} and \emph{No Improving Attention Cycles (NIAC)}
is necessary and sufficient for having a costly information acquisition representation for a finite dataset---and furthermore that satisfying NIAS and NIAC is a sufficient condition for such a representation by a normalized, weakly monotone in information, and mixture feasible cost function. 
For our purposes, it suffices to note that whenever an infinite dataset satisfies NIAS and NIAC, so does any finite subset of it.

\begin{theorem}[\citealp{caplin2015revealed}]\label{inattention-finite}
A finite dataset $\inattentiondataset$ has a costly information acquisition representation if and only if it satisfies NIAS and NIAC.
Moreover, when NIAS and NIAC hold there exists a costly information acquisition representation function for the dataset that is weakly monotone in information, mixture feasible, and normalized.
\end{theorem}

Recently, \citet{caplin2017rationally} proved an infinite-dataset-size version of \cref{inattention-finite} via a novel proof that diverges from \citeauthor{caplin2015revealed}'s proof of the finite case. As we now show, this infinite version is also readily provable via \cref{scaling-lemma}, using \cref{inattention-finite} for establishing the finite-subset property, with a very similar well description to the one we used to prove \citeauthor{reny2015}'s infinite-data version of \citeauthor{afriat67}'s theorem. This is despite the differences between the two settings considered, and despite the fact that neither the original proofs of the finite versions nor the original proofs of the infinite versions of any of these quite different theorems share any common core technique.

\begin{theorem}[\citealp{caplin2017rationally}]\label{inattention-infinite}
A (possibly infinite) dataset $\inattentiondataset$ has a costly information acquisition representation if and only if it satisfies NIAS and NIAC.
Moreover, when NIAS and NIAC hold there exists a costly information acquisition representation function for the dataset that is weakly monotone in information, mixture feasible, and normalized.
\end{theorem}

\begin{proof}[Proof of \cref{inattention-infinite}]
As noted, the ``only if'' direction is trivial, so we prove the ``if'' direction. 
Let $\inattentiondataset$ be a dataset that satisfies NIAS and NIAC, and let $\pi_A$ be the associated revealed information structure.
Of the two conditions for having a costly information acquisition representation, the first condition does not involve $K$, and is in fact a condition on each decision problem in the data separately; the fact that it holds for each decision problem in the dataset separately follows from \cref{inattention-finite} since each decision problem in the dataset satisfies NIAS and NIAC as a dataset by itself, and therefore has a costly information acquisition representation by itself as the only data point therefore satisfies the first condition for having a costly information acquisition representation.\footnote{In fact, only NIAS is required here, but we are intentionally stating the proof in a way that is agnostic to the details of NIAS and NIAC.} The challenge with infinite data is therefore to find a cost function that satisfies the second condition for having a costly information acquisition representation, simultaneously for all infinitely many decision problems in the dataset. We do so using \cref{scaling-lemma}.

\proofstep{Definition of $\mathcal{P}$}
Fixing $X$, $\Omega$, and $u$, let $\mathcal{P}$ be the set of all datasets satisfying NIAS and NIAC. A \emph{solution} for a dataset~$\inattentiondataset\in\mathcal{P}$ is a costly information acquisition representation function for $\inattentiondataset$ that is weakly monotone in information, mixture feasible, and normalized.

\proofstep{Well describability}
We set $\varepsilon_n\eqdef2^{-n}$ for every $n\in\NN$.
We define a variable $\cost{n}{\pi}{c}$ for every $n\in\NN$, every information structure $\pi$, and every $c\in C_n\eqdef\{0,\varepsilon_n,2\cdot\varepsilon_n,\ldots,n\}$ (note that unlike in our proof of \cref{garp-infinite}, this set goes up to $n$ rather than only up to $1$, so both the fineness of the discretization and the upper bound depend on $n$ in this proof).
In what follows, for each $\inattentiondataset\in\mathcal{P}$ we define a set $\formulae{\inattentiondataset}$ of formulae over these variables so that models of $\formulae{\inattentiondataset}$ are in one-to-one correspondence with the (not-yet-proven-to-be-nonempty) set of solutions for $\inattentiondataset$.
The correspondence is obtained by endowing the variable $\cost{n}{\pi}{c}$ with the semantic interpretation ``$\cnfloor{K(\pi)}=c$ for the corresponding cost function $K$,'' where for every $n\in\NN$ and every $x\in\RR\cup\{\infty\}$ we denote by $\cnfloor{x}$ the rounding-down of~$x$ to the nearest number in $C_n$. (Note that in particular, $\cnfloor{x}=n$ for every $x>n$.)
We define the set $\formulae{\inattentiondataset}$ to consist of the following formulae (which, except for the sixth formula-type, mirror those from our proof of \cref{garp-infinite}):
\begin{enumerate}
\item
\begin{sloppypar}
for all $n\in\NN$ and all information structures $\pi$, the (finite!) formula
`$\bigvee_{c\in C_n}\cost{n}{\pi}{c}$',
requiring that $\pi$ have a rounded-down-to-$C_n$ cost;
\end{sloppypar}
\item
for all $n\in\NN$, all information structures $\pi$, and all distinct $c,d\in C_n$, the formula
`$\cost{n}{\pi}{c}\rightarrow\lnot\cost{n}{\pi}{d}$',
requiring that the rounded-down-to-$C_n$ cost of $\pi$ be unique;
\item
for all $n\in\NN$, all information structures $\pi$, and all $c\in C_n$, the following formula:
\begin{itemize}
\item
if $c<n$, the formula
`$\cost{n}{\pi}{c}\rightarrow\bigl(\cost{n+1}{\pi}{c}\vee\cost{n+1}{\pi}{c+\varepsilon_{n+1}}\bigr)$',
\item
if $c=n$, the (finite) formula
`$\cost{n}{\pi}{c}\rightarrow
\bigvee_{d=\{n,n+\varepsilon_{n+1},\ldots,n+1\}}
\cost{n+1}{\pi}{d}$',
\end{itemize}
requiring that $\cnfloor{K(\pi)}=\cnfloor{\cnnfloor{K(\pi)}}$;
\item
For all $n\in\NN$, all pairs of information structures $\pi,\phi$ s.t.\ $\pi$ is a garbling of $\phi$, and all $c\in C_n$, the (finite) formula
`$\cost{n}{\pi}{c}\rightarrow
\bigvee_{
d\in C_n:
d\ge c
}\cost{n}{\phi}{d}$',
requiring that the rounded-down-to-$C_n$ cost function be weakly monotone in information;
\item
\begin{sloppypar}
for all $n\in\NN$, all pairs of information strucures $\pi,\phi$, all $\alpha\in(0,1)$, and all $c,d\in C_n$, the (finite) formula
`$\bigl(\cost{n}{\pi}{c}\wedge\cost{n}{\phi}{d}\bigr)\rightarrow
\bigvee_{
c'\in C_n:
c'\le\alpha c+(1-\alpha)d+\varepsilon_n
}\cost{n}{\alpha\circ\pi+(1-\alpha)\circ\phi}{c'}$',
requiring that the rounded-down-to-$C_n$ cost function be mixture feasible, up to an error of at most~$\varepsilon_n$;
\end{sloppypar}
\item
for all $n\in\NN$, the formula
`$\cost{n}{\mathbbm{1}_\mu}{0}$',
requiring that the rounded-down-to-$C_n$ cost function be normalized (and, in particular, that not all costs are infinite);
\item
\begin{sloppypar}
for all decision problems $A\in D$, all $n\in\NN$, all information structures $\phi$, and all $c\in C_n\setminus\{n\}$, the (finite) formula
`$\cost{n}{\phi}{c}\rightarrow
\bigvee_{
d\in C_n:
d\le c-(G(A,\phi)-G(A,\pi_A))+\varepsilon_n
}\cost{n}{\pi_A}{d}$',
requiring that choosing $\pi_A$ maximizes gross payoff minus rounded-down-to-$C_n$ costs, up to an error of at most~$\varepsilon_n$.
\end{sloppypar}
\end{enumerate}
We now argue that $(\formulae{\inattentiondataset})_{\inattentiondataset\in\mathcal{P}}$ is a well description of~$\mathcal{P}$. Let $\inattentiondataset\in\mathcal{P}$.

We first claim that every model that satisfies $\formulae{\inattentiondataset}$ corresponds to a solution for $\inattentiondataset$. Fix a model for $\formulae{\inattentiondataset}$. For every information structure $\pi$ and every $n\in\NN$, let $c_n\in C_n$ be the value such that $\cost{n}{\pi}{c_n}$ is $\mathsf{True}$ in the model (well defined by the first and second formula-types above), and define $K(\pi)=\lim_{n\rightarrow\infty} c_n$ (well defined, e.g., by the third formula-type above since $c_n$ is monotone nondecreasing). The resulting cost function $K$ is a limit of normalized functions that are weakly monotone in information (by the fourth and sixth formula-types above) and that up to an error that tends to $0$, both (1)~are mixture feasible (by the fifth formula-type above), and, (2)~choosing $\pi_A$ for each $A\in D$ maximizes gross payoff minus costs according to them (by the seventh formula-type above). Hence, $K$ itself is normalized, weakly monotone in information, mixture feasible, and choosing $\pi_A$ for each $A\in D$ maximizes gross payoff minus costs according to it (i.e., is a costly information acquisition representation function for $\inattentiondataset$).

Second, if $\inattentiondataset$ has a solution $K$, then using it we can construct a model for $\formulae{\inattentiondataset}$ (by setting each $\cost{n}{\pi}{c}$ to be $\mathsf{True}$ iff $c=\cnfloor{K(\pi)}$), and so $\formulae{\inattentiondataset}$ has a model. To sum up, $(\formulae{\inattentiondataset})_{\inattentiondataset\in\mathcal{P}}$ is a well description of~$\mathcal{P}$.

\proofstep{Finite-subset property}
Let $\inattentiondataset\in\mathcal{P}$. Let $\Phi'\subset\formulae{\inattentiondataset}$ be a finite subset.
Since $\Phi'$ is finite, there are only finitely many formulae of the above seventh type (the only formula type that depends on the dataset) in $\Phi'$. For each of these formulae, select some decision problem in $D$ that induces it, and let $D'\subset D$ be the set of these (at most $|\Phi'|$) problems. By definition, $\Phi'\subseteq\formulae{\inattentiondatasetprime}$. Furthermore, $\inattentiondatasetprime$ satisfies NIAS and NIAC since any sub-dataset of $\inattentiondataset$ satisfies NIAS and NAIC, and hence, by \cref{inattention-finite}, $\inattentiondatasetprime$~has a solution. Therefore, $\inattentiondataset$~satisfies the finite-subset property. Thus, by \cref{scaling-lemma}, $\inattentiondataset$~has a solution.
\end{proof}

\section{Existence of Man-Optimal Stable Matching}\label{menopt}

In this \lcnamecref{menopt}, we demonstrate how our proof of \cref{sp-infinite} is quite a bit more flexible than one might imagine, by using essentially the same well description with a similar argument to scale \cref{menopt-finite} from finite to infinite markets.
The main complexities of our proof are analogous to the complexities of the proof of \cref{sp-infinite} that we discussed before that proof:
\begin{itemize}
\item
The well description constructed for the infinite problem may not be the most straightforward way to model the object whose existence we wish to prove, i.e., it is not ``a stable matching that matches each man to his favorite stable partner'' but rather ``a stable matching that matches each man to a woman at least as preferred as his favorite stable partner.'' While a well description for the more straightforward statement of the problem exists (and is in fact slightly simpler), this (equivalent) well description allows us to use additional results from the literature on finite matching markets in critical steps of the proof of the finite-subset property.
\item
The solution defined for the smaller (finite) problems is not what one would expect as ``the analogous stand-alone problem but in the smaller market,'' i.e., it is \emph{not} ``existence of a man-optimal stable matching in the smaller market'' but rather a problem defined relying on the specifics of the of the infinite problem, i.e., ``existence of a stable matching in the smaller market that matches each man to a woman at least as preferred as his favorite stable partner \emph{in the original, infinite market}.''
\item
The argument more carefully chooses the finite market to be used for proving the finite-subset property, i.e., it is not ``the market of all participants mentioned in any formula in the given finite subset of formulae,'' but rather a carefully chosen larger (yet still finite) market, chosen so that certain specific facts about the infinite market can be projected down to the finite market.
\end{itemize}

\begin{proof}[Proof of \cref{menopt-infinite}]
Let $\simplemarket$ be a matching market. Let $\widehat{M}$ be the subset of men who are matched in at least one stable matching in this market. For each $m\in\widehat{M}$, let $w^m$ be the woman most preferred by $m$ of all women to which he is matched in at least one stable matching in this market. We wish to show that there exists a stable matching in this market in which each $m\in\widehat{M}$ is matched to $w^m$.  We  use \cref{scaling-lemma} to prove an equivalent statement: that each $m\in\widehat{M}$ is matched to a woman he weakly prefers to $w^m$.

\proofstep{Definition of $\mathcal{P}$}
We use the same set of problems as in the proof of \cref{stable-infinite}, i.e., $\mathcal{P}\eqdef\bigl\{(M',W')~\big|~M'\subseteq M \And W'\subseteq W\bigr\}$. The definition of a solution will be different, though: A \emph{solution} for $(M',W')\in\mathcal{P}$ is a stable matching between $M'$ and $W'$ (in the induced submarket of $\simplemarket$) in which each $m\in M'\cap\widehat{M}$ is matched to a woman he weakly prefers to $w^m$, the woman to whom he is matched in his most-preferred stable matching of the entire market.

\proofstep{Well describability}
We use the same set of variables as in the proof of \cref{stable-infinite}, with the same semantic interpretation, however the well description is different.
For each $(M',W')\in\mathcal{P}$, we define the set $\formulae{(M',W')}$ to consist of the same formulae as in the proof of \cref{stable-infinite}, and in addition:
\begin{enumerate}
    \item[5.] for all $m\in M'\cap\widehat{M}$, the (finite!) formula
`$\bigvee_{w\succcurlyeq_m w^m}\matched{m}{w}$',
requiring that $m$ be matched to a woman he weakly prefers to $w^m$.
\end{enumerate}
By construction, models of $\formulae{(M',W')}$ (over the variables that appear in $\formulae{(M',W')}$) are in one-to-one correspondence with the (not-yet-proven-to-be-nonempty) set of solutions of $(M',W')$, and so $(\formulae{(M',W')})_{(M',W')\in\mathcal{P}}$ is a well description of~$\mathcal{P}$.

\proofstep{Finite-subset property}
Let $\Phi'\subset\formulae{(M,W)}$ be a finite subset.
Since $\Phi'$ is finite, it ``mentions'' (through variables used) only finitely many men and women; denote the sets of these men and these women by $M'\subset M$ and $W'\subset W$, respectively. As in the proof of \cref{sp-infinite}, the problem that we use to show that the finite-subset property is satisfied is not the problem $(M',W')$. Rather, in this proof it is $(M',W'')$ for $W''\eqdef W'\cup\bigcup_{m\in M'\cap\widehat{M}}\mu^{m}(M')$, where for every $m\in M'\cap\widehat{M}$ we denote by $\mu^{m}$ some stable matching of the entire market $(M,W)$ that matches $m$ to $w^{m}$. By definition, $\Phi'\subseteq\formulae{(M',W'')}$. By \cref{menopt-finite}, there exists a man-optimal stable matching between $M'$ and $W''$. To show that this matching constitutes a solution for $(M',W'')$, it remains to show that it matches each $m\in M'\cap\widehat{M}$ to a woman he weakly prefers to $w^m$.\footnote{As in the proof of \cref{sp-infinite}, the argument that we give for this would not have shown that the man-optimal stable matching in the market $(M',W')$ is a solution for $(M',W')$.}

Fix $m\in M'\cap\widehat{M}$. We have that $M'\subset M' \cup \mu^{m}(W')$ and $W''\supset W'\cup \mu^{m}(M')$. Theorem 2.25 in \citeN{RothSotomayor1990} thus assures that the man-optimal stable matching in $(M',W'')$ is weakly preferred by all men in $M'$, and in particular by $m$, to any stable matching in $\bigl(M'\cup\mu^{m}(W'),W'\cup\mu^{m}(M')\bigr)$. But one stable matching in the latter market is  the restriction of $\mu^{m}$ to this market, which matches $m$ with $w^{m}$ (it is stable since any blocking pair would also block the matching $\mu^{m}$ in the full market). Thus, the man-optimal stable matching in $(M',W'')$ matches $m$ to a woman he weakly prefers to $w^{m}$,\footnote{This argument, in particular, does not show that the man-optimal stable matching in $(M',W'')$ matches $m$ with $w^{m}$. This is the reason for using `$\bigvee_{w\succcurlyeq_m w^m}\matched{m}{w}$' rather than `$\matched{m}{w^m}$' as the additional formula type of our well description.} and this holds for every man  $m\in M'\cap\widehat{M}$, and hence this matching is a solution for $(M',W'')$.
Therefore, $(M,W)$~satisfies the finite-subset property. Thus, by \cref{scaling-lemma}, there exists a stable matching between $M$ and~$W$ in which each man $m\in M'\cap\widehat{M}$ is matched to a woman he weakly prefers to $w^m$, as required.
\end{proof}

\section{Infinite Time: Stable Matching with a Doubly Infinite Horizon}\label{dynamic}

In this \lcnamecref{dynamic} we use \cref{scaling-lemma} to prove the existence of stable matchings in dynamic, infinite-size, and infinite-horizon markets, generalizing a dynamic stable matching model of \citeN{pereyra2013dynamic} in which time is bounded from below. For simplicity, we formulate the dynamic setting with one-to-one matching (in each period).
Like all of the other models in this paper, the dynamic model we consider is motivated by an established framework---in this case, the model of teachers-to-schools assignment with tenure constraints introduced by \citeN{pereyra2013dynamic}. For consistency with our other matching  applications, here we speak of the agents as ``men'' and ``women,'' even though they are really stand-ins for ``teachers'' and ``schools.''

\subsection{Setting}
\begin{sloppypar}
A dynamic matching market is a tuple \dynamicmarket,
where \simplemarket\ is a (possibly infinite) matching market as in \cref{matching-warmup}, and where for each $m\in M$, we have that $a_m,d_m\in\ZZ\cup\{-\infty,\infty)$ such that $a_m<d_m$, respectively called the \emph{arrival time} and \emph{departure time} of $m$. For each $m\in M$, we say that $m$ is \emph{on the market} at all times $t\in\ZZ$ such that $a_m\le t<d_m$. (All $w\in W$ are considered to always be on the market.) A \emph{matching chronology} in a dynamic matching market is a mapping from woman-time pairs to men who are on the market at the relevant time, such that each man is matched to at most one woman at any given time. We say that a matching chronology is \emph{stable subject to tenure} if:
\begin{itemize}
\item
\textit{Men have tenure}: for every time $t\in\ZZ$, every man who is on the market both at time $t$ and at time $t+1$, weakly prefers his match at time $t+1$ to his match at time $t$.
\item
The matching is \textit{otherwise stable}: at any time $t\in\ZZ$, there exists no pair of man $m$ who is on the market at $t$ and woman $w$ such that $m$ strictly prefers $w$ to his match at $t$, and $w$ strictly prefers $m$ to her match at time $t$ who is furthermore not her match at time $t-1$.\footnote{Note that, following \citeN{pereyra2013dynamic}, we implicitly assume that agents' preferences over partners are consistent over time (unlike in, e.g., the framework of \citealp{kadam2018multiperiod}). Additionally, again following \citeN{pereyra2013dynamic}, we enforce stability myopically (unlike in the frameworks of \citealp{doval2014theory,liu,aliliu}).}
\end{itemize}
\end{sloppypar}

Our dynamic setting builds on the model of \citeN{pereyra2013dynamic}, in which arrival times are required to be nonnegative. If we were to restrict ourselves to nonnegative arrival times (or more generally, to arrival times that have a finite lower bound), then, following \citeN{pereyra2013dynamic}, a simple iterative application of the man-optimal stable matching mechanism would find a stable-subject-to-tenure matching chronology:
\begin{enumerate}
\item
As the matching at time~$0$,
use the man-optimal stable matching
for all women and all men with arrival time~$0$.
\item
For every $t>0$, as the matching at time~$t$,
use the man-optimal stable matching
for all women and all men who are on the market at time~$t$, with respect to slightly modified preferences: any man who is matched at time~$t-1$ and still on the market at time~$t$ is promoted (for the purposes of
finding the man-optimal stable
matching at time~$t$) to be top-ranked on the preferences of his match at time~$t-1$.
\end{enumerate}
The preceding argument in fact constitutes a full proof of the following.

\begin{theorem}[\citealp{pereyra2013dynamic}]\label{dynamic-existence-finite}
In any dynamic matching market where all arrival times are nonnegative, a stable-subject-to-tenure matching chronology exists.
\end{theorem}

Our analysis from \cref{matching-warmup} suffices to immediately scale \cref{dynamic-existence-finite} to infinite markets (via the same iterative process of finding successive man-optimal stable matchings). In this \lcnamecref{dynamic}, however, we scale this result in a different way: making time, rather than market size, infinite---or more precisely, doubly infinite.

\subsection{Challenge} 

As we have just seen, the proof of \cref{dynamic-existence-finite} depends heavily on our ability to identify a ``starting'' matching that we can adjust/build off of in subsequent time periods. However, the need to assume a fixed start time makes the model less representative of a steady-state. 

What if  arrival times  have no finite lower bound? Due to symmetry considerations, there can be no ``reasonable'' deterministic variant of the man-optimal stable matching mechanism that reaches a stable-subject-to-tenure marriage chronology:
\begin{example}\label{two-symmetric}
Consider a case of one woman $w$ and an  infinite set of men $m_t$ such that for each $t\in\ZZ$, a man $m_t$ has arrival time $t$ and departure time $t+2$. For any profile of preference lists in which no agent finds any other agent unacceptable, there are precisely two stable-subject-to-tenure matchings chronologies:
\begin{itemize}
\item
All men with even arrival times are matched to $w$ throughout their time on the market; all men with odd arrival times are never matched.
\item
All men with odd arrival times are matched to $w$ throughout their time on the market; all men with even arrival times are never matched.
\end{itemize}
\end{example}
Which of the two preceding stable-subject-to-tenure matching chronologies (none more or less ``man-optimal'' than the other) would a deterministic
variant of the man-optimal stable matching mechanism,
if one existed, choose in the setting of \cref{two-symmetric}? The only way to break the symmetry and choose between these two is to give special treatment to some specific time period, such as $0$---but it is easy to see that just picking some finite time $t$, matching $m_t$ with $w$, and solving forward and (somehow) backward would not work since the removal of even a single man, say with very small (negative) arrival time $t'$, from the market would collapse the two stable matching chronologies starting at $t'+1$ (with $m_{t'+1}$ being matched in any stable matching).

In the absence of a reasonable variant of
the man-optimal stable matching mechanism, we need a new way to prove existence when arrival times are unbounded. To resolve this problem, we turn again to \cref{scaling-lemma}.

\subsection{Existence} To prove the existence of a stable-subject-to-tenure matching chronology in the infinite-history model, we require a mild ``local finiteness'' condition (but not local boundedness).
This condition holds, for example, when in each $t$ only finitely many men are on the market. 
We later show that this condition cannot be dropped.

\begin{definition}
Let \dynamicmarket\ be a dynamic matching market. If for every time $t\in\ZZ$, only finitely many men are on the market at both $t$ and $t+1$, then we say that the dynamic matching market has \emph{finite presence}.
\end{definition}

\begin{theorem}\label{dynamic-existence}
In any dynamic matching market (with arbitrary arrival and departure times) that has finite presence, a stable-subject-to-tenure matching chronology exists.
\end{theorem}

Before proving \cref{dynamic-existence}, we note that the finite presence condition in that \lcnamecref{dynamic-existence} cannot be dropped. 

\begin{example}It is straightforward to verify that  in a dynamic market with one woman $w$ and countably many men $\{m_t\}_{t\in\NN}$, with $a_{m_t}=-t$ and $d_{m_t}=0$ for every $t\in\NN$, no stable-subject-to-tenure matching chronology exists if all participants find all possible partners acceptable.
\end{example}

\begin{proof}[Proof of \cref{dynamic-existence}]
We prove the theorem using \cref{scaling-lemma}.

\proofstep{Definition of $\mathcal{P}$}
Let $\dynamicmarket$ be a dynamic matching market. Let $\mathcal{P}\eqdef\bigl\{(M',W',t_0)~\big|~M'\subseteq M \And W'\subseteq W \And t_0\in\ZZ\cup\{-\infty\} \bigr\}$ be the set of all triplets of sets of men, sets of women, and ``start times.'' A \emph{solution} for $(M',W',t_0)\in\mathcal{P}$ is a stable-subject-to-tenure matching chronology between $M'$ and $W'$ (in the induced submarket of $\dynamicmarket$) that starts at time $t_0$ (i.e., setting all earlier arrival times to $t_0$, dropping all men for whom $d_m\le t_0$, and starting with any arbitrary individually rational matching at time $t_0$).

\proofstep{Well describability}
We define a variable $\tmatched{m}{w}{t}$ for every $(m,w)\in M\times W$ and time $t\in\ZZ$.
In what follows, for each $(M',W')\in\mathcal{P}$ we define a set $\formulae{(M',W',t_0)}$ of formulae over these variables so that models of $\formulae{(M',W',t_0)}$ are in one-to-one correspondence with the (not-yet-proven-to-be-nonempty) set of stable-subject-to-tenure matching chronologies between $M'$ and $W'$ that start at $t_0$.
The correspondence is obtained by endowing the variable $\tmatched{m}{w}{t}$ with the semantic interpretation ``$m$ and $w$ are matched at time $t$.'' That is, it maps a model for $\formulae{(M',W',t_0)}$ to the matching chronology such that for every $(m,w)\in M'\times W'$ and for every $t\in\ZZ$, we have that $m$ and $w$ are matched at time $t$ if and only if the variable $\tmatched{m}{w}{t}$ is $\mathsf{True}$ in that model.
We define the set $\formulae{(M',W',t_0)}$ to consist of the following formulae:
\begin{enumerate}
\item
\begin{sloppypar}
for all $m\in M'$, all $t\ge t_0$, and all distinct $w,w'\in W'$, the formula
`$\tmatched{m}{w}{t}\rightarrow\lnot\tmatched{m}{w'}{t}$',
requiring that man $m$ be matched to at most one woman at time $t$;
\end{sloppypar}
\item
\begin{sloppypar}
for all $w\in W'$, all $t\ge t_0$, and all distinct $m,m'\in M'$, the formula
`$\tmatched{m}{w}{t}\rightarrow\lnot\tmatched{m'}{w}{t}$',
requiring that woman $w$ be matched to at most one man at time $t$;
\end{sloppypar}
\item
for all $w\in W'$, all $t\ge t_0$, and all $m\in M'$ such that either $m$ is not on the market at $t$, or $m$ and $w$ are not both acceptable to each other, the formula
`$\lnot\tmatched{m}{w}{t}$',
requiring that no one be matched to someone who is unacceptable to them, and that men not be matched when they are not on the market;
\item
for all $t>t_0$, all $w\in W'$, and all $m\in M'$ who are on the market at time $t$ and are acceptable to each other, the (finite!) formula
`$\lnot\tmatched{m}{w}{t}\rightarrow\bigl(\bigvee_{w'\succ_{m}w}\tmatched{m}{w'}{t}\vee\bigvee_{m'\succ_{w}m}(\tmatched{m'}{w}{t}\wedge\lnot\tmatched{m}{w}{t-1})\vee\linebreak\bigvee_{m':a_{m'}< t<d_{m'}}(\tmatched{m'}{w}{t}\wedge\tmatched{m'}{w}{t-1})\bigr)$',
requiring that $(m,w)$ not be a blocking pair at time $t$. The first two disjunctions are finite since $m$ (resp.\ $w$) prefers only $k$ partners over his (resp.\ her) $k$-th ranked partner, and the last disjunction is finite by finite presence.
\end{enumerate}
\begin{sloppypar}
By construction, the models (over the variables that appear in $\formulae{(M',W',t_0)}$) of $\formulae{(M',W',t_0)}$ are in one-to-one correspondence with stable-subject-to-tenure matchings chronologies between $M'$ and $W'$ that start at $t_0$. Therefore, $(\formulae{(M',W',t_0)})_{(M',W',t_0)\in\mathcal{P}}$ is a well description of~$\mathcal{P}$.
\end{sloppypar}

\proofstep{Finite-subset property}
Let $\Phi'\subset\formulae{(M,W,-\infty)}$ be a finite subset.
Since $\Phi'$ is finite, it ``mentions'' (through variables used) only finitely many men, women, and times; denote the set of these men by $M'\subset M$, the set of these women by $W'\subset W$, and the minimum such time by~$t_0$.
By definition, $\Phi'\subseteq\formulae{(M',W',t_0)}$. By \cref{dynamic-existence-finite}, using $t_0$ as ``period 0,'' there exists a stable-subject-to-tenure matching chronology between $M'$ and $W'$ that starts at~$t_0$ (this chronology is furthermore stable in the usual sense at $t_0$, but our proof does not require this). Therefore, $(M,W,-\infty)$~satisfies the finite-subset property. Thus, by \cref{scaling-lemma}, there exists a stable-subject-to-tenure matching chronology between $M$ and~$W$ (``starting at $-\infty$''), as required.
\end{proof}

\section{Nash Equilibria in Games on Infinite Graphs}\label{nash}

In this \lcnamecref{nash}, we turn to a different setting---games on graphs \cite[see, e.g.,][and the references therein]{kearns2007graphicalgames}---which includes overlapping-generations models, even with doubly infinite time. 
We use \cref{scaling-lemma} to show the existence of a Nash equilibrium in games on infinite graphs. Our result here is covered by \citeN{Peleg1969} (who directly scales the seminal existence result of \citealp{Nash1951}), but we give a new proof that uses the same principled approach we use throughout this paper.

Here, we use \cref{scaling-lemma} to scale the existence of arbitrarily good approximate Nash equilibria, and then show that the existence of such approximate equilibria implies the existence of an exact equilibrium. This two-step proof strategy is chosen for convenience: with additional variables, it is easy to encode the second step of the proof into the logical formulation just like we did in \cref{walrasian}.

In a \emph{game on a graph}, there is a (potentially infinite) set of players $I$, each having a finite set of pure \emph{strategies} $S_i$. 
 Each player $i\in I$ is linked to a finite set of neighbors $N(i) \subset I$  with $i\in N(i)$, and her utility only depends on the strategies played by players in the set $N(i)$.\footnote{Readers familiar with \citeN{Peleg1969} will note that even on graphs, Peleg's assumptions are weaker than those stated here. Our analysis can be generalized to cover such weaker assumptions, and that our assumptions in other sections can also be similarly weakened. Nonetheless, in general, throughout in this paper we prefer ease and clarity of exposition over tightening assumptions (as noted in the Introduction, we consider the results that we present to be minimal working examples), as our goal is to introduce a unified, transparent technique.} This setting occurs, for example, in infinite-horizon overlapping-generations models, where at each point in time there are only finitely many players alive, and a player's utility depends only on the behavior of contemporary players. For any player $i$ we denote by $\Sigma_i\eqdef\Delta(S_i)$ the set of \emph{mixed strategies} (i.e., distributions over pure strategies) of player $i$. A \emph{mixed-strategy profile} $(\sigma_i)_{i\in I}$ is a specification of a mixed strategy $\sigma_i\in\Sigma_i$ for every player $i\in I$. A mixed-strategy profile $(\sigma_i)_{i\in I}$ is a \emph{Nash equilibrium} if for every $i\in I$ and every possible deviating strategy $\sigma'_i\in\Sigma_i$, it holds that $u_i(\sigma_{N(i)})\geq u_i(\sigma'_i,\sigma_{N(i)\setminus \{i \} })$.

Games on \emph{finite} graphs have finitely many players and finitely many strategies per player; hence, the seminal analysis of \citeN{Nash1951} implies that they have Nash equilibria.

\begin{theorem}[Follows from \citealp{Nash1951}]\label{nash-finite}
Every game on a finite graph has a Nash equilibrium. 
\end{theorem}

Our main result of this \lcnamecref{nash} is that Nash equilibria are guaranteed to exist even in games on infinite graphs.

\begin{theorem}[Follows from \citealp{Peleg1969}]\label{exact-nash}
Every game on a (possibly infinite) graph has a Nash equilibrium. 
\end{theorem}

As already noted, we prove \cref{exact-nash} by first using \cref{scaling-lemma} to prove the existence of arbitrarily good approximate Nash equilibria, and then showing that the existence of such approximate Nash equilibria implies \cref{exact-nash}. For a given $\varepsilon>0$, a mixed-strategy profile $(\sigma_i)_{i\in I}$ is an \emph{$\varepsilon$-Nash equilibrium} if for every $i\in I$ and every possible deviating strategy $\sigma'_i\in\Sigma_i$, it holds that $u_i(\sigma_{N(i)})\geq u_i(\sigma'_i,\sigma_{N(i)\setminus \{i \}})-\varepsilon.$

\begin{lemma}\label{eps-nash}
For any $\varepsilon>0$, every (possibly infinite) game on a graph has an $\varepsilon$-Nash equilibrium. 
\end{lemma}
    
\begin{proof}
Let $\varepsilon>0$. For each player $i\in I$, the space of profiles of mixed-strategies of players in $N(i)$ is a compact metric space. Specifically, for this proof it is convenient to consider the space of profiles of mixed-strategies as a metric space with respect to the $\ell^\infty$ metric;\footnote{By equivalence of all norms on $\RR^n$,  the space of profiles of mixed-strategies is also compact with respect to the $\ell^\infty$ metric.} as each player~$i$ has a continuous utility function whose domain is this compact metric space, players' utility functions are uniformly continuous by the Heine--Cantor theorem. Thus, there exists $\hat{\delta}_i>0$ that assures that if two profiles of  mixed strategies of players in $N(i)$ are less than $\hat{\delta}_i$ apart, then the utilities they yield to  $i$ differs by no more than $\nicefrac{\varepsilon}{2}$.

For each player~$i$, choose $\delta_i\eqdef\min\bigl\{ \hat{\delta}_j ~\big|~ j\in N(i)\bigr\}>0$. Recall that $\Sigma_i$ denotes the space of player~$i$'s mixed strategies, and let $\Sigma^{\delta_i}_{i}\subset \Sigma_i$ be a finite set of strategies that includes all of $i$'s pure strategies, and includes for any mixed strategy in $\Sigma_i$ a strategy that is at most $\delta_i$ away from it; such a set exists by the compactness of $\Sigma_i$. We prove the lemma by proving that the given game admits an $\varepsilon$-Nash equilirbium in which each player $i$ plays a strategy in $\Sigma^{\delta_i}_{i}$. We prove this using \cref{scaling-lemma}.

\proofstep{Definition of $\mathcal{P}$}
Let $\mathcal{P}$ be all the subsets of $I$. A \emph{solution} for $I'\in\mathcal{P}$ is a strategy profile for $I$ that is a $\varepsilon$-Nash equilibrium in the induced game between all players in $I'$ (where all other ``players'' play any arbitrary strategy), in which each player $i\in I$ plays a strategy in $\Sigma^{\delta_i}_{i}$.

\proofstep{Well describability}
We define a variable $\plays{i}{\sigma_i}$ for every player $i\in I$ and discretized strategy $\sigma_i\in\Sigma^{\delta_i}_{i}$.
In what follows, for each $I'\in\mathcal{P}$ we define a set $\formulae{I'}$ of formulae over these variables so that models of $\formulae{I'}$ are in one-to-one correspondence with the (not-yet-proven-to-be-nonempty) set of solutions for $I'$.
The correspondence is obtained by endowing the variable $\plays{i}{\sigma_i}$ with the semantic interpretation ``$i$ plays the strategy $\sigma_i$.'' That is, it maps a model for $\formulae{I}$ to the strategy profile such that for every $i\in I'$, we have that $i$ plays the strategy $\sigma_i$ if and only if the variable $\plays{i}{\sigma_i}$ is $\mathsf{True}$ in that model.
For every player~$i$ and every profile $\sigma_{N(i)\setminus\{i\}}$ of mixed-strategies for $N(i)\setminus\{i\}$, we define the set of $\varepsilon$-best responses of $i$:
\[
\text{BR}^{\varepsilon}_i (\sigma_{N(i)\setminus\{i\}})\eqdef\Bigl\{\sigma_i ~\Big|~ u_i (\sigma_i,\sigma_{N(i)\setminus\{i\}})\geq \underset{\sigma'_i \in \Sigma_i}{\max}\bigl\{u_i (\sigma'_i,\sigma_{N(i)\setminus\{i\}})\bigr\}-\varepsilon\Bigr\}.
\]
We define the set $\formulae{I'}$ to consist of the following formulae:
\begin{enumerate}
\item
for all $i\in I$, the (finite!) formula
`$\bigvee_{\sigma\in\Sigma^{\delta_i}_i}\plays{i}{\sigma}$',
requiring that $i$ plays some (discretized) strategy (this formula is finite because $\Sigma^{\delta_i}_i$ is);
\item
for all $i\in I$ and all distinct $\sigma_i,\sigma'_i \in \Sigma^{\delta_i}_i $, the formula
`$\plays{i}{\sigma_i} \rightarrow \lnot \plays{i}{\sigma'_i}$',
requiring that the strategy that player $i$ plays be unique;
\item
for all $i\in I'$ and all profiles $\sigma=(\sigma_j)_{j\in N(i)\setminus\{i\}}\in\bigtimes_{j\in N(i)\setminus\{i\}}\Sigma^{\delta_i}_{i}$ of discretized mixed strategies of $N(i)\setminus\{i\}$, the (finite!) formula
\[
\biggl(\smashoperator[r]{\bigwedge_{j\in N(i)\setminus\{i\}}}\plays{j}{\sigma_j}\biggr) \rightarrow \biggl(\smashoperator[r]{\bigvee_{\sigma_i \in \Sigma^{\delta_i}_i \cap \text{BR}^{\varepsilon}_i (\sigma)}}\plays{i}{\sigma_i}\biggr),
\]
requiring that $i$ $\varepsilon$-best-responds to the strategies played by the other players.
\end{enumerate}
By construction, $(\formulae{I'})_{I'\in\mathcal{P}}$ is a well description of~$\mathcal{P}$.

\proofstep{Finite-subset property}
Let $\Phi'\subset\formulae{I}$ be a finite subset.
Since $\Phi'$ is finite, it ``mentions'' (through variables used) only finitely many players; denote the set of these players by $I'\subset I$.
By definition, $\Phi'\subseteq\formulae{I'}$. Consider the induced game on $I'$ obtained by having each player $i\in I\setminus I'$ mechanically play some fixed strategy in $\Sigma^{\delta_i}_{i}$. By \cref{nash-finite}, this game has a Nash equilibrium. By choosing for each player $i\in I'$ a closest strategy in $\Sigma_i^{\delta_i}$ to the one she plays at this Nash equilibrium, each player's utility changes by at most $\nicefrac{\varepsilon}{2}$ (by uniform continuity), and so does the utility attainable by best responding. Therefore, since we started with a Nash equilibrium, it is assured that each player is now playing an $\varepsilon$-best response, so the resulting strategy profile is a solution to $I'$.
Therefore, $I$~satisfies the finite-subset property. Thus, by \cref{scaling-lemma}, there exists an $\varepsilon$-Nash equilibrium in the grand game (among all player in $I$), as required.
\end{proof}

Now, we can use \cref{eps-nash} to prove \cref{exact-nash} by way of a ``diagonalization'' argument.
\begin{proof}[Proof of \cref{exact-nash}]
Since each player in the graph has finitely many neighbors, every connected component of the graph consists of at most countably many players. As it is enough to show the existence of a Nash equilibrium in each connected component separately (we use the Axiom of Choice here),
let us focus on one connected component.
By \cref{eps-nash} there exists a sequence $(\sigma^n)_{n=1}^{\infty}$ of $\frac{1}{n}$-Nash equilibria in the game on this connected component.
 Since each of the at-most-countably-many coordinates of each element in this sequence lies in $[0,1]$, we can choose a subsequence (a ``diagonal subsequence'') that converges in all coordinates; let $\sigma^{*} $ denote the limit of that subsequence.

We claim that $\sigma^*$ is a Nash equilibrium. To see this, note that for every $i\in I$ and $\sigma'_i\in\Sigma_i$, we have for the $n$th elements of the sequence that
\[
u_i(\sigma^n_{N(i)})\geq u_i(\sigma'_i,\sigma^n_{N(i)\setminus \{i \} })-\tfrac{1}{n}.
\]
By the continuity of $u_i$, this means that for every $i\in I$ and $\sigma'_i\in\Sigma_i$, we have
\[
u_i(\sigma^*_{N(i)})\geq u_i(\sigma'_i,\sigma^*_{N(i)\setminus \{i \} }),
\]
so no player has a profitable deviation under the profile $\sigma^{*}$. Hence, $\sigma^{*}$ is indeed a Nash equilibrium---and in particular, we see that a Nash equilibrium exists in the game, as desired.
\end{proof}

\section{Related Approaches}\label{related-methodology}

To our knowledge, we are the first to use Propositional Logic as a general tool for scaling results in economics. It is worth mentioning within this context, though, the work of \citet{holzman1984extension}, who used Logical Compactness to relax topological conditions in \citet{fishburn1984comment}.\footnote{Logical Compactness is frequently used to scale existence results in mathematics from finite settings to infinite ones (See, e.g.,  \citealp{bruijn1951colour} and  \citealp{halmos2009marriage}). }

Our approach was stated using Propositional logic, but, in fact, \cref{scaling-lemma} generalizes to well descriptions using First-Order Logic as well.
 Propositional logic is a special case of First-Order Logic. Importantly, it does not use quantifiers (i.e., $\forall$ and $\exists$). We chose to focus on this special case in order to simplify the exposition, since we were not able to identify any economic application in which the added generality would be beneficial.\footnote{Well-describing economic problems using the full generality of First-Order Logic is Challenging. For example,  fixing sets of objects (e.g., men and women) is not straightforward. In fact, by the (upward) L{\"o}wenheim--Skolem theorem, if a first-order theory has an infinite model (a model with an infinite domain) then it has a model of any larger cardinality, which implies that first-order theories cannot bound the cardinality of their infinite models. Hence, constants would have to play an important role in the well description.}   
 
 Other papers have used First-Order (rather than Propositional) Logic and nonstandard analysis to unify, refine, and scale results in economic theory. Examples include  \citet{anderson1978elementary}, \citet{brown1980extension}, \citet{anderson1991non}, \citet{khan1993irony}, \citet{blume1994algebraic}, \citet{halpern2009nonstandard}, and \citet{halpern2016characterizing}. \citet{CES} used Compactness in First-Order Logic to formalize the notion of the empirical content of a model. Like us, that paper studies applications to revealed preferences theory (see also \citealp{CES2}), however it deals with  different questions from us, and uses different techniques.

\citeN{HellmanLevy2019} use (still different) tools from mathematical logic to prove  conceptually related, yet incomparable,  results: while our paper scales certain finite results to infinite settings, their paper scales certain countably infinite results to \emph{un}countably infinite settings.
Specifically, they give sufficient conditions to scale certain existence results that are known to hold whenever there are countably many possible \emph{states of the world} into scenarios with uncountably many possible states of the world.
Their results are incomparable to any of our results, and even to our existence-in-large-market results, first because they always assume that the number of agents is finite (an infinite number of agents, even with only two possible types for each, would already result in an uncountably infinite set of possible states of the world to begin with), and second, because they require that the theorems that they scale be already known to hold for the countably infinite, rather than only the finite, case.

\begin{sloppypar}
We have been asked about the relation to various theorems in topology. \cref{scaling-lemma} is stated in terms of logical propositions and its proof relies on Logical Compactness. 
Logical propositions can be translated into closed sets in an application-specific topological (product) space, in which setting Logical Compactness follows from Tychonoff's theorem on topological compactness. 
In other words, \cref{scaling-lemma} can be proved using Tychonoff's theorem, and its statement can be translated to the language of topology.  
 However, in our view, the resulting Lemma would be harder to directly formulate and the conditions would be harder to verify. And while topological compactness or the language of nets are stronger and more general approaches, in the domains we study, they often introduce technical issues that can render arguments incorrect in subtle ways (e.g., matchings may converge to an object that is not a matching).
We therefore view the methodological part of our contribution as introducing a unifying approach that is simple and  intuitive to work with, and that does not require us to look for the ``right'' topological space or  apply topological  reasoning directly.\footnote{Once a proof is derived using \cref{scaling-lemma}, it is of course possible to then translate it to a topological statement and attempt to achieve greater generality, if/when such generality is of interest.} 
\end{sloppypar}

\section{Proofs Omitted from Section~\ref{revealed-prefs}}

\subsection{Limited Attention Rationalizability}\label{app:warpla}

\begin{proof}[Proof of \cref{InfiniteLA}]
As with \cref{szpilrajn}, the ``only if'' direction is trivial, so we prove the ``if'' direction. We do so using \cref{scaling-lemma}.

\proofstep{Definition of $\mathcal{P}$}
Fixing $X$, let $\mathcal{P}$ be the set of all pairs $(X',\dataset)$ such that $X'\subseteq X$ and $\dataset$ is a full dataset over $X'$ that satisfies WARP-LA. A \emph{solution} for a pair~$(X',\dataset)\in\mathcal{P}$ is a pair of strict preference order over $X'$ and attention filter that rationalize~$\dataset$.

\proofstep{Well describability}
We define a variable $\gt{a}{b}$ for every pair of distinct $a,b\in X$ and a variable $\attention{S}{T}$ for each pair of finite sets $S,T$ s.t.\ $\emptyset\ne T\subseteq S\subseteq X$.
In what follows, for each $(X',\dataset)\in\mathcal{P}$ we define a set $\formulae{(X',\dataset)}$ of formulae over these variables so that models of $\formulae{(X',\dataset)}$ are in one-to-one correspondence with the (not-yet-proven-to-be-nonempty) set of solutions for $(X',\dataset)$.
The correspondence is obtained by endowing the variable $\gt{a}{b}$ with the semantic interpretation ``$a$ is preferred to~$b$ (when both are attention attracting),'' and the variable $\attention{S}{T}$ with the semantic interpretation ``$T$ is the set of attention-attracting elements when the menu is $S$.''
That is, it maps a model for $\formulae{(X',\dataset)}$ to the preference~$\succ$ such that for every distinct $a,b\in X'$, we have that $a\succ b$ if and only if the variable $\gt{a}{b}$ is $\mathsf{True}$ in that model and to the attention filter $\Gamma$ such that for every $S,T$ such that $\emptyset\ne T\subseteq S\subseteq X'$, we have that $\Gamma(S)=T$ if and only if the variable $\attention{S}{T}$ is $\mathsf{True}$ in that model. 
We define the set $\formulae{(X',\dataset)}$ to consist of the following formulae:
\begin{enumerate}
\item
for all distinct $(S,a)\in\dataset$, all $b\in S\setminus\{a\}$,
and all $T\subseteq S$ s.t.\ $b\in T$, the formula `$\attention{S}{T}\rightarrow\gt{a}{b}$',
requiring that the preferences and attention filter rationalize $\dataset$;
\item
for all distinct $a,b\in X'$, the formula `$\gt{a}{b}\vee\gt{b}{a}$',
requiring that the preferences be complete;
\item
for all distinct $a,b\in X'$, the formula `$\lnot(\gt{a}{b}\wedge\gt{b}{a})$', requiring that the preferences be antisymmetric;
\item
for all distinct $a,b,c\in X'$, the formula `$(\gt{a}{b}\wedge\gt{b}{c})\rightarrow\gt{a}{c}$', requiring that the preferences be transitive;
\item
for all menus $S\subseteq X'$, the (finite!) formula `$\vee_{\emptyset\ne T\subseteq S}
\attention{S}{T}$', requiring that $S$ have a set of attention-attracting elements that is a nonempty subset of $S$;
\item
for all menus $S\subseteq X'$ and all distinct $T,T'\in 2^S\setminus\{\emptyset\}$, the formula `$\attention{S}{T}\rightarrow\lnot\attention{S}{T'}$', requiring that the set of attention-attracting elements from $S$ be unique;
\item
for all menus $S\subseteq X'$, all $T\in 2^S\setminus\{\emptyset\}$, and all $x\in S\setminus T$, the formula `$\attention{S}{T}\rightarrow\attention{S\setminus\{x\}}{T}$', requiring that if $x$ is not attention attracting from $S$, then $S$ an $S\setminus\{x\}$ have the same attention-attracting set.
\end{enumerate}
By construction, $(\formulae{(X',\dataset}))_{(X',\dataset)\in\mathcal{P}}$ is a well description of~$\mathcal{P}$.

\proofstep{Finite-subset property}
Let $(\bar{X},\dataset)\in\mathcal{P}$. Let $\Phi'\subset\formulae{(\bar{X},\dataset)}$ be a finite subset.
Since $\Phi'$ is finite, it ``mentions'' only finitely many elements of~$\bar{X}$ (through variables used, whether by mentioning these elements directly or by mentioning menus that contain them); denote the set of these elements by $X'\subset\bar{X}$.
Let $\dataset'\eqdef\bigl\{(S,a)\in D~\big|~ S\subseteq X'\bigr\}$. By definition, $\Phi'\subseteq\formulae{(X',\dataset')}$. Furthermore, $\dataset'$ is a full dataset and it satisfies WARP-LA since any sub-dataset of $\dataset$ satisfies WARP-LA. Hence, by \cref{FiniteLA}, $\dataset'$~is rationalizable by some strict preference order over $X'$ and attention filter. Therefore, $(\bar{X},\dataset) $~satisfies the finite-subset property. Thus, by \cref{scaling-lemma}, $D$~is rationalizable by a strict preference order over $\bar{X}$ and an attention filter.
\end{proof}

\subsection{Proof of Lemma~\ref{marginals}}\label{app:marginals}

\begin{proof}[Proof of \cref{marginals}]
The first condition immediately implies the second condition. We therefore assume the second condition and prove the first condition from it.
Let $T=\binom{X}{2}$, where we denote an element of $T$ as $(i_1,j_1)$ with the choice of the ordering of $i_1$ and $j_1$ consistent throughout this proof for each pair. For every distinct $(i_1,j_1),\ldots,(i_k,j_k)\in T$, let $\nu_{(i_1,j_1),\ldots,(i_k,j_k)}$ be the probability measure over $\RR^k$ defined as follows: Let $E=\{i_1,j_1,i_2,j_2,\ldots,i_k,j_k\}$ and let $n=|E|\le 2k$. We define a probability measure $\mu_E$ over the finite set $E!$ of all permutations of $E$ by assigning probability $p_{(e_1,\ldots,e_m)}$ to the ordering $(e_1,\ldots,e_m)$. By the second condition of the \lcnamecref{marginals}, this is indeed a probability measure. We define $\nu_{(i_1,j_1),\ldots,(i_k,j_k)}$ as follows: to draw $r_1,\ldots,r_k\in\RR^k$ according to $\nu_{(i_1,j_1),\ldots,(i_k,j_k)}$, first draw a permutation in $E!$ according to $\mu_E$, and then for every $\ell\in\{1,\ldots,k\}$, set $r_\ell=1$ if $i_\ell$ precedes $j_\ell$ according to this permutation, and $r_\ell=0$ otherwise. Notice that for every measurable $F_1,\ldots,F_k\subseteq\RR$, both of the following hold:
\begin{itemize}
\item
For every permutation $\pi$ of $\{1,\ldots,k\}$, we have, immediately by definition, that $\nu_{\pi((i_1,j_1)),\ldots,\pi((i_k,j_k))}(F_{\pi(1)},\ldots,F_{\pi(k)})=\nu_{(i_1,j_1),\ldots,(i_k,j_k)}(F_1,\ldots,F_k)$.
\item
\begin{sloppypar}
For every $n\in\NN$ and $(i_{k+1},j_{k+1}),\ldots,(i_{k+n},j_{k+n})\in T$ distinct from one another and from $(i_1,j_1),\ldots,(i_k,j_k)$, it is the case that $\nu_{(i_1,j_1),\ldots,(i_n,j_n)}(F_1,\ldots,F_k,\RR,\ldots,\RR)=\nu_{(i_1,j_1),\ldots,\ldots,(i_k,j_k)}(F_1,\ldots,F_k)$, by the second condition of the \lcnamecref{marginals}.
\end{sloppypar}
\end{itemize}
By these two conditions and by the Kolmogorov Extension Theorem, there exists a probability measure $\nu$ over $\RR^T$ with the product $\sigma$-algebra whose marginals are the above-defined $\nu_{(\cdot)}$ measures.

To define the required probability measure $\mu$, consider the following embedding into $\RR^T$ of the space of total orders $\pi$ over $X$: map a total order $\pi$ over $X$ to $(r_t)_{t\in T}$, where for every $(i,j)\in T$  we set $r_{(i,j)}=1$ if $i$ precedes $j$ according to $\pi$, and $r_{(i,j)}=0$ otherwise. We note that this embedding is an isomorphism of measurable spaces (i.e., a measurable bijection whose inverse is also measurable) of its domain (w.r.t.\ the $\sigma$-algebra generated by all of its subsets of the form $\{\pi \mid a_1\succ_{\pi}\cdots\succ_{\pi} a_m\}$ where the $a_i$s are distinct elements in $X$) and its image (w.r.t.\ the product $\sigma$-algebra), and therefore via this embedding the measure $\nu$ induces a measure $\mu$ over the space of total orders over ($X$ w.r.t.\ to the above-defined $\sigma$-algebra). We note that the complement of the image of this embedding, inside $\RR^T$, has measure $0$ w.r.t.\ $\nu$ by the above construction of the marginals of $\nu$, and so $\mu$ is a probability measure. By the definition of the marginals of $\nu$ via the $\mu_E$ measures defined above, the probability measure $\mu$ satisfies the first condition of the \lcnamecref{marginals}, as required.
\end{proof}

\section{Proof Omitted from Section~\ref{matching}}

\subsection{Stable Matchings with Couples}\label{app:couples}

\begin{proof}[Proof of \cref{couples-infinite}]
We prove the theorem using \cref{scaling-lemma}.

\proofstep{Definition of $\mathcal{P}$}
Let $\couplesmarket$ be a matching-with-couples market. Assume without loss of generality that the preference lists of couples are such that no couple ranks any $(h,h')\in H\times H$ below $(h,\emptyset)$ or $(\emptyset,h')$.
Let $\mathcal{P}\eqdef\bigl\{(D',H')~\big|~D'\subseteq D \And H'\subseteq H \And \forall (d_1,d_2)\in D^2: (d_1\in D' \And d_2\in D') \text{~or~} (d_1\notin D' \And d_2\notin D')\bigr\}$ be the set of all pairs of subsets of doctors that do not partially intersect any couple and subsets of hospitals. A \emph{solution} for $(D',H')\in\mathcal{P}$ is a perturbed-feasible stable matching between $D'$ and $H'$, i.e., a matching between $D'$ and $H'$ that is stable (in the induced submarket of $\couplesmarket$) with respect to a capacity vector $k^*$ with $|k_h-k^*_h|\le2$ for every $h\in H'$.

\proofstep{Well describability}
We define a variable $\matched{d}{h}$ for every $(d,h)\in D\times H$, a variable $\matched{d}{\emptyset}$ for every doctor $d\in D$ from a couple, and a variable $\quota{h}{q}$ for every hospital $h\in H$ and capacity $q\in\{k_h-2,\ldots,k_h+2\}$.
In what follows, for each $(D',H')\in\mathcal{P}$ we define a set $\formulae{(D',H')}$ of formulae over these variables so that models of $\formulae{(D',H')}$ are in one-to-one correspondence with the (not-yet-proven-to-be-nonempty) set of perturbed-feasible stable matchings between $D'$ and $H'$.
The correspondence is obtained by endowing the variable $\matched{d}{h}$ with the semantic interpretation ``$d$ and $h$ are matched,'' the variable $\matched{d}{\emptyset}$ with the semantic interpretation ``$d$ is unmatched while $d$'s partner is matched,''
and the variable $\quota{h}{q}$ with the semantic interpretation ``the perturbed capacity of $h$ is $q$.'' That is, it maps a model for $\formulae{(D',H')}$ to the matching such that for every $(d,h)\in D'\times H'$, we have that $d$ and $h$ are matched if and only if the variable $\matched{d}{h}$ is $\mathsf{True}$ in that model, which would be stable with respect to the capacity vector $k^*$ where $k^*_h=q$ if and only if the variable $\quota{h}{q}$ is $\mathsf{True}$ in that model. (The variable $\matched{d}{\emptyset}$ is a convenience variable whose value can be inferred from that of the $\matched{d}{h}$ variables. Its purpose is to enhance the readability of our formulae.)
We define the set $\formulae{(D',H')}$ to consist of the following formulae:
\begin{itemize}
\item
Allowed perturbed capacities:
\begin{enumerate}
\item
for all $h\in H'$, the formula `$\bigvee_{q=k_h-2}^{k_h+2}\quota{h}{q}$',
requiring that $h$ have a perturbed capacity within the range $k_h\pm2$;
\item
for all $h\in H'$ and all distinct $q,q'$ in $\{k_h-2,\ldots,k_h+2\}$, the formula
`$\quota{h}{q}\rightarrow\lnot\quota{h}{q'}$',
requiring that $h$ have no more than one perturbed capacity;
\end{enumerate}
\item
Matching respecting perturbed capacities:
\begin{enumerate}[resume]
\item
for all $d\in D'$ and all distinct $h,h'\in H'$, the formula `$\matched{d}{h}\rightarrow\lnot\matched{d}{h'}$',
requiring that $d$ be matched to at most one hospital;
\item
for all $h\in H'$, all $q\in\{k_h-2,\ldots,k_h+2\}$, and all distinct $d,d_1,d_2,\ldots,d_q\in D'$, the formula
`$\bigl(\quota{h}{q}\wedge\bigwedge_{i=1}^q\matched{d_i}{h}\bigr)\rightarrow\lnot\matched{d}{h}$',
requiring that $h$ not be matched to more doctors than its perturbed capacity;
\end{enumerate}
\item
Individual rationality:
\begin{enumerate}[resume]
\item
for all $d\in D'$ and all $h\in H'$ such that one or more of the following holds:
\begin{enumerate}
\item $h$ does not rank $d$,
\item $d$ is not in any couple and does not rank $h$, or
\item $d$ is in a couple, and no pair of hospitals in this couple's preference list has $h$ matched to~$d$,
\end{enumerate}
the formula
`$\lnot\matched{d}{h}$',
requiring that no doctor or hospital is matched in a way that they individually find unacceptable;
\item
for all couples $c$ with members in $D'$ and all distinct $h,h'\in H'$ such that $c$ does not rank $(h,h')$, the formula
`$\lnot(\matched{c_1}{h}\wedge\matched{c_2}{h'})$',
requiring that $c$ not be matched to a pair of hospitals that they find unacceptable;
\stepcounter{enumi}
\item[\theenumi{}a.]
\begin{sloppypar}
for all couples $c$ with members in $D'$ and all $h\in H'$ such that $c$ ranks $(h,\emptyset)$, letting $h_1,\ldots,h_n$ be the hospitals in $H'$ such that $c$ ranks every $(h,h_i)$, the formula
`$\matched{c_1}{h}\rightarrow\bigl(\matched{c_2}{\emptyset}\leftrightarrow\lnot\bigvee_{i=1}^n\matched{c_2}{h_i}\bigr)$',
effectively setting (for convenience), when $c_1$ is matched with $h$, $\matched{c_2}{\emptyset}$ as a shorthand for $c_2$ not being matched to any of these $h_i$s;
\end{sloppypar}
\color{gray}
\item[\theenumi{}b.]
Completely symmetrically, for all couples $c$ with members in $D'$ and all $h\in H'$ such that $c$ ranks $(\emptyset,h)$, letting $h_1,\ldots,h_n$ be the hospitals in $H'$ such that $c$ ranks every $(h_i,h)$, the formula
`$\matched{c_2}{h}\rightarrow\bigl(\matched{c_1}{\emptyset}\leftrightarrow\lnot\bigvee_{i=1}^n\matched{c_1}{h_i}\bigr)$',
effectively setting (for convenience), when $c_2$ is matched with $h$, $\matched{c_1}{\emptyset}$ as a shorthand for $c_1$ not being matched to any of these $h_i$s;
\color{black}
\end{enumerate}
\item
Not blocked with respect to perturbed capacities:
\begin{enumerate}[resume]
\item
for all single doctors in $D'$ and all $h\in H'$ such that each ranks the other, and for all $q\in\{k_h-2,\ldots,k_h+2\}$, letting $h_1,\ldots,h_l$ be all the hospitals in $H'$ that $d$ prefers to $h$ and letting $D_1,\ldots,D_n$ be all the $q$-tuples of doctors in $D'$ that $h$ ranks above $d$, the (finite!) formula
\begin{multline*}
(\quota{h}{q}\wedge\lnot\matched{d}{h})\rightarrow\\
\rightarrow\left(\bigvee_{i=1}^l\matched{d}{h_i}\vee\bigvee_{i=1}^n\bigwedge_{d'\in D_i}\matched{d'}{h}\right),
\end{multline*}
requiring that $(d,h)$ is not a blocking pair;
\item
for all couples $c$ with members in $D'$ and all distinct $h,h'\in H'$ such that $c$ ranks $(h,h')$ and such that $h$ ranks $c_1$ and $h'$ ranks $c_2$, and for all $q\in\{k_h-2,\ldots,k_h+2\}$ and $q'\in\{k_{h'}-2,\ldots,k_{h'}+2\}$, letting $(h_1,h'_1)\ldots,(h_l,h'_l)$ be all the assignments to pairs of hospitals in $H'$ that $c$ prefers to $(h,h')$, letting $D_1,\ldots,D_n$ be all the $q$-tuples of doctors in $D'$ that $h$ ranks above $c_1$, and letting $D'_1,\ldots,D'_{n'}$ be all the $q'$-tuples of doctors in $D'$ that $h'$ ranks above $c_2$, the (finite) formula
\begin{multline*}
(\quota{h}{q}\wedge\quota{h'}{q'}\wedge\lnot(\matched{c_1}{h}\wedge\matched{c_2}{h'}))\rightarrow\\*\rightarrow\Biggl(\bigvee_{i=1}^l(\matched{c_1}{h_i}\wedge\matched{c_2}{h'_i})\vee\\*\vee\bigvee_{i=1}^n\bigwedge_{d\in D_i}\matched{d}{h}\vee\bigvee_{i=1}^{n'}\bigwedge_{d'\in D'_i}\matched{d'}{h'}\Biggr),
\end{multline*}
requiring that $c$ not block with $(h,h')$;
\stepcounter{enumi}
\item[\theenumi{}a.]
\begin{sloppypar}
for all couples $c$ with members in $D'$ and all $h\in H'$ such that $c$ ranks $(h,\emptyset)$ and such that $h$ ranks $c_1$, and for all $q\in\{k_h-2,\ldots,k_h+2\}$, letting $(h_1,h'_1)\ldots,(h_l,h'_l)$ be all the assignments to pairs of hospitals in $H'$ that $c$ prefers to $(h,\emptyset)$ and letting $D_1,\ldots,D_n$ be all the $q$-tuples of doctors in $D'$ that $h$ ranks above $c_1$, the (finite) formula
\begin{multline*}
(\quota{h}{q}\wedge\lnot(\matched{c_1}{h}\wedge\matched{c_2}{\emptyset}))\rightarrow\\*\rightarrow\left(\bigvee_{i=1}^l(\matched{c_1}{h_i}\wedge\matched{c_2}{h'_i})\vee\bigvee_{i=1}^n\bigwedge_{d\in D_i}\matched{d}{h}\right),
\end{multline*}
requiring that $c$ not block with $(h,\emptyset)$;
\end{sloppypar}
\color{gray}
\item[\theenumi{}b.]
Completely symmetrically, for all couples $c$ with members in $D'$ and all $h\in H'$ such that $c$ ranks $(\emptyset,h)$ and such that $h$ ranks $c_2$, and for all $q\in\{k_h-2,\ldots,k_h+2\}$, letting $(h_1,h'_1)\ldots,(h_l,h'_l)$ be all the assignments to pairs of hospitals in $H'$ that $c$ prefers to $(\emptyset,h)$ and letting $D_1,\ldots,D_n$ be all the $q$-tuples of doctors in $D'$ that $h$ ranks above $c_2$, the (finite) formula
\begin{multline*}
(\quota{h}{q}\wedge\lnot(\matched{c_1}{\emptyset}\wedge\matched{c_2}{h}))\rightarrow\\*\rightarrow\left(\bigvee_{i=1}^l(\matched{c_1}{h_i}\wedge\matched{c_2}{h'_i})\vee\bigvee_{i=1}^n\bigwedge_{d\in D_i}\matched{d}{h}\right),
\end{multline*}
requiring that $c$ not block with $(\emptyset,h)$;
\color{black}
\item
for all couples $c$ with members in $D'$ and all $h\in H'$ such that $c$ ranks $(h,h)$ and such that $h$ ranks both $c_1$ and $c_2$, and for all $q\in\{k_h-2,\ldots,k_h+2\}$, letting $(h_1,h'_1)\ldots,(h_l,h'_l)$ be all the assignments to pairs of hospitals in $H'$ that $c$ prefers to $(h,h)$ and letting $D_1,\ldots,D_n$ be all the $(q\!-\!1)$-tuples of doctors in $D'$ that include neither $c_1$ nor $c_2$ and that $h$ ranks above one of $c_1$ or $c_2$, the (finite) formula
\begin{multline*}
(\quota{h}{q}\wedge\lnot(\matched{c_1}{h}\wedge\matched{c_2}{h}))\rightarrow\\*\rightarrow\left(\bigvee_{i=1}^l(\matched{c_1}{h_i}\wedge\matched{c_2}{h'_i})\vee\bigvee_{i=1}^n\bigwedge_{d\in D_i}\matched{d}{h}\right),
\end{multline*}
requiring that $c$ not block with (matching both doctors in $c$ to) $h$.
\end{enumerate}
\end{itemize}

By construction, the models of $\formulae{(D',H')}$ are in one-to-one correspondence with perturbed-feasible stable matchings between $D'$ and $H'$. (The only subtle part is noting that if for some couple $c$ and for some hospital $h$ the couple $c$ does not rank $(h,\emptyset)$ then since the preferences of $c$ are downward closed, this means that no $(h,h')$ are ranked by $c$ and so we have added the formula $\lnot\matched{c_1}{h}$ so indeed it is impossible that $c$ be matched with $(h,\emptyset)$, and symmetrically for when $c$ does not rank $(\emptyset,h)$.) 	Therefore, $(\formulae{(D',H')})_{(D',H')\in\mathcal{P}}$ is a well description of~$\mathcal{P}$.

\proofstep{Finite-subset property}
Let $\Phi'\subset\formulae{(D,H)}$ be a finite subset.
Since $\Phi'$ is finite, it ``mentions'' (through variables used) only finitely many doctors and hospitals; denote the set of these doctors and all the doctors couples with any of them by $D'\subset D$ and the set of these hospitals by $H'\subset H$.
By definition, $\Phi'\subseteq\formulae{(D',H')}$. By \cref{couples-finite}, there exists a perturbed-feasible stable matching between $D'$ and $H'$. Therefore, $(D,H)$~satisfies the finite-subset property. Thus, by \cref{scaling-lemma}, there exists a perturbed-feasible stable matching between $d$ and~$H$, as required.
\end{proof}

\renewcommand\refname{References for Appendices}
\small
\singlespacing
\addcontentsline{toc}{section}{References for Appendices}
\putbib
\end{bibunit}

\end{document}